\documentclass[showpacs,twocolumn,aps,prx,longbibliography,superscriptaddress,notitlepage]{revtex4-2}
\usepackage{qcircuit}
\usepackage[dvips]{graphicx}
\usepackage[normalem]{ulem}
\usepackage{amsmath,amssymb,amsthm,mathrsfs,amsfonts,dsfont}
\usepackage{subfigure, epsfig}
\usepackage{braket}
\usepackage{bm}
\usepackage{enumerate}
\usepackage{algorithm}
\usepackage{algpseudocode}
\usepackage{diagbox}
\usepackage{physics}
\usepackage{color}
\usepackage{multirow}
\usepackage[marginal]{footmisc}
\usepackage{comment}
\usepackage{tikz}
\usepackage{gensymb}
\usepackage[]{qcircuit}
\usetikzlibrary{arrows}
\usetikzlibrary{shapes,fadings,snakes}
\usetikzlibrary{decorations.pathmorphing,patterns}
\usetikzlibrary{calc}
\usetikzlibrary{positioning}
\usepackage[colorlinks = true]{hyperref}

\newtheorem{theorem}{Theorem}

\newtheorem{corollary}{Corollary}
\newtheorem{proposition}{Proposition}
\newcommand{\comments}[1]{}

\begin{document}

\title{Predicting Arbitrary State Properties from Single Hamiltonian Quench Dynamics}

\begin{abstract}
Analog quantum simulation is an essential routine for quantum computing and plays a crucial role in studying quantum many-body physics.
Typically, the quantum evolution of an analog simulator is largely determined by its physical characteristics, lacking the precise control or versatility of quantum gates. This limitation poses challenges in extracting physical properties on analog quantum simulators, an essential step of quantum simulations.
To address this issue, we introduce the \emph{Hamiltonian shadow} protocol, which uses a single quench Hamiltonian for estimating arbitrary state properties, eliminating the need for ancillary systems and random unitaries.
Additionally, we derive the sample complexity of this protocol and show that it performs comparably to the classical shadow protocol. The Hamiltonian shadow protocol does not require sophisticated control and can be applied to a wide range of analog quantum simulators. We demonstrate its utility through numerical demonstrations with Rydberg atom arrays under realistic parameter settings. The new protocol significantly broadens the application of randomized measurements for analog quantum simulators without precise control and ancillary systems.

\end{abstract}
\date{\today}% It is always \today, today,

\author{Zhenhuan Liu}
\email{liu-zh20@mails.tsinghua.edu.cn}
\affiliation{Center for Quantum Information, Institute for Interdisciplinary Information Sciences, Tsinghua University, Beijing 100084, China}

\author{Zihan Hao}
\email{haozh20@mails.tsinghua.edu.cn}
\affiliation{Center for Quantum Information, Institute for Interdisciplinary Information Sciences, Tsinghua University, Beijing 100084, China}

\author{Hong-Ye Hu}
\email{hongyehu@fas.harvard.edu}
\affiliation{Department of Physics, Harvard University, 17 Oxford Street, Cambridge, MA 02138, USA}

\maketitle

\textbf{Introduction.} Analog quantum simulation stands as one of the flagship applications of emerging quantum technology \cite{bernien2017dynamics,ebadi2021quantumphase,daley2022practicaladvantage,evered2023highfidelity}.
Up to date, various platforms for analog quantum simulation have been suggested, including optical tweezer with neutral atoms \cite{ebadi2021quantumphase,bernien2017dynamics} and molecules \cite{hu2019molecule,burchesky2021rotation}, optical lattices with quantum gas microscopes \cite{Bloch,thomas2023fermionpairing,lenard2023FQH} and solid-state materials \cite{Kiczynski2022quantumdot,wang2022quantumdot2}.
It provides insights into the realization of exotic phases of matter \cite{aidelsburger2015chernnumber,smits2018timecrystal,semeghini2021spinliquid,randall2021MBL,leonard2023FQH}, the investigation of superconductivity \cite{hart2015AFM,chiu2019strings,hartke2020hubbard,ji2021mobilehole}, and the simulation of quantum field theory \cite{martinez2016gaugefield,yang2020gauge,paulson20212Dgauge,meth2023simulating}, among others.

Analog simulators are specially designed to execute specific quantum evolutions dictated by their intrinsic Hamiltonians. Although analog quantum simulators excel at simulating certain complex quantum many-body systems, extracting physical properties from such platforms, such as correlation functions, quantum fidelity, and entanglement entropy, is not straightforward. The root of this challenge comes from the fact that the measurement bases of analog quantum simulators are usually limited to one or a few. For example, it is easy to perform computational basis measurements in Rydberg atoms arrays, or particle number basis in quantum gas microscopes. To extract more sophisticated physical properties, such as fidelity or correlation function, one needs to perform a unitary basis transformation to rotate such a desired physical observable to those simple measurement bases. However, such basis rotation is challenging because analog quantum simulators are restricted by their fixed Hamiltonian forms and have no universal quantum gates. For example, if one only has global laser pulse control of the Rydberg atom arrays, it is impossible to rotate arbitrary Pauli observables to computational measurement basis, making it hard to perform even quantum tomography.

\begin{figure}[htbp]
    \centering
    \includegraphics[width=0.95\linewidth]{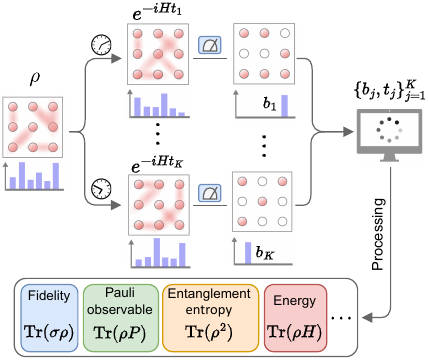}
    \caption{Overview of the Hamiltonian shadow protocol. In each experiment round, the target state $\rho$ evolves with a fixed Hamiltonian and a random time, followed by a computational basis measurement. The data pair of measurement result $b_j$ and evolution time $t_j$ is recorded on a classical computer to estimate arbitrary properties of $\rho$.}
    \label{fig:overview}
\end{figure}

Inspired by randomized measurement protocols which estimate observables via random bases measurements \cite{Huang2020predicting,Elben2023toolbox,Elben2019toolbox,elben2018quench,brydges2019entropy,chen2021robust,denzler2023learning,van2022hardware,hu2023locallybias,Akhtar2023scalableflexible,bertoni2022shallow}, researchers have made many efforts trying to solve this challenge. The key ideas are either using many chaotic Hamiltonians \cite{hu2022hamiltonian} or using ancillary systems \cite{tran2023arbitrary,mcginley2023shadow} to induce random bases measurements and extract information. However, it is neither experimentally favorable to calibrate many chaotic Hamiltonians nor enlarge the system size. It would be desirable to use a single Hamiltonian without ancillary systems to effectively measure many or arbitrary physical properties.

In this work, we address this challenge by presenting the 
\emph{Hamiltonian shadow} protocol. This protocol allows for predicting any state property through quench evolution with a single Hamiltonian and different evolution times, as illustrated in Fig.~\ref{fig:overview}. 
It essentially leverages the inherent randomness in eigen-energies of a single Hamiltonian, eliminating the need for multiple Hamiltonians or ancillary systems, which is a significant reduction in resources. 
We theoretically prove that computational basis measurements after quantum evolution under a single Hamiltonian can unbiasedly extract arbitrary properties of the target quantum state. 
Furthermore, the requirements for the quench Hamiltonian are minimal, such as no eigen-energy degeneracy and no computational basis eigenstates, conditions generally met by practical analog systems. 
Even with a single quench Hamiltonian, both theoretical and numerical analysis demonstrate that the performance of the Hamiltonian shadow is comparable to that of the original shadow with random Clifford unitaries. We apply the proposed method to a Rydberg atoms system with realistic parameter settings and demonstrate its efficacy in estimating several essential physical quantities with only global laser pulses, including quantum fidelity, local correlation function, and purity. Conventionally, these quantities are hard to measure since no basis rotation with global control can transform them to computational measurement basis. This advancement signifies a substantial leap forward in analog quantum simulation and classical shadow tomography, holding considerable promise for near-term applications. 

\begin{figure}
\centering
\includegraphics[width=0.4\linewidth]{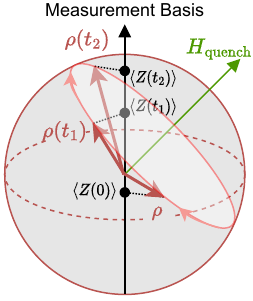}
\caption{Geometric intuition of Hamiltonian shadow in a single-qubit Bloch sphere. When $\rho$ is rotating around the Hamiltonian by a known velocity, we are able to recover the initial position of $\rho$ by its $z$ position at different times together with the information of the quench Hamiltonian.} 
\label{fig:dev_overview}
\end{figure}

\textbf{Hamiltonian shadow.} 
We use a toy model to geometrically illustrate the validation of state learning using a single Hamiltonian. 
As shown in Fig.~\ref{fig:dev_overview}, an unknown state $\rho$ is a vector in the Bloch sphere, rotating around a known axis with a known angular velocity determined by the quench Hamiltonian $H_{\text{quench}}$. 
Based on geometric intuitions, the dynamical trajectory of the $Z$-basis expectation value $\langle Z(t)\rangle$ together with $H_{\text{quench}}$ can recover $\rho$ \cite{special_cases}, which stands for the initial position of the rotating vector. 
Thus, the fundamental reason making the state learning with a single Hamiltonian possible is that the measurement basis does not align with the Hamiltonian. 
We can further ask whether it is possible to fully recover the unknown state $\rho$ by taking single-shot $Z$-basis measurements at different quench times $t$.

Operationally speaking, after quench evolution with $e^{-iHt}$, we take a single-shot measurement, and the state collapses to $\ket{b}$.
Then, the classical representation, $\hat{\sigma} = e^{iHt}\ketbra{b}{b}e^{-iHt}$, contains information of $\rho$.
As quantum mechanics is fundamentally linear, the average of the dataset $\{\hat{\sigma}_j = e^{iHt_j}\ketbra{b_j}{b_j}e^{-iHt_j}\}_{j}$ is related to $\rho$ through a linear map
\begin{equation}
\begin{aligned}
\mathcal{M}_{H}(\rho)=\mathbb{E}_{t,b}\left[e^{iHt}\ketbra{b}{b}e^{-iHt}\right].
\end{aligned}
\end{equation} 
If the map $\mathcal{M}_H$ is invertible, $\hat{\rho}=\mathcal{M}_H^{-1}(e^{iHt}\ketbra{b}{b}e^{-iHt})$ becomes an unbiased estimator of $\rho$. Then, one can use the classical dataset to predict arbitrary state properties, which means that the state learning with a single Hamiltonian is tomography-complete. Specifically, given $\{\hat{\rho}_j\}_{j=1}^K$, one can construct estimators such as $\frac{1}{K}\sum_j\Tr(O\hat{\rho}_j)$ and $\frac{1}{K(K-1)}\sum_{j\neq k}\Tr[O(\hat{\rho}_j\otimes\hat{\rho}_k)]$ to unbiasedly predict linear and nonlinear properties.
This logic aligns with the classical shadow tomography \cite{Huang2020predicting}. We thus name this measurement scheme as Hamiltonian shadow and the map $\mathcal{M}_H(\cdot)$ as Hamiltonian shadow map.

\comments{
The Hamiltonian shadow map is uniquely determined by its Choi matrix.
Substituting Born's rule and the spectral decomposition $e^{-iHt}=V_H \Lambda_t V_H^{\dagger}$ with $(\Lambda_t)_{k,l}=e^{-iE_kt}\delta_{k,l}$, we can rewrite the Hamiltonian shadow map as
\begin{equation}
\begin{aligned}
\mathcal{M}_{H}(\rho)=V_H \mathcal{N}\left( V_H^{\dagger} \rho V_{H}\right) V_H^{\dagger}.
\end{aligned}
\end{equation}
From tensor network representation shown in Fig.~\ref{fig:dev_overview}(b), the Choi matrix of $\mathcal{N}$ is $\mathbb{E}_{t}(\Lambda_t ^{\otimes 2}X_2 \bar{\Lambda}_t^{\otimes 2})$
with $X_2=\sum_bV_H^{\dagger\otimes 2}\ketbra{bb}{bb}V_H^{\otimes 2}$.
In general, $\mathcal{N}$ has a complicated form and is hard to invert.
However, if eigen-energies of $H$ and evolution time $t$ contain enough randomness to make $\Lambda_t$ a random diagonal unitary, the Choi matrix reduces to $\Phi_2^{\mathrm{D}}(X_2)=\sum_{i,j}(X_2)_{ij,ij}\ket{ij}\bra{ij}+\sum_{i\neq j}(X_2)_{ij,ji}\ket{ij}\bra{ji}$, where $\Phi_2^{\mathrm{D}}(\cdot)$ denotes the second-order integral over random diagonal unitaries.
Most elements cancel out due to the integral, and the left ones can be summarized using a much smaller matrix, the post-processing matrix $X_H$.}

Normally speaking, inverting a linear map is a tough task.
While, we realize that when taking the spectral decomposition, $e^{-iHt}=V_H\Lambda_t V_H^\dagger$, one gets a random diagonal unitary $\Lambda_t$  when $t$ can be randomly selected \cite{RDU}.
Therefore, assisted with the theory of random diagonal unitaries \cite{nakata2013diagonal}, we can largely reduce the complexity of inverting the Hamiltonian shadow map and arrive at:

\begin{theorem}\label{theorem:unbiased_estimator}
Let the target state $\rho$ evolve using $e^{-iHt}$ with a fixed Hamiltonian $H$ and random $t$, then measured in the computational basis, and the measurement result is $\ket{b}$. 
Suppose the spectral decomposition of $e^{-iHt}$ is $V_H\Lambda_t V_H^\dagger$, where the ensemble $\{\Lambda_t\}_t$ is a random diagonal unitary ensemble. Defining $V_H^{\mathrm{sq}}=\sum_{i,j}\abs{(V_H)_{i,j}}^2\ketbra{i}{j}$ with $\ket{i}$ and $\ket{j}$ being computational basis vectors, if matrix $X_H=(V_H^{\mathrm{sq}})^TV_H^{\mathrm{sq}}$ is invertible and all off-diagonal elements are nonzero, the following expression is the unbiased estimator of $\rho$
\begin{equation}\label{eq:estimator_theorem1}
\hat{\rho}=V_H\mathcal{N}^{-1}\left(V_H^{\dagger}\hat{\sigma}V_H\right)V_H^\dagger,
\end{equation}
where $\hat{\sigma}=e^{iHt}\ketbra{b}{b}e^{-iHt}$ and the action of map $\mathcal{N}^{-1}$ is
\begin{equation}\label{eq:N_inverse}
\mathcal{N}^{-1}(\sigma)=\sum_{i,j}(X_H^{-1})_{ij}\sigma_{j,j}\ketbra{i}{i}+\sum_{i\neq j}(X_H)_{i,j}^{-1}\sigma_{i,j}\ketbra{i}{j}.
\end{equation}
\end{theorem}

In the literature, randomized measurements using Hamiltonian evolution either employ multiple chaotic Hamiltonians \cite{elben2018quench,hu2022hamiltonian} or large ancillary systems \cite{tran2023arbitrary,mcginley2023shadow} to induce randomness. 
Conversely, we capitalize on the inherent randomness in eigen-energies of a single Hamiltonian $H$.
This is more practical as implementing the coherent dynamics $e^{-iHt}$ governed by an intrinsic Hamiltonian $H$ is readily achievable for many analog quantum simulators.
\comments{Note that directly treating $e^{-iHt}$ as a random Haar unitary and performing the data post-processing of the original global classical shadow will lead to biased estimations.
For example, in Fig.~\ref{fig:dev_overview}(c), we show the fidelity estimation with two data post-processing methods, Hamiltonian shadow and global shadow, using the same measurement dataset collected with a single Hamiltonian quench evolution. It is clear that the estimation of Hamiltonian shadow approaches the real value while the global shadow does not.}

When dealing with local observables, such as local correlation functions, it becomes advantageous to employ a local version of the Hamiltonian shadow. Assume that the entire system can be segmented into multiple localized patches, with techniques such as atoms reconfiguration in Rydberg atom arrays. 
Then the whole system evolution is $\bigotimes_{p=1}^N e^{-iH_pt}$ with independent $H_p$. The unbiased estimator of $\rho$ can be constructed as
\begin{equation}
\hat{\rho}=\bigotimes_{p=1}^N V_{H_p}\mathcal{N}^{-1}_p(V_{H_p}^\dagger\hat{\sigma}_p V_{H_p})V_{H_p}^\dagger,
\end{equation}
where $\hat{\sigma}_p=e^{iH_pt}\ketbra{b_p}{b_p}e^{-iH_pt}$ and $\mathcal{N}_p^{-1}$ is defined in the same way using $X_{H_p}$. The local Hamiltonian shadow reduces classical computational resources by only dealing with local-patch Hamiltonians in the data post-processing. Besides, the experiment sample complexity is independent of the total system size when measuring local observables.

\textbf{Performance guarantee.} The performance of the Hamiltonian shadow protocol largely depends on the Hamiltonian $H$, including its eigenstates and eigen-energies. According to Eq.~\eqref{eq:N_inverse}, the Hamiltonian shadow is tomography-complete if the post-processing matrix $X_H$, which is determined solely by eigenstates, is invertible, and all off-diagonal elements are non-zero. One common situation where the Hamiltonian shadow becomes tomography-incomplete is when some eigenstates align with the measurement basis, making $X_H$ a block-diagonal matrix. Yet, most Hamiltonians without fine-tuning in general satisfy the requirements for $X_H$ \cite{tomography_incomplete}.
The requirement for eigen-energies originates from the assumption that $\Lambda_t$ approximates random diagonal unitary. We show that it can be summarized with the non-resonance condition \cite{Reimann2008resonance,degeneracy}, in which no eigen-energy pairs satisfy $E_{a_1}+E_{a_2}=E_{b_1}+E_{b_2}$ for $(a_1,a_2)\neq(b_1,b_2)$ and $(a_1,a_2)\neq(b_2,b_1)$. 
Similarly, this condition can generally be satisfied by a quench dynamics Hamiltonian without certain global symmetries \cite{Ueda2020ETH}. Considering the practical scenario where the evolution time is limited, one may further require eigen-energies of $H$ to have large level spacing.

When requirements for eigenstates and eigen-energies are all satisfied, one can efficiently predict many physical properties of the state simultaneously.

\begin{theorem}
To estimate expectation values of $M$ arbitrary observables $\{O_i\}_{i=1}^M$ to $\epsilon$ accuracy with the Hamiltonian shadow protocol, the sample complexity is upper bounded by $K=\mathcal{O}\left(\max_i\norm{O_i}_{\mathrm{HShadow}}^2\log(M)/\epsilon^2\right)$, where $\norm{O}_{\mathrm{HShadow}}$ is the Hamiltonian shadow norm determined by $V_H$ and $O$. 
\end{theorem}

Details of the Hamiltonian shadow norm can be found in Appendix \ref{sec:variance}. Despite its complicated form, which obscures the connection between sample complexities and structures of Hamiltonians, we find that
\begin{equation}
f(O,V_H)=\sum_{i\neq j}\frac{1}{(X_H)_{i,j}}\abs{\bra{i}V_H^\dagger OV_H\ket{j}}^2
\label{eq:approximation}
\end{equation} 
can approximate $\norm{O}_{\mathrm{HShadow}}^2$ well. This formula can be easily extended to nonlinear observables and helps us build intuition between the sample complexity and the Hamiltonian. For example, when eigenstates of $H$ are close to the computational basis, $X_H$ approximates a diagonal matrix and $f(O,V_H)$ increases for having terms of $\frac{1}{(X_H)_{i,j}}$. 
It also agrees with the physical intuition since off-diagonal elements of $\rho$ are hard to probe when the unitary evolution $e^{-iHt}$ almost commutes with measurement basis. 

\begin{figure}
\centering
\includegraphics[width=0.95\linewidth]{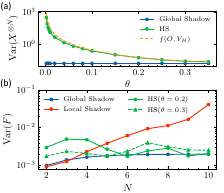}
\caption{Variance scaling. The evolution unitary is chosen to be $V_H\Lambda_t V_H^\dagger=e^{iP\theta}\Lambda_t e^{-iP\theta}$, where $\Lambda_t$ is a random diagonal unitary, $P$ is a fixed random Hermitian matrix and $\theta$ is a control parameter that determines the difference between measurement basis and eigen-basis of evolution unitary. The number of measurements is $K=1000$, and target states are set to be GHZ states \cite{median}. (a) The variance in estimating $P=X^{\otimes N}$, where $X$ is the Pauli-$X$ operator and $N=4$. (b) The variance scaling with qubit number in estimating fidelity. }
\label{fig:variance}
\end{figure}

To further validate Eq.~\eqref{eq:approximation}, we choose a Hamiltonian with the diagonalization unitary $V_H=e^{iP\theta}$, where $P$ is a fixed Hermitian matrix which is randomly sampled and $\theta$ is a tunable parameter. When $\theta$ approaches zero, eigen-basis of this Hamiltonian approaches the measurement basis. While $\theta$ is large, the eigen-basis approaches a random basis. In Fig.~\ref{fig:variance}(a), we use a Pauli observable to show that $f(O,V_H)$ is indeed a good approximation for the real variance. 
Another critical observation is that the variance of Hamiltonian shadow gradually approaches the original classical shadow protocol when $\theta$ increases. This can also be observed in Fig.~\ref{fig:variance}(b), where variances of Hamiltonian shadow for estimating fidelity do not increase with qubit number, reproducing a key feature of classical shadow. These observations show that, with a single Hamiltonian, the Hamiltonian shadow generally has a similar ability for state learning compared with the protocol using a set of many independent random unitaries. 
The similarity between the Hamiltonian shadow and the original shadow can also be theoretically verified through some special Hamiltonians \cite{wang2023MUB,zhou2023MUB}, which is discussed in Appendix \ref{sec:case_study}.

\textbf{Applications.} 
The Hamiltonian shadow is applicable to many analog systems, and here we choose Rydberg atom arrays with global laser pulse control as our platform. In the following, we will show how the Hamiltonian shadow accomplishes three important tasks for quantum many-body physics: (1) measuring quantum fidelity, (2) measuring local stabilizer observable of a topological state, and (3) measuring purity dynamics in quantum thermalization, all of which are thought to be hard to measure with only global controls.
Using the ground state $\ket{g}$ and the Rydberg state $\ket{r}$ of the neutral atom as a qubit, we model our system with Hamiltonian
\begin{equation}
H=\frac{\Omega}{2}\sum_j\left(e^{i\phi}\ketbra{g_j}{r_j}+h.c.\right)-\Delta \sum_j \hat{n}_j+\sum_{j<k}V_{jk}\hat{n}_j\hat{n}_k,
\end{equation}
where $\Omega=1.1\times 2\pi \ \mathrm{MHz}$, $\phi=2.1$, and $\Delta=1.2\times 2\pi  \ \mathrm{MHz}$ denote the Rabi frequency, laser phase, and detuning of the global driving laser field on atoms that couples the ground and Rydberg state. We will denote $\ket{g}$ as $\ket{0}$, and $\ket{r}$ as $\ket{1}$ afterwards. $V_{jk}=C/|\textbf{x}_j-\textbf{x}_k|^6$ describes the van der Waals interaction between two atoms, where $\textbf{x}$ is the position vector of an atom and the strength $C= 2\pi \times 862690 \ \mathrm{MHz}\cdot\mu m^6$ depends on the Rydberg state \footnote{All parameters of the Rydberg Hamiltonian are publically available from the website of \href{https://queracomputing.github.io/Bloqade.jl/dev/hamiltonians/}{Bloqade}}. With global control pulses, positions of atoms are randomly set to introduce randomness in eigen-energies of $H$, which is feasible for Rydberg atoms platforms \cite{Bluvstein2022moving}.
Specifically, we arrange atoms in a line and set the position of the $j$-th atom to be $j\times D+\delta d_j$, where $D=8.781 \mu m$ is near the blockade radius, and $\delta d_j$ is uniformly and independently sampled from $[-0.488,0.488](\mu m)$. Positions of atoms are kept the same throughout the process of Hamiltonian shadow. This setup is visualized in Fig.~\ref{fig:rydberg} (a). Ideally, the evolution time $t$ should be randomly chosen from a long time window $\Delta t=t_{\mathrm{max}}-t_{\mathrm{min}}$ such that $\Lambda_t=\text{diag}(e^{-iE_1 t},e^{-i E_2 t},\cdots)$ approximates random diagonal unitary. However, considering errors and decoherence in current platforms, $t$ can only be sampled in a limited time window. \comments{Thus, the distribution of $\Lambda_t$ may deviate from the ideal random diagonal unitary and introduce errors to the Hamiltonian shadow. }In the following numerical demonstrations, we will focus on the performance of the Hamiltonian shadow under a limited time window.

\begin{figure}[htbp]
    \centering
    \includegraphics[width=0.95\linewidth]{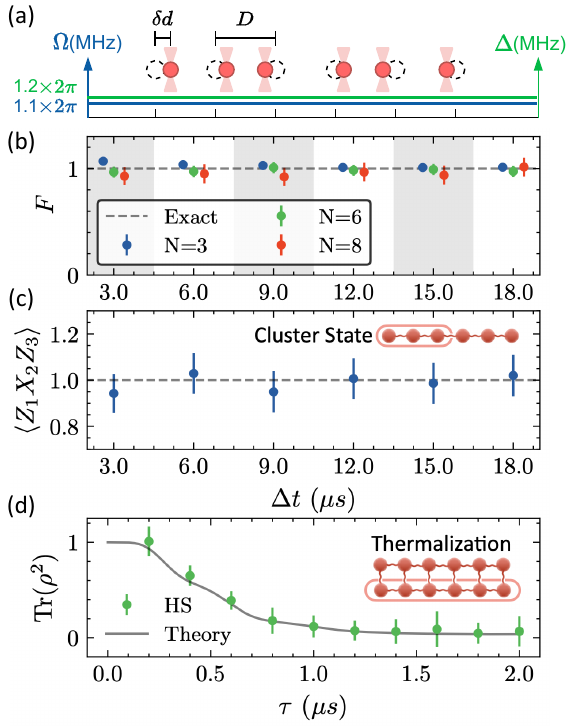}
    \caption{Performances of Hamiltonian shadow with Rydberg atom arrays. (a) The Rydberg atom array setup in Hamiltonian shadow. All atoms are controlled by global pulses ($\Omega,\Delta$) and arranged in a line with random positions. (b) Fidelity estimation with different time windows $\Delta t$ and qubit numbers $N$, where initial states are set to GHZ states. (c) Stabilizer expectation with different time windows, where the initial state is set to be a six-qubit cluster state. (d) Quantum thermalization of a 12-atom system, observed using purity of the 6-atom subsystem. All error bars are estimated with $K=10000$ samples (3 standard deviations) and $t_{\mathrm{min}}=2\mu s$. }
    \label{fig:rydberg}
\end{figure}

For measuring quantum fidelity, we initialize Greenberger–Horne–Zeilinger (GHZ) states with different qubit numbers, $\ket{\psi}=(\ket{010\cdots}+\ket{101\cdots})/2$, which can be prepared with Rydberg atom arrays using blockade mechanism \cite{omran2019catstate}. For measuring local observables, we prepare the underlying state to be the cluster state, a symmetry-protected topological state with $\langle Z_{i-1} X_{i} Z_{i+1}\rangle=1$ \cite{verresen2017spt}. To estimate local stabilizer, $\langle Z_{1} X_{2} Z_{3}\rangle$, we use the local-patch Hamiltonian shadow with a three-qubit Hamiltonian acting on target qubits. Both quantities are challenging to measure with conventional means using global pulse control. In Fig.~\ref{fig:rydberg} (b) and (c), we plot fidelity and stabilizer expectations as functions of the time window $\Delta t$. It shows that the Hamiltonian shadow protocol gives precise estimations for both quantities with realistic time windows. Moreover, in Appendix \ref{subsec:limit_time}, we show that the bias caused by the limited time window can be fixed by modifying the Hamiltonian shadow map. 

Besides linear properties of the quantum state, the Hamiltonian shadow can also predict nonlinear properties given only single-copy access to the quantum state each time. We demonstrate this by observing quantum thermalization under coherent dynamics. Twelve atoms are arranged in a two-leg ladder lattice with initial state $\ket{0}^{\otimes 6}\otimes\ket{1}^{\otimes 6}$, where the upper leg is $\ket{1}^{\otimes 6}$ and lower leg is $\ket{0}^{\otimes 6}$. The atoms are equally separated by $10.733\mu m$ and evolve under the Rydberg Hamiltonian for time $\tau$. With time evolution, entanglement between atoms on the upper and lower legs first grows and then saturates at a maximum value, which is reflected by the purities of atoms on the lower leg for different $\tau$. To estimate purities, after each evolution, we remove atoms on the upper leg and rearrange atoms on the lower leg to the selected random positions as shown in Fig.\ref{fig:rydberg} (a) and perform the Hamiltonian shadow protocol. The atom-moving has already been demonstrated in recent experiments \cite{Bluvstein2022moving}. To obtain the precise estimation as shown in Fig.\ref{fig:rydberg} (d), we set the maximal evolution time as $t_{\mathrm{max}}=60\mu s$. While compared to fidelity estimation in Fig.~\ref{fig:rydberg}(b) with the same qubit number, it indicates that a precise estimation of the nonlinear property may require a longer evolution time.

\textbf{Discussion.}
In this work, we show that it is possible to extract arbitrary quantum state properties on analog quantum simulators with a single Hamiltonian quench dynamics following simple computational basis measurements. 
The new protocol is experimentally appealing since one doesn't need to calibrate many Hamiltonians or use large ancillary systems. At the same time, the sample complexity is similar to the classical shadow protocol, which utilizes many independent Clifford gates. Moreover, the requirements for Hamiltonian are minimal, making it applicable to many different physical platforms.  
In the future, we will explore potential applications of Hamiltonian shadow in many different physical platforms, such as quantum gas microscopes with optical lattices \cite{Anton2017FH,christie2019string,thomas2023fermionpairing,lenard2023FQH}.
Besides, it is also desirable to find special Hamiltonians that can help to reduce the sample complexity and evolution time of the Hamiltonian shadow protocol.

\textbf{Acknowledgement} The authors would like to thank Satoya Imai, Otfried Gühne, Richard Kueng, Chao Yin, Soonwon Choi, Christian Kokail, Yi-Zhuang You, Susanne Yelin, Liang Jiang, Katherine Van Kirk, Yanting Teng, Jonathan Kunjummen, Stefan Ostermann, You Zhou, Fuchuan Wei, and Hongyi Wang, Yi Tan for insightful discussions and suggestions. Zhenhuan Liu is supported by the National Natural Science Foundation of China (Grant No. 12174216). H.Y.H. is grateful for the support from the Harvard Quantum Initiative Fellowship. The code is open-sourced on \href{https://github.com/Haozh20/Hamiltonian-Shadow-Repo.git}{GitHub}.

%\bibliography{Bib_RDShadow}

\appendix

\onecolumngrid
\newpage

\section{Preliminary}\label{sec:pre}
\subsection{Tensor Network Representation}
We introduce the tensor network representation in this section, which is an important tool for our derivation. In tensor network, a matrix is represented using a box with legs, shown in Fig.~\ref{fig:tensor}(a), where the left and right legs stand for the row and column indices, respectively. Different pairs of legs stand for different subsystems of the Hilbert space. Vectors are represented by the triangles with only one-side legs, as shown in Fig.~\ref{fig:tensor}(b) and (c). The connection of legs stands for the index contraction, such as the matrix production $AB$ shown in Fig.~\ref{fig:tensor}(d) and the trace function shown in Fig.~\ref{fig:tensor}(e). If two tensors are listed without connection of legs, like Fig.~\ref{fig:tensor}(f), this is the tensor product of two matrices $A\otimes B$.

\begin{figure}[hbp]
    \centering
    \includegraphics[width=0.7\textwidth]{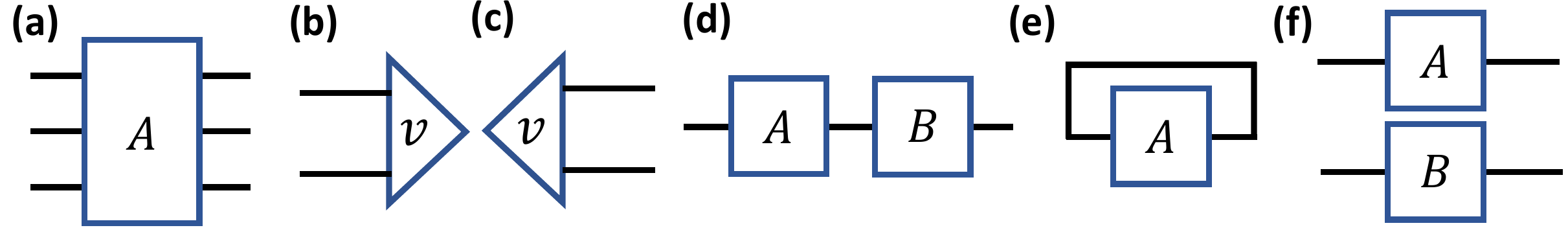}
    \caption{Tensor network representation. (a) A tripartite matrix. (b) A bipartite vector $\ket{v}$. (c) $\bra{v}$. (d) Matrix multiplication. (e) Trace. (f) Tensor product. }
    \label{fig:tensor}
\end{figure}

Tensor network is good at representing permutation operators, as shown in Fig.~\ref{fig:permutation}. A straight line is used to represent the identity operator $\mathbb{I}=\sum_i\ketbra{i}{i}$. A pair of cross lines shown in Fig.~\ref{fig:permutation}(b) represents the SWAP operator $S=\sum_{i,j}\ketbra{ij}{ji}$, which is a second order permutation operator. By adding more lines, we can represent higher-order permutation operators, like shown in Fig.~\ref{fig:permutation}(c). Tensor network representations can be used to simplify some calculations. Fig.~\ref{fig:permutation}(d) shows a graphical proof of the SWAP trick $\Tr(S\rho^{\otimes 2})=\Tr(\rho^2)$. 

\begin{figure}[htbp]
    \centering
    \includegraphics[width=0.7\textwidth]{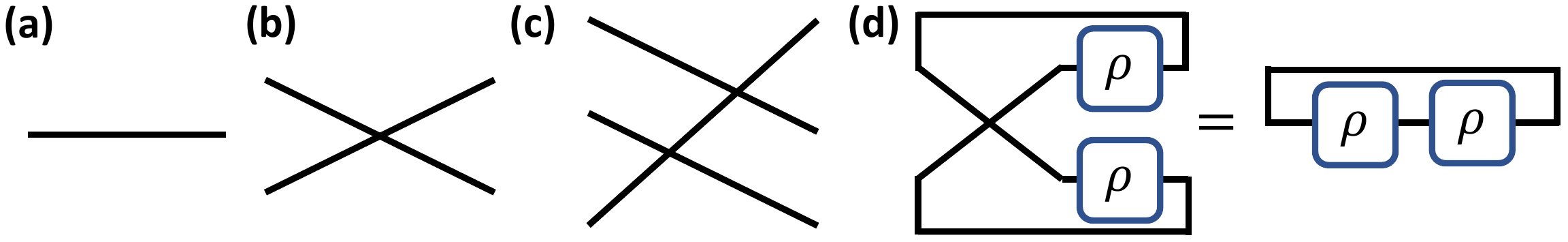}
    \caption{Tensor network representation of permutation operators. (a) The identity operator. (b) The SWAP operator. (c) The third-order cyclic permutation operator. (d) The graphical proof of the SWAP trick.}
    \label{fig:permutation}
\end{figure}

To facilitate our calculation of random diagonal unitaries, we need to introduce a new kind of labels. We use Fig.~\ref{fig:GHZ}(a) to represent the matrix of $\sum_i\ketbra{i,\cdots,i}{i,\cdots,i}$, which can be used to contract multiple indices. With this new tool, we can represent matrices which are generated by only preserving some specific elements of original matrices. Take a bipartite matrix, $M=\sum_{i,j,k,l}M_{ij,kl}\ketbra{ij}{kl}$, as an example, the tensor in Fig.~\ref{fig:GHZ}(b) gives a diagonal matrix $\sum_{i,j}M_{ij,ij}\ketbra{ij}{ij}$. The tensor in Fig.~\ref{fig:GHZ}(c) represents the matrix of $\sum_{i,j}X_{ij,ji}\ketbra{ij}{ji}$, whose nonzero elements lie in same positions of the SWAP operator $S=\sum_{i,j}\ketbra{ij}{ji}$. The tensor in Fig.~\ref{fig:GHZ}(d) is an EPR-like matrix $\sum_{i,j}X_{ii,jj}\ketbra{ii}{jj}$, whose nonzero elements are in positions of maximally entangled state. Fig.~\ref{fig:GHZ}(e) is constructed further by the EPR-like matrix, $\sum_{i}X_{ii,ii}\ketbra{ii}{ii}$.

\begin{figure}[hbp]
    \centering
    \includegraphics[width=0.6\textwidth]{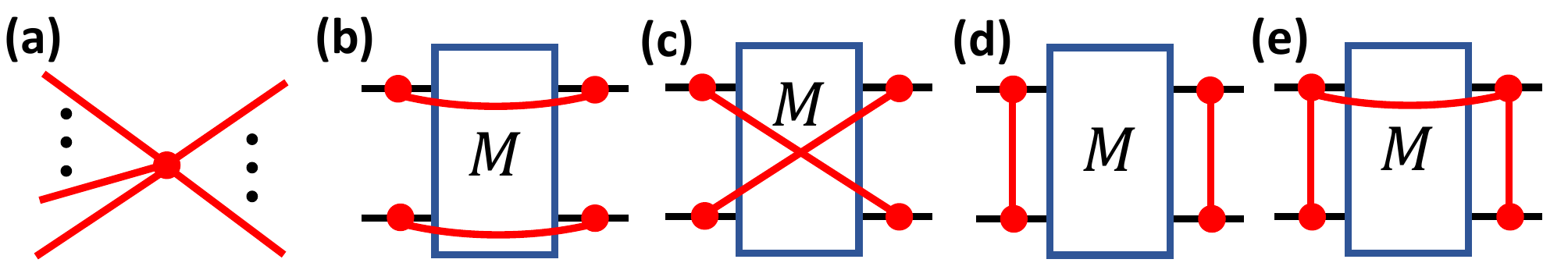}
    \caption{(a) A GHZ-like tensor, $\sum_i\ketbra{i,\cdots,i}{i,\cdots,i}$. (b)-(e), Matrices constructed using the GHZ-like tensor and a bipartite matrix $M=\sum_{i,j}M_{ij,ij}\ketbra{ij}{ij}$. }
    \label{fig:GHZ}
\end{figure}

Another tool that is important for our derivation is the Choi representation of a linear map. Shown in Fig~\ref{fig:choi}, the output of a linear map $C(\rho)$ can be represented using a higher-dimensional matrix $\mathcal{C}$ contracting with $\rho$. The matrix $\mathcal{C}$ is the Choi matrix of map $C(\cdot)$. In this work, we also refer to $\mathcal{C}^{T_1}$ as the Choi matrix of $C(\cdot)$.

\begin{figure}[htbp]
    \centering    \includegraphics[width=0.4\textwidth]{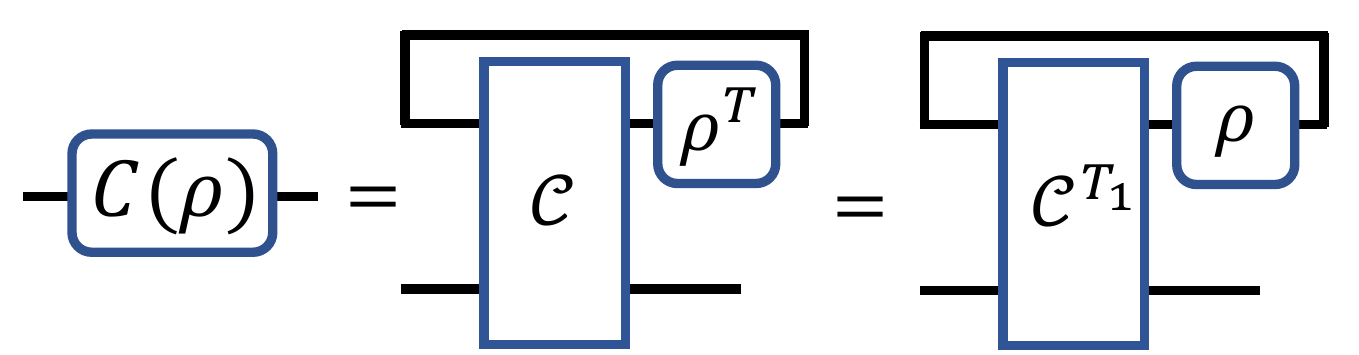}
    \caption{Tensor representation of Choi matrix.}
    \label{fig:choi}
\end{figure}

\subsection{Random Unitaries}
Haar measure random unitary is a uniform distribution in the unitary space, which satisfies 
\begin{equation}
\int_{U\sim\mathrm{Haar}}g(UV)dU=\int_{U\sim\mathrm{Haar}}g(U)dU
\end{equation}
for any unitary $V$ and function $g(\cdot)$. The Haar-measure random unitary is important for classical shadow due to its relation with permutation operators. According to Schur-Weyl duality, we define the $k$-th order twirling map
\begin{equation}
\Phi_k(M)=\int_{U\sim\mathrm{Haar}}U^{\otimes k}MU^{\dagger\otimes k}dU=\sum_{\pi,\sigma\in\mathcal{S}_k}\mathrm{Wg}_{\pi,\sigma}\Tr(\hat{W}_\pi M)\hat{W}_\sigma,
\end{equation}
where $\mathrm{Wg}$ stands for the Weingarten matrix \cite{gu2013moments}, $\mathcal{S}_k$ is the $k$-th order permutation group, $\pi$ and $\sigma$ stand for two elements in $\mathcal{S}_k$, and $\hat{W}_\pi$ and $\hat{W}_\sigma$ are corresponding permutation operators. In this work, we will denote the integral over Clifford group to be $\Phi_k^{\mathrm{C}}(\cdot)$ and random diagonal unitaries to be $\Phi_k^{\mathrm{D}}(\cdot)$. Specifically, the second order twirling function is
\begin{equation}\label{eq:second_twirling}
\Phi_2(M)=\frac{1}{d^2-1}\left(\Tr(\mathbb{I}M)\mathbb{I}-\frac{1}{d}\Tr(SM)\mathbb{I}-\frac{1}{d}\Tr(\mathbb{I}M)S+\Tr(SM)S\right)
\end{equation}

Instead of integrating over the Haar measure random unitary, we can use the average over a set containing a finite number of unitaries to get the same twirling map
\begin{equation}
\Phi_{k^\prime}^{\mathcal{E}_k}(M)=\frac{1}{|\mathcal{E}_k|}\sum_{U\in\mathcal{E}_k}U^{\otimes k^\prime}M U^{\dagger\otimes k^\prime}=\Phi_{k^\prime}(M)
\end{equation}
for all $k^\prime\le k$ and matrix $M$. We refer to $\mathcal{E}_k$ as the unitary $k$-design. The Clifford group has been proved to be a unitary $3$-design \cite{zhu2017clifford}.

\subsection{Random Diagonal Unitaries}\label{subsec:RDU}
A $d$-dimensional random diagonal unitary is $\Lambda=\mathrm{diag}(e^{i\theta_1}, \cdots, e^{i\theta_d})$, where $\theta_1,\cdots,\theta_d$ are random phases uniformly and independently sampled in $[0,2\pi)$. We define the map of $k$-th order integral over random diagonal unitaries as
\begin{equation}
\Phi_k^{\mathrm{D}}(M)=\int_{\Lambda\sim\mathrm{RDU}}\Lambda^{\otimes k}M\overline{\Lambda}^{\otimes k}d\Lambda,
\end{equation}
where $\overline{\Lambda}$ is the complex conjugate of $\Lambda$. 
In this work, $\Phi_2^{\mathrm{D}}(\cdot)$ and $\Phi_3^{\mathrm{D}}(\cdot)$ will be frequently employed to compute the unbiased estimator and the sample complexity. 
The action of the second-order map is 
\begin{equation}\label{eq:diagonal_second_twirling}
\Phi_2^{\mathrm{D}}(M)=\sum_{i,j}M_{ij,ij}\ketbra{ij}{ij}+\sum_{i,j}M_{ij,ji}\ketbra{ij}{ji}-\sum_{i}M_{ii,ii}\ketbra{ii}{ii}.
\end{equation}
We can easily prove this equation using the definition of random diagonal matrix. The element of $\Lambda^{\otimes2}M\overline{\Lambda}^{\otimes 2}$ is 
\begin{equation}
\left(\Lambda^{\otimes2}M\overline{\Lambda}^{\otimes 2}\right)_{ij,kl}=\Lambda_{i,i}\Lambda_{j,j}\overline{\Lambda}_{k,k}\overline{\Lambda}_{l,l}M_{ij,kl},
\end{equation}
which survives from the integral if and only if $i=k$ and $j=l$, or $i=l$ and $j=k$. These elements keep unchanged from the integral while others become zero. This concludes the proof of the second order integral, and the proof for higher order ones are similar. According to Ref.~\cite{nechita2021RDU}, the second and third order integrals can be graphically represented using Fig.~\ref{fig:integral}(a) and (b).

\begin{figure}[htbp]
    \centering
    \includegraphics[width=0.8\textwidth]{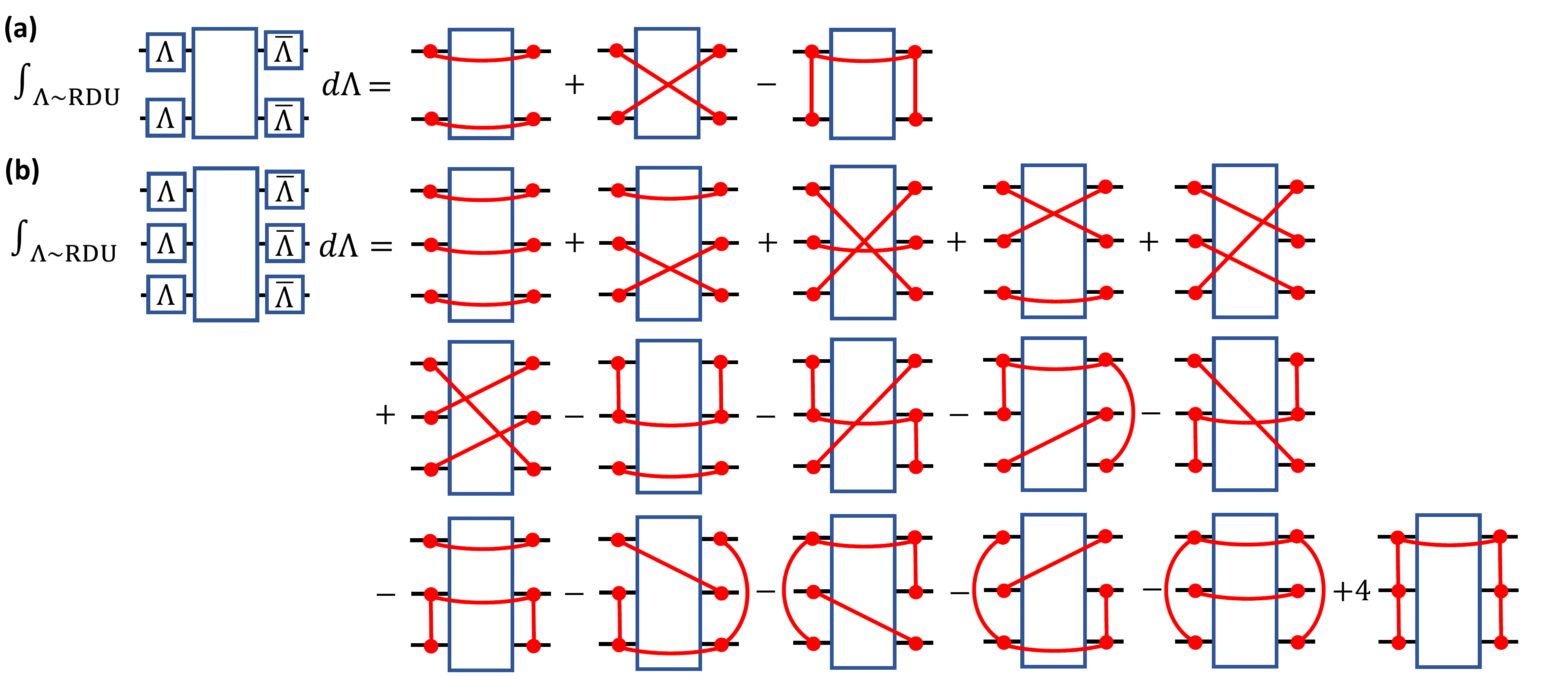}
    \caption{Integrals over random diagonal unitaries. (a) The second-order integral. (b) The third-order integral.}
    \label{fig:integral}
\end{figure}

Another important problem is that if the random diagonal unitary has degeneracy, how does the integral change? We consider a specific distribution of diagonal unitary $\mathcal{D}_1$. For all $\Lambda$ following the distribution of $\mathcal{D}_1$, 
\begin{equation}
\Lambda_{j,j}=\begin{cases}
e^{i\theta_j} & j\neq b\\
\Lambda_{a,a} & j=b
\end{cases},
\end{equation}
where $\theta_j$ is randomly and uniformly sampled from $[0,2\pi)$ and $a\neq b$. 
Following a similar proof of the second order integral over random diagonal unitaries, more elements can survive from the integral over $\mathcal{D}_1$,
\begin{equation}\label{eq:degeneracy}
\begin{aligned}
\int\limits_{\Lambda\sim\mathcal{D}_1}\Lambda^{\otimes2}M\overline{\Lambda}^{\otimes 2}d\Lambda&=\sum_{i,j}M_{ij,ij}\ketbra{ij}{ij}+\sum_{i,j}M_{ij,ji}\ketbra{ij}{ji}-\sum_{i}M_{ii,ii}\ketbra{ii}{ii}\\
&+\left(\sum_iM_{ai,bi}\ketbra{ai}{bi}+\sum_iM_{ai,ib}\ketbra{ai}{ib}+\sum_iM_{ia,bi}\ketbra{ia}{bi}+\sum_iM_{ia,ib}\ketbra{ia}{ib}+h.c.\right)\\
&-\left(M_{aa,ba}\ketbra{aa}{ba}+M_{aa,ab}\ketbra{aa}{ab}+M_{ba,bb}\ketbra{ba}{bb}+M_{ab,bb}\ketbra{ab}{bb}+h.c.\right)\\
&+(M_{aa,bb}\ketbra{aa}{bb}+M_{bb,aa}\ketbra{bb}{aa}).
\end{aligned}
\end{equation}
Besides, the second-order degeneracy can also affect the second-order integral. Consider a distribution of diagonal unitaries $\mathcal{D}_2$ without first-order degeneracy, whose element $\Lambda$ satisfies
\begin{equation}
\Lambda_{j,j}=\begin{cases}
e^{i\theta_j} & j\neq b_2\\
e^{i(\theta_{a_1}+\theta_{a_2}-\theta_{b_1})} & j=b_2
\end{cases},
\end{equation}
where $(a_1,a_2)\neq(b_1,b_2)$, $(a_1,a_2)\neq(b_2,b_1)$, $a_1\neq a_2$, and $b_1\neq b_2$. Similarly, every $\theta_j$ is sampled uniformly and independently from $[0,2\pi)$. It can be proved that
\begin{equation}\label{eq:second_degeneracy}
\begin{aligned}
&\int\limits_{\Lambda\sim\mathcal{D}_2}\Lambda^{\otimes2}M\overline{\Lambda}^{\otimes 2}d\Lambda=\sum_{i,j}M_{ij,ij}\ketbra{ij}{ij}+\sum_{i,j}M_{ij,ji}\ketbra{ij}{ji}-\sum_{i}M_{ii,ii}\ketbra{ii}{ii}\\
&+\left(M_{a_1a_2,b_1b_2}\ketbra{a_1a_2}{b_1b_2}+M_{a_2a_1,b_1b_2}\ketbra{a_2a_1}{b_1b_2}+M_{a_1a_2,b_2b_1}\ketbra{a_1a_2}{b_2b_1}+M_{a_2a_1,b_2b_1}\ketbra{a_2a_1}{b_2b_1}+h.c.\right).
\end{aligned}
\end{equation}
Moreover, if the third-order degeneracy exists, the second-order integral will not be affected. Following the same logic, one can prove that the k-th-order integral only requires that there exists no degeneracy less than k-th order. This tells us that integrating over some other distribution of diagonal unitaries instead of random diagonal unitaries can also give $\Phi_k^{\mathrm{D}}(\cdot)$. We can thus introduce the concept of diagonal unitary design.

Similar as the Haar measure random unitary, a set of diagonal unitaries $\mathcal{E}_k^{\mathrm{D}}$ is said to be a diagonal unitary $k$ design if
\begin{equation}
\Phi_{k^\prime}^{\mathcal{E}^{\mathrm{D}}_k}(M)=\frac{1}{|\mathcal{E}^{\mathrm{D}}_k|}\sum_{\Lambda\in\mathcal{E}^{\mathrm{D}}_k}\Lambda^{\otimes k^\prime}M \overline{\Lambda}^{\dagger\otimes k^\prime}=\Phi_{k^\prime}^{\mathrm{D}}(M)
\end{equation}
for all $k^\prime\le k$ and matrix $M$. We can prove that:
\begin{proposition}
Given a set of diagonal unitaries $\mathcal{E}_k^{\mathrm{D}}$, whose element $\Lambda=\mathrm{diag}(e^{i\theta_1}, \cdots, e^{i\theta_d})$. If every $\theta$ is sampled uniformly and independently from the set of $\{0,\frac{2\pi}{k+1},\frac{4\pi}{k+1},\cdots,\frac{2k\pi}{k+1}\}$, $\mathcal{E}_k^{\mathrm{D}}$ is a diagonal unitary $k$-design.
\end{proposition}

\begin{proof}
To prove that $\mathcal{E}_k^{\mathrm{D}}$ is a diagonal unitary $k$-design, we need to prove that the equality $\mathbb{E}_{\Lambda\in\mathcal{E}_k^{\mathrm{D}}}\Lambda^{\otimes k^\prime}M\overline{\Lambda}^{\otimes k^\prime}=\mathbb{E}_{\Lambda\sim\mathrm{RDU}}\Lambda^{\otimes k^\prime}M\overline{\Lambda}^{\otimes k^\prime}$ holds for all $k^\prime\le k$ and $M$. It is easy to find that every nonzero element of the right hand side matrix equals to the element in left hand side matrix at the same positions as random phases cancel out. Therefore, we only need to prove that every element of the left matrix in the same position with the zero element of the right one equals to zero. By expanding the left matrix, every element of it can be written as the multiplication of polynomials of $\Lambda_{i,i}$ and $\overline{\Lambda}_{j,j}$ and elements of $M$. Orders of this polynomials are less than $k^\prime$. Thus, if we could prove the equality of
\begin{equation}
\sum_{\theta\in\{0,\frac{2\pi}{k+1},\cdots,\frac{2k\pi}{k+1}\}}(e^{i\theta})^{k^\prime}=0
\end{equation}
for all $k^\prime\le k$, we can conclude our proof. Instituting the sum formula for proportional sequence of numbers, we have
\begin{equation}
\sum_{\theta\in\{0,\frac{2\pi}{k+1},\cdots,\frac{2k\pi}{k+1}\}}e^{ik^\prime\theta}=\frac{1-e^{i\frac{2k^\prime(k+1)}{k+1}\pi}}{1-e^{i\frac{2k^\prime}{k+1}\pi}}=0.
\end{equation}
\end{proof}

To quantify the distance between a given distribution of diagonal unitaries $\mathcal{D}$ and the set of ideal random diagonal unitaries, we introduce the concept of frame potential for diagonal unitaries. Defining $Q_k=\int_{U\sim\mathcal{D}}(U^{\dagger})^{\otimes k}\otimes U^{\otimes k} dU-\int_{U\sim\mathrm{RDU}}(U^{\dagger})^{\otimes k}\otimes U^{\otimes k} dU$, we have
\begin{equation}
0\le\Tr(Q_kQ_k^\dagger)=\int_{U\sim\mathcal{D}}\int_{V\sim\mathcal{D}}\abs{\Tr(U^\dagger V)}^{2k}dUdV
-2\int_{U\sim\mathcal{D}}\int_{V\sim\mathrm{RDU}}\abs{\Tr(U^\dagger V)}^{2k}dUdV
+\int_{U\sim\mathrm{RDU}}\int_{V\sim\mathrm{RDU}}\abs{\Tr(U^\dagger V)}^{2k}dUdV.
\end{equation}
Inserting another integral into the second term, we have
\begin{equation}
\begin{aligned}
\int_{U\sim\mathcal{D}}\int_{V\sim\mathrm{RDU}}\abs{\Tr(U^\dagger V)}^{2k}dUdV
=&\int_{U\sim\mathcal{D}}\int_{V\sim\mathrm{RDU}}\int_{W\sim\mathrm{RDU}}\abs{\Tr(U^\dagger VW^\dagger)}^{2k}dUdVdW\\
=&\int_{U\sim\mathcal{D}}\int_{V\sim\mathrm{RDU}}\int_{W\sim\mathrm{RDU}}\abs{\Tr\left[(UW)^\dagger V\right]}^{2k}dUdVdW\\
=&\int_{U\sim\mathrm{RDU}}\int_{V\sim\mathrm{RDU}}\abs{\Tr(U^\dagger V)}^{2k}dUdV,
\end{aligned}
\end{equation}
where we adopt properties of random diagonal unitaries that a random diagonal unitary can be written as the product of two independent random diagonal unitaries, and $UW$ is a random diagonal unitary if $U$ is a fixed unitary and $W$ is a random diagonal unitary. Therefore, we can define the $k$-th order frame potential of a distribution of diagonal unitary as $F^{(k)}_\mathcal{D}=\int_{U\sim\mathcal{D}}\int_{V\sim\mathcal{D}}\abs{\Tr(U^\dagger V)}^{2k}dUdV$ and show that
\begin{equation}
0\le\Tr(Q_tQ_t^\dagger)=\int_{U\sim\mathcal{D}}\int_{V\sim\mathcal{D}}\abs{\Tr(U^\dagger V)}^{2k}dUdV-\int_{U\sim\mathrm{RDU}}\int_{V\sim\mathrm{RDU}}\abs{\Tr(U^\dagger V)}^{2k}dUdV=F^{(k)}_\mathcal{D}-F^{(k)}_{\mathrm{RDU}}.
\end{equation}
This means that the $k$-th order frame potential of a diagonal unitary distribution is always larger than the frame potential of the distribution of ideal random diagonal unitaries. Only when the distribution is a diagonal unitary $k$-design, these two frame potentials are equivalent. Thus, we can use $F_\mathcal{D}^{(k)}$ to show the difference between $\mathcal{D}$ and diagonal unitary $k$-design.

Another thing we need to calculate is the value of frame potential for random diagonal unitaries. Denoting a random diagonal unitary by $U=\mathrm{diag}\left(z_1,\cdots,z_d\right)$, where these terms satisfy $\mathbb{E}\left[z_i^n (z_j^*)^m\right]=\delta_{n,m}\delta_{i,j}$, we have
\begin{equation}\label{eq:frame_potential_RDU}
\begin{aligned}
F_{\mathrm{RDU}}^{(k)}=&\int_{U\sim\mathrm{RDU}}\int_{V\sim\mathrm{RDU}}\abs{\Tr(U^\dagger V)}^{2k}dUdV\\
=&\int_{U\sim\mathrm{RDU}}\abs{\Tr(U)}^{2k}dU\\
=&\mathbb{E}\left[(z_1+\cdots+z_d)^k(z_1^*+\cdots+z_d^*)^k\right]\\
=&\mathbb{E}\left(\sum_{n_1,\cdots,n_d\in\mathbb{N}, n_1+\cdots+n_d=k}\binom{k}{n_1,\cdots,n_d}z_1^{n_1}\cdots z_d^{n_d}\right)\left(\sum_{n_1^\prime,\cdots,n_d^\prime\in\mathbb{N}, n_1^\prime+\cdots+n_d^\prime=k}\binom{k}{n_1^\prime,\cdots,n_d^\prime}(z_1^{n_1^\prime}\cdots z_d^{n_d^\prime})^*\right)\\
=&\sum_{n_1,\cdots,n_d\in\mathbb{N}, n_1+\cdots+n_d=k}\binom{k}{n_1,\cdots,n_d}^2.
\end{aligned}
\end{equation}
In this work, we mainly focus on $k=1,2,3$, the frame potentials can be calculated to be
\begin{equation}
F_{\mathrm{RDU}}^{(1)}=d \ , \ F_{\mathrm{RDU}}^{(2)}=2d^2-d \ , \ F_{\mathrm{RDU}}^{(3)}=6d^3-9d^2+4d.
\end{equation}

\section{Brief Introduction to Classical Shadow}\label{sec:classical_shadow}
In this section, we give a brief and graphical review of the derivation of original shadow map and its variance. This review will be helpful for our construction of Hamiltonian shadow protocol.

Here we focus on the original shadow protocol with global random Clifford gates, which can be easily extended to local version. One first acts a global random unitary $U$, which is sampled from the Haar measure random unitary or a Clifford group, on the target quantum state $\rho$. Then one performs the projective measurement on the evolved state to get the measurement result $\ket{b}$. The shadow map is defined as 
\begin{equation}
\mathcal{M}(\rho)=\mathbb{E}_{U,b}\left(U^\dagger\ketbra{b}{b}U\right)=\mathbb{E}_U\left(\sum_b\bra{b}U\rho U^\dagger\ket{b}U^\dagger\ketbra{b}{b}U\right).
\end{equation}
Using tensor network, we can graphically represent the shadow map as
\begin{equation}
\begin{aligned}
\mathcal{M}(\rho)=&\mathbb{E}_U\left(\sum_b\begin{tabular}{c}
\includegraphics[scale=0.25]{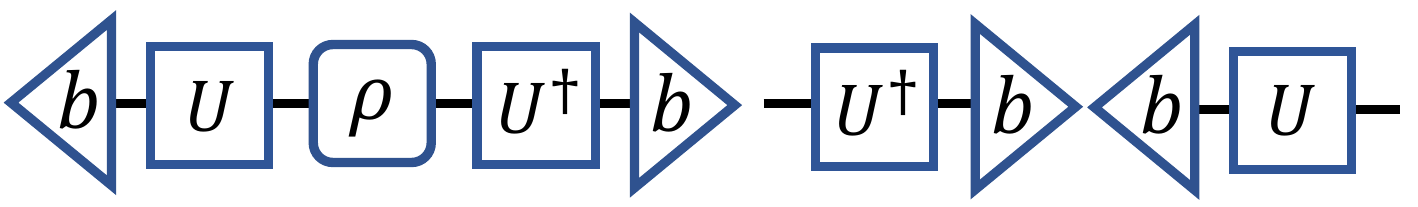}\end{tabular}\right)\\
=&\mathbb{E}_U\left(\sum_b\begin{tabular}{c}
\includegraphics[scale=0.25]{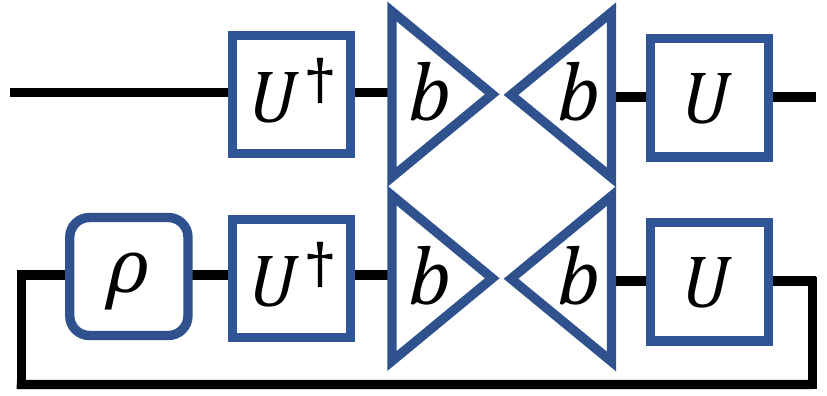}\end{tabular}\right)
=\mathbb{E}_U\left(\begin{tabular}{c}
\includegraphics[scale=0.25]{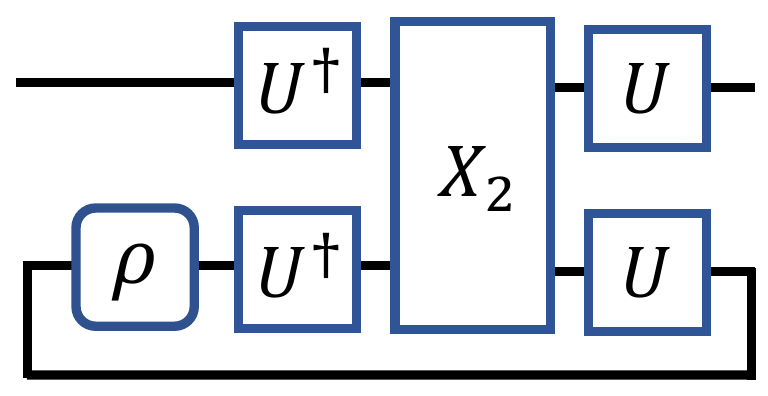}\end{tabular}\right),
\end{aligned}
\end{equation}
where $X_2=\sum_b\ketbra{bb}{bb}$. Combining the second order twirling function in Eq.~\eqref{eq:second_twirling} with the property of $X_2$, $\Tr(\mathbb{I}X_2)=\Tr(SX_2)=d$, we have
\begin{equation}
\mathcal{M}(\rho)=\begin{tabular}{c}
\includegraphics[scale=0.25]{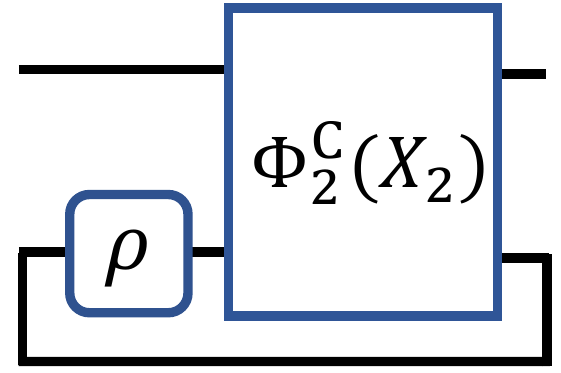}\end{tabular}=\frac{d-1}{d^2-1}\left(\begin{tabular}{c}
\includegraphics[scale=0.25]{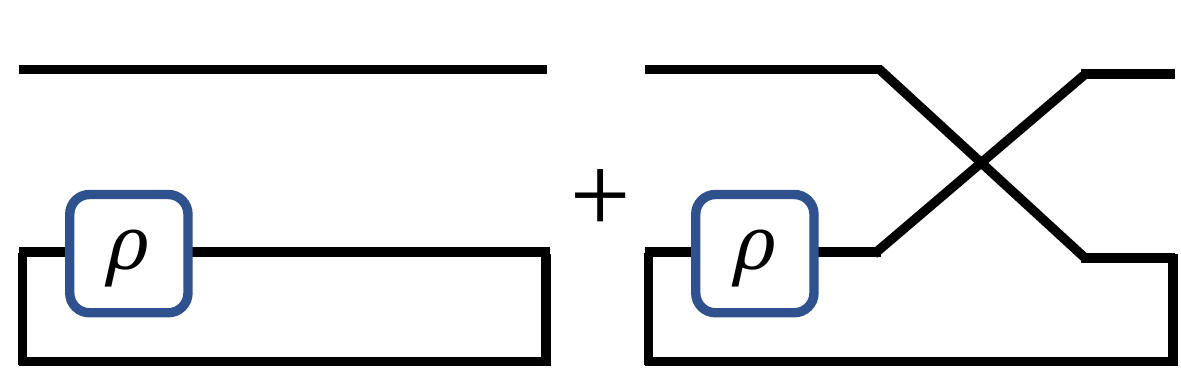}\end{tabular}\right)=\frac{1}{d+1}\left(\mathbb{I}+\rho\right).
\end{equation}
Therefore, the unbiased estimator of $\rho$ constructed by global shadow is
\begin{equation}
\hat{\rho}=\mathcal{M}^{-1}\left(U^\dagger\ketbra{b}{b}U\right)=(d+1)U^\dagger\ketbra{b}{b}U-\mathbb{I},
\end{equation}
and the unbiased estimator of $\Tr(O\rho)$ is $\hat{o}=\Tr(O\hat{\rho})=(d+1)\bra{b}UOU^\dagger\ket{b}-\Tr(O)$. The unbiasedness of $\hat{\rho}$ is shown by $\mathbb{E}_{U,b}\left(\hat{\rho}\right)=\mathcal{M}^{-1}\left[\mathbb{E}_{U,b}\left(U^\dagger\ketbra{b}{b}U\right)\right]=\mathcal{M}^{-1}[\mathcal{M}(\rho)]=\rho$.

Now we spend some time to see how to derive the variance of shadow estimator $\hat{o}$. By definition,
\begin{equation}\label{eq:shadow_variance_def}
\begin{aligned}
\mathrm{Var}(\hat{o})=&\mathbb{E}_{U,b}\left[(d+1)\bra{b}UOU^\dagger\ket{b}-\Tr(O)\right]^2-\Tr(O\rho)^2\\
=&\mathbb{E}_{U,b}\left[(d+1)^2\bra{b}UOU^\dagger\ket{b}^2\right]-2\Tr(O)\left[\Tr(O\rho)+\Tr(O)\right]+\Tr(O)^2-\Tr(O\rho)^2\\
=&\mathbb{E}_{U,b}\left[(d+1)^2\bra{b}UOU^\dagger\ket{b}^2\right]-2\Tr(O)\Tr(O\rho)-\Tr(O)^2-\Tr(O\rho)^2.
\end{aligned}
\end{equation}
We can expand the first term as
\begin{equation}\label{eq:shadow_variance_diag}
\begin{aligned}
&\mathbb{E}_{U,b}\left[(d+1)^2\bra{b}UOU^\dagger\ket{b}^2\right]=(d+1)^2\mathbb{E}_U\left[\sum_b\bra{b}U\rho U^\dagger\ket{b}\bra{b}UOU^\dagger\ket{b}^2\right]\\
=&(d+1)^2\mathbb{E}_U\left(\sum_b\begin{tabular}{c}
\includegraphics[scale=0.25]{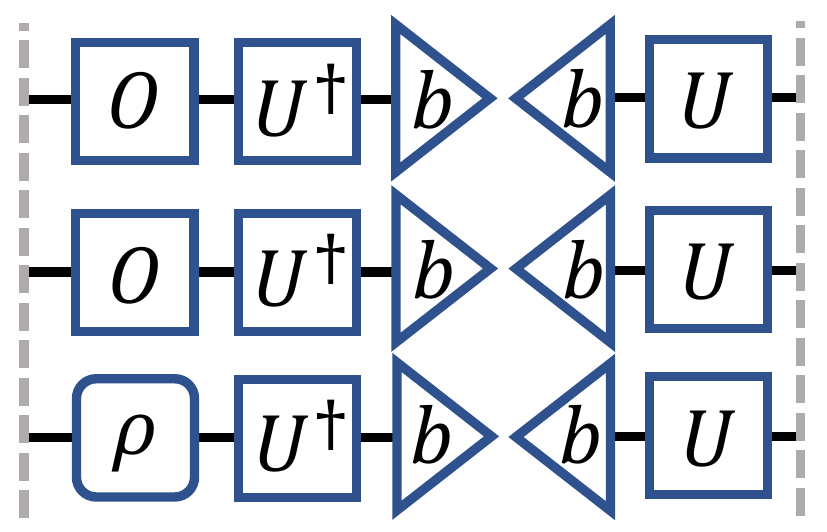}\end{tabular}\right)=(d+1)^2\begin{tabular}{c}
\includegraphics[scale=0.25]{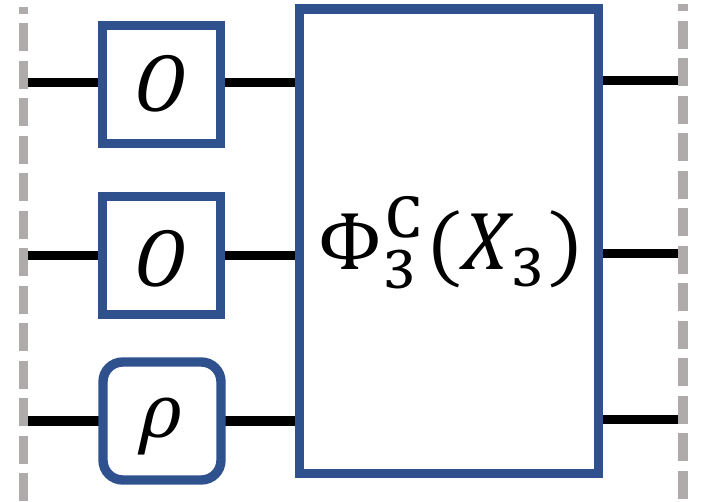}\end{tabular}
=\frac{d+1}{d+2}\sum_{\pi\in\mathcal{S}_3}\begin{tabular}{c}
\includegraphics[scale=0.25]{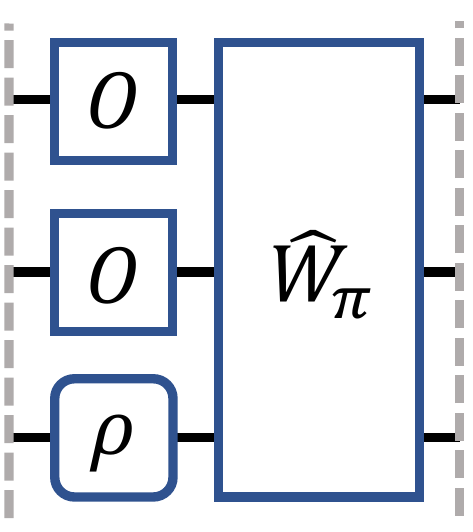}\end{tabular}\\
=&\frac{d+1}{d+2}\left[\Tr(O)^2+\Tr(O^2)+2\Tr(O\rho)\Tr(O)+2\Tr(O^2\rho)\right],
\end{aligned}
\end{equation}
where grey dashed lines represent the trace function, $X_3=\sum_b\ketbra{bbb}{bbb}$, $\mathcal{S}_3$ is the third-order permutation group, and $\Phi_3^{\mathrm{C}}(\cdot)$ is the third-order twirling map over global Clifford group whose explicit form can be found in Ref.~\cite{gu2013moments}. The last but one equal sign is due to the fact of $\tr(W_\pi X_3)=d$ for all $\pi\in\mathcal{S}_3$. Combining Eq.~\eqref{eq:shadow_variance_def} and Eq.~\eqref{eq:shadow_variance_diag}, we have
\begin{equation}
\mathrm{Var}(\hat{o})=\frac{d+1}{d+2}\left[\Tr(O^2)+2\Tr(O^2\rho)\right]-\frac{1}{d+2}\left[2\Tr(O)\Tr(O\rho)+\Tr(O)^2\right]-\Tr(O\rho)^2\le 3\Tr(O^2),
\end{equation}
where we use the relation of $\Tr(O^2\rho)\le\Tr(O^2)$. While, except for some special cases like $O=\rho=\ketbra{\psi}{\psi}$, $\Tr(O^2\rho)$ is normally much smaller than $\Tr(O^2)$. So, we obtain the key observation that the leading term of variance is
\begin{equation}
\Tr\left[\left(O^{\otimes 2}\otimes\rho\right)\left(S\otimes\mathbb{I}\right)\right]=\begin{tabular}{c}
\includegraphics[scale=0.25]{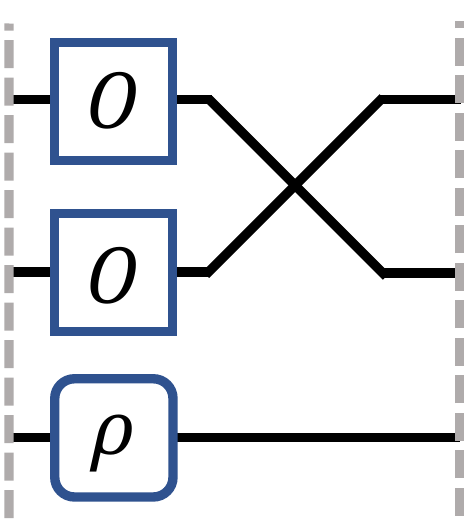}\end{tabular},
\end{equation}
which stands for the indices contraction between two observable matrices. This will provide us important intuition when we calculate the variance of our Hamiltonian shadow.

\section{Hamiltonian Shadow Map}\label{sec:shadow_map}

In this section, we will derive the unbiased Hamiltonian shadow estimators for cases of global and local Hamiltonian evolution, shown in Fig.~\ref{fig:Hamiltonian}. We will start from the global case and extend it to the local case.

\begin{figure}[htbp]
    \centering
    \includegraphics[width=0.7\textwidth]{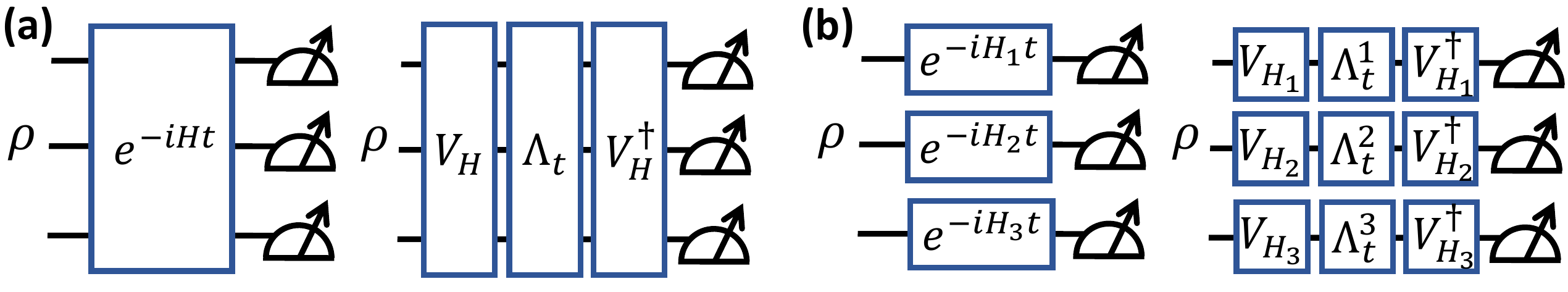}
    \caption{The Hamiltonian evolution with different times. (a) Global Hamiltonian. (b) Local Hamiltonian.}
    \label{fig:Hamiltonian}
\end{figure}

In a single round of experiment, the state $\rho$ will be evolved using a unitary $e^{-iHt}$ with a fixed Hamiltonian and a random $t$. This unitary can be decomposed as $V_H \Lambda_t V_H^\dagger$, where $V_H$ is a fixed unitary that is independent of $t$ and $\Lambda_t=\mathrm{diag}(e^{-iE_1t},\cdots,e^{-iE_dt})$. As shown in Fig.~\ref{fig:Hamiltonian}, after each experiment, suppose the measurement outcome is $b$, the shadow map is
\begin{equation}
\mathcal{M}_H(\rho)=\mathbb{E}_{b,t}\left(e^{iHt}\ketbra{b}{b}e^{-iHt}\right)=V_H\left[\mathbb{E}_{b,t}\left(\overline{\Lambda}_tV_H^\dagger\ketbra{b}{b}V_H\Lambda_t\right)\right]V_H^\dagger
\end{equation}
Substituting the Born's rule, we have
\begin{equation}
\begin{aligned}
\mathcal{M}_H(\rho)=&\mathbb{E}_t\left(\sum_b\bra{b}e^{-iHt}\rho e^{iHt}\ket{b}e^{iHt}\ketbra{b}{b}e^{-iHt}\right)\\
=&\mathbb{E}_t\left(\sum_b\begin{tabular}{c}
    \includegraphics[scale=0.25]{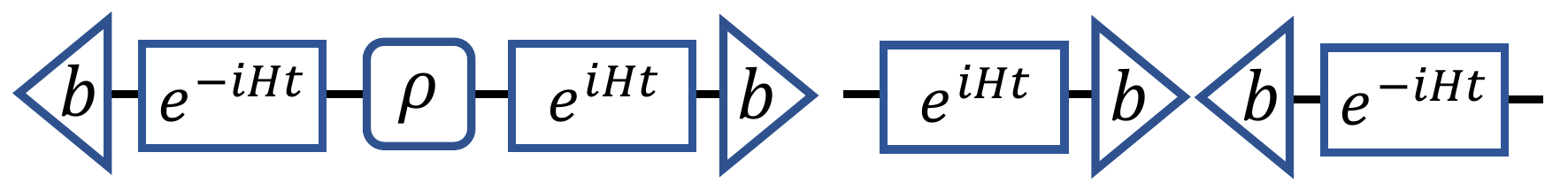}
\end{tabular}\right)\\
=&\mathbb{E}_t\left(\sum_b\begin{tabular}{c}
    \includegraphics[scale=0.25]{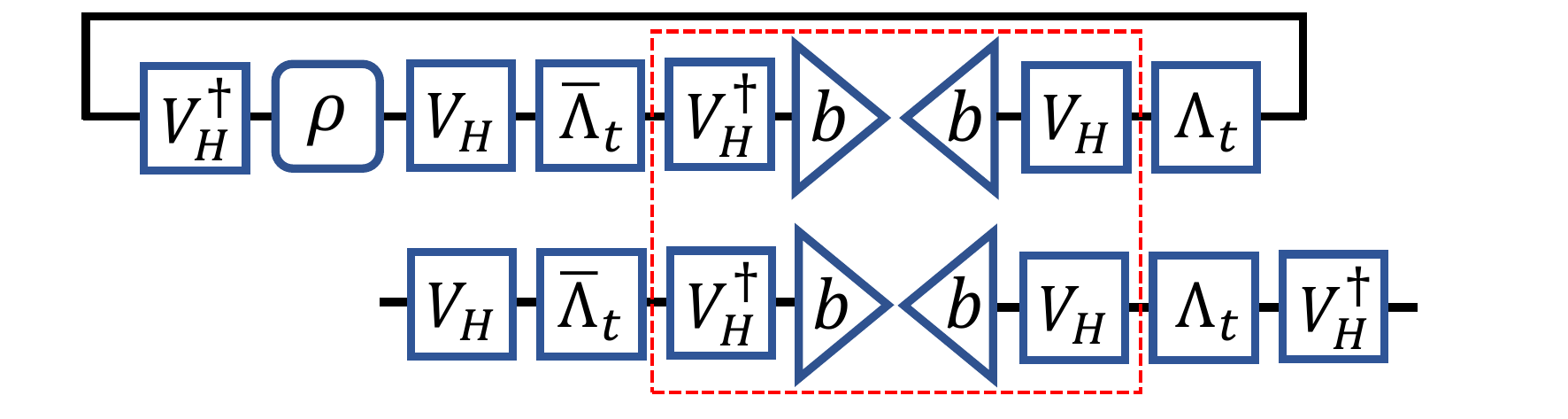}
\end{tabular}\right).
\end{aligned}
\end{equation}
By slightly abusing the indicators, we define $X_2=\sum_b(V_H^\dagger)^{\otimes 2}\ketbra{bb}{bb}V_H^{\otimes 2}$ hereafter. The specific form of Hamiltonian shadow map highly depends on the property of $\Lambda_t$. According to Eq.~\eqref{eq:degeneracy} and Eq.~\eqref{eq:second_degeneracy}, correlations among different eigenvalues of $H$, like the first and second-order degeneracy, and limited evolution time can significantly affect $\mathcal{M}_H$ and sometimes can even make it irreversible. We will discuss this in detail later. Assuming $\Lambda_t$ to be an ideal random diagonal unitaries, we have
\begin{equation}
\begin{aligned}
\mathcal{M}_H(\rho)=&\mathbb{E}_t\left(\begin{tabular}{c}
    \includegraphics[scale=0.25]{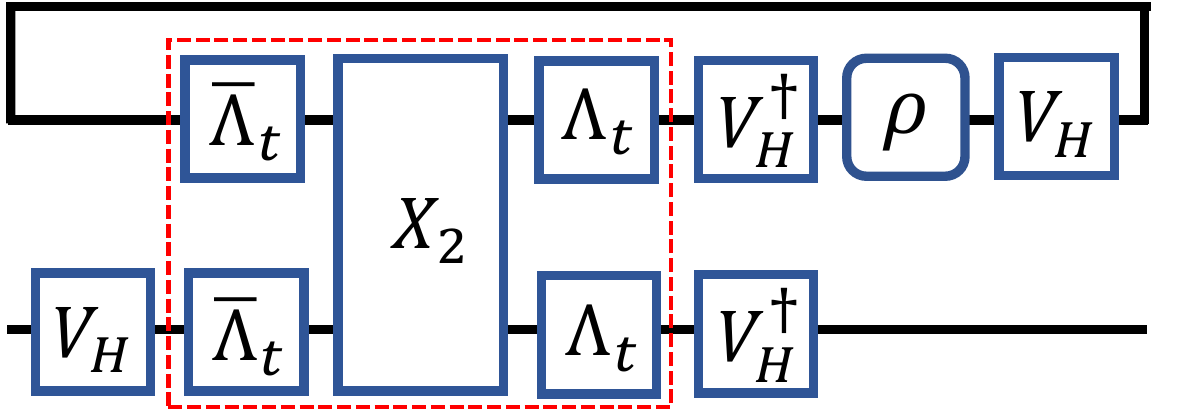}.
\end{tabular}\right)
=\begin{tabular}{c}
    \includegraphics[scale=0.25]{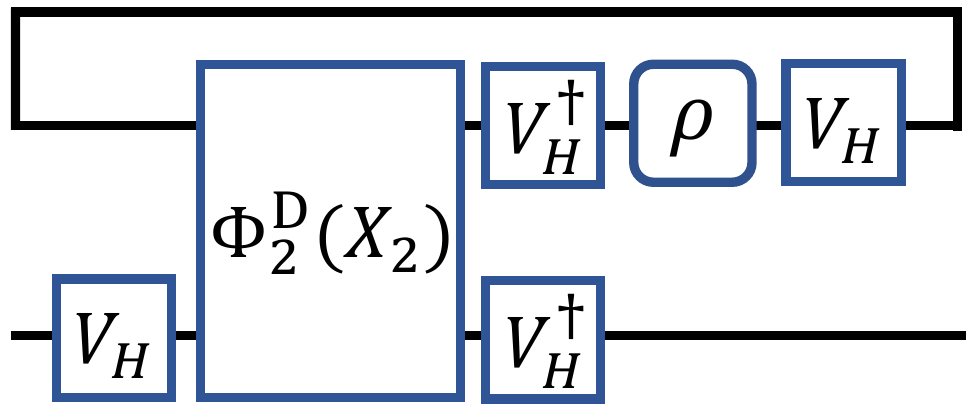}
\end{tabular}\\
=&\begin{tabular}{c}
    \includegraphics[scale=0.25]{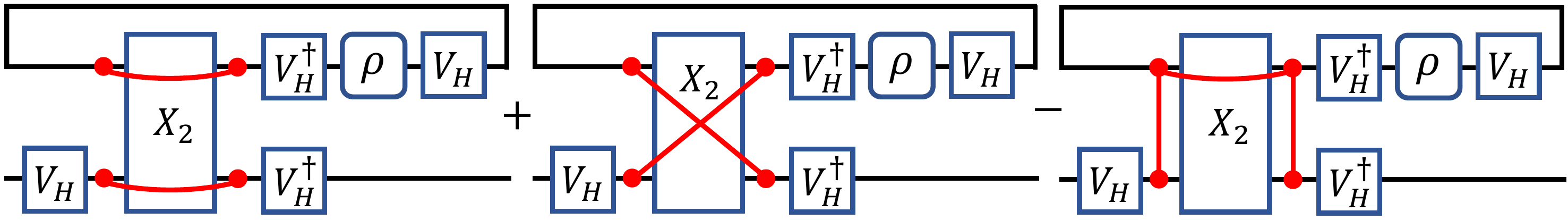}
\end{tabular}\\
=& V_H\mathcal{N}(V_H^\dagger\rho V_H)V_H^\dagger,
\end{aligned}
\end{equation}
where the Choi matrix of $\mathcal{N}$ is $\Phi_2^{\mathrm{D}}(X_2)$. Notice that whether this classical shadow protocol is tomography-complete depends on the reversibility of $\mathcal{N}$ and the Hamiltonian shadow estimator is
\begin{equation}\label{eq:estimator_in_app}
\hat{\rho}=\mathcal{M}_H^{-1}\left(e^{iHt}\ketbra{b}{b}e^{-iHt}\right)=
V_H\mathcal{N}^{-1}\left(V_H^\dagger e^{iHt}\ketbra{b}{b} e^{-iHt} V_H\right)V_H^\dagger=
V_H\mathcal{N}^{-1}\left(\overline{\Lambda}_tV_H^\dagger\ketbra{b}{b} V_H \Lambda_t\right)V_H^\dagger.
\end{equation}

Generally speaking, it is hard to invert a general linear map. To ease the computational complexity, we need to utilize the structure of $\Phi_2^{\mathrm{D}}(X_2)$. Using the second-order integral over random diagonal unitaries, Eq.~\eqref{eq:diagonal_second_twirling},
the action of map $\mathcal{N}$ can be simplified as
\begin{equation}
\begin{aligned}
\mathcal{N}(\sigma)=&\sum_{i,j}(X_2)_{ij,ij}\sigma_{i,i}\ketbra{j}{j}+\sum_{i,j}(X_2)_{ij,ji}\sigma_{j,i}\ketbra{j}{i}-\sum_i(X_2)_{ii,ii}\sigma_{i,i}\ketbra{i}{i}\\
=&\sum_j\left(\sum_i(X_2)_{ij,ij}\sigma_{i,i}\right)\ketbra{j}{j}+\sum_{i\neq j}(X_2)_{ij,ji}\sigma_{j,i}\ketbra{j}{i}.
\end{aligned}
\end{equation}
It is shown that only a small part of elements of $X_2$ contributes to the definition of map $\mathcal{N}$. We thus define a new matrix $X_H$ with $(X_H)_{i,j}=(X_2)_{ji,ji}$. Substituting the definition of $X_2$, we have
\begin{equation}
\begin{aligned}
(X_H)_{i,j}&=(X_2)_{ji,ji}=\sum_b\bra{ji}(V_H^\dagger)^{\otimes 2}\ketbra{bb}{bb}V_H^{\otimes 2}\ket{ji}=\sum_b|(V_H)_{i,b}|^2|(V_H)_{j,b}|^2=(X_2)_{ij,ij}=(X_H)_{j,i},\\
(X_2)_{ij,ji}&=\sum_b\bra{ij}(V_H^\dagger)^{\otimes 2}\ketbra{bb}{bb}V_H^{\otimes 2}\ket{ji}=\sum_b|(V_H)_{i,b}|^2|(V_H)_{j,b}|^2=(X_2)_{ji,ij}=(X_H)_{i,j}.
\end{aligned}
\end{equation}
Then, the map $\mathcal{N}$ can be further simplified as
\begin{equation}
\mathcal{N}(\sigma)=\sum_i\left(\sum_j(X_H)_{i,j}\sigma_{j,j}\right)\ketbra{i}{i}+\sum_{i\neq j}(X_H)_{i,j}\sigma_{i,j}\ketbra{i}{j}.
\end{equation}
The inverse map can thus be written as 
\begin{equation}\label{eq:inverse_N_map}
\mathcal{N}^{-1}(\sigma)=\sum_i\left(\sum_j(X_H^{-1})_{i,j}\sigma_{j,j}\right)\ketbra{i}{i}+\sum_{i\neq j}(X_H)_{i,j}^{-1}\sigma_{i,j}\ketbra{i}{j}.
\end{equation}
Notice that the matrix $X_H$ is related with $V_H$ in a more straightforward way. Defining $V_H^{\mathrm{sq}}=\sum_{i,j}\abs{(V_H)_{i,j}}^2\ketbra{i}{j}$, we can prove that
\begin{equation}
X_H=\sum_{i,j}\left(\sum_b\abs{(V_H)_{b,i}}^2\abs{(V_H)_{b,j}}^2\right)\ketbra{i}{j}=(V_H^{\mathrm{sq}})^TV_H^{\mathrm{sq}}.
\end{equation}

Based on Eq.~\eqref{eq:inverse_N_map}, we can derive the conditions under which the Hamiltonian shadow map is reversible.
Firstly, the matrix $X_H$ needs to be invertible, which makes sure that diagonal elements of $V_H^\dagger\rho V_H$ can be estimated. Secondly, off-diagonal terms of $X_H$ cannot be zero, which ensures that all off-diagonal terms of $V_H^\dagger\rho V_H$ can be estimated. Note that, if some condition is not satisfied, although the shadow map is not invertible, we can still estimate some elements of $V_H^\dagger\rho V_H$ by taking the peudo-inverse of $\mathcal{N}$.

\subsection{Local Version}
These conclusions can be easily extended to local version, where the evolution unitary is $\bigotimes_{p=1}^Ne^{-iH_pt}$. We assume all the single-patch Hamiltonians $H_p$ are independent and the evolution time $t$ is also randomly sampled. Denoting the measurement result as $\bigotimes_{p=1}^N\ketbra{b_p}{b_p}$, the corresponding shadow map is 
\begin{equation}
\begin{aligned}
\mathbb{E}_{t}\left(\bigotimes_{p=1}^N\mathbb{E}_{b_p}e^{iH_pt}\ketbra{b_p}{b_p}e^{-iH_pt}\right)
=\mathbb{E}_t\left(\begin{tabular}{c}
    \includegraphics[scale=0.22]{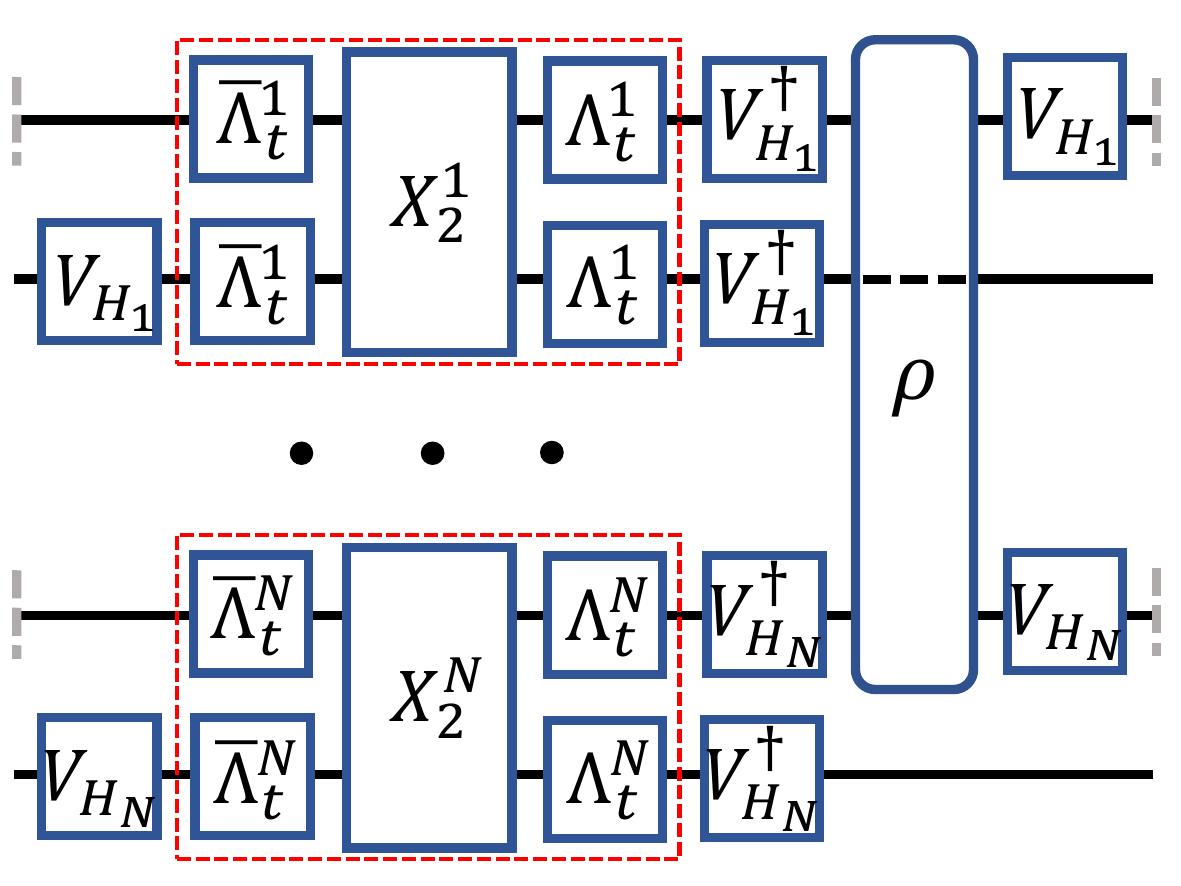}
\end{tabular}\right)
=\begin{tabular}{c}
    \includegraphics[scale=0.22]{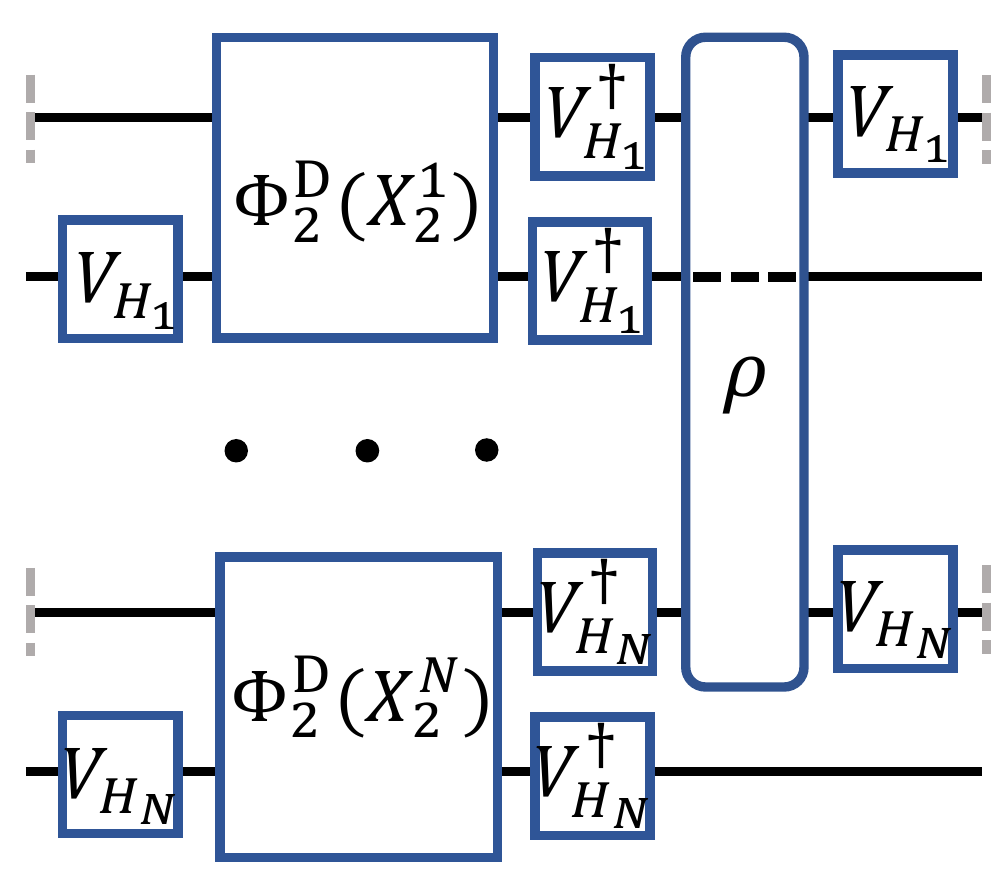}
\end{tabular}
=\left(\bigotimes_{p=1}^N\mathcal{M}_{H_p}\right)(\rho),
\end{aligned}
\end{equation}
where $X_2^p=\sum_{b_p}V_{H_p}^{\dagger\otimes 2}\ketbra{b_pb_p}{b_pb_p}V_{H_p}^{\otimes 2}$. 
It can be similarly proved that the following gives the unbiased estimator of $\rho$
\begin{equation}
\hat{\rho}=\bigotimes_{p=1}^NV_{H_p}\mathcal{N}_p^{-1}\left(\overline{\Lambda}_{t}^pV_{H_p}^\dagger\ketbra{b_p}{b_p}V_{H_p}\Lambda_{t}^p\right)V_{H_p}^\dagger,
\end{equation}
where $\mathcal{N}_p$ is defined in the same way as the global version using the single-patch Hamiltonian $H_p$.

Note that the tensor product structure of the Hamiltonian shadow map, i.e. the tensor product structure of Choi matrix, highly depends on the requirement that all $\Lambda_t^p$ are mutually independent. If this condition is not satisfied, such as some eigenvalues of $H_p$ are equivalent with some eigenvalues of $H_{p^\prime}$, we will fail to get a Choi matrix with a tensor product structure. Considering the difficulty in choosing different local Hamiltonians for some analog systems, one can also set different evolution times $t_p$ for different patches to achieve a same target.

\subsection{Numerical Demonstration}
\begin{figure}
\centering
\includegraphics[width=0.6\linewidth]{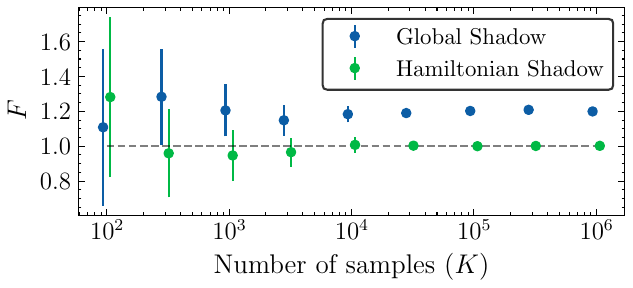}
\caption{Performance comparison between two data post-processing methods, the Hamiltonian shadow and the original global shadow, in predicting fidelity. A four-qubit GHZ state is evolved under $e^{-iHt}$ with a single random Hermitian matrix $H$ and random evolution time $t$ and measured in computational basis to get $\ket{b}$. Then, both methods use the dataset of $\{e^{-iHt_j},b_j\}_{j=1}^K$ to construct their fidelity estimators. The error bar indicates a 99.7$\%$ confidence interval (3 standard deviations).}
\label{fig:unbiase_numerics}
\end{figure}

We use a simple numerical experiment to show the unbiasedness of Hamiltonian shadow estimator.
It is also worth emphasising that directly treating $e^{-iHt}$ as a random Haar unitary and performing the data post-processing of the original global classical shadow will lead to biased estimations.
In Fig.~\ref{fig:unbiase_numerics}, we show the fidelity estimation with two data post-processing methods, Hamiltonian shadow and global shadow, using the same measurement dataset collected with a single Hamiltonian quench evolution. 
It is clear that the estimation of Hamiltonian shadow approaches the real value while the global shadow does not.

\subsection{Limited Evolution Time}\label{subsec:limit_time}
Derivations till now are based on the assumption that $\Lambda_t$ is an ideal random diagonal unitary. While, in some cases where the time period $t\in [t_{\mathrm{min}},t_{\mathrm{max}}]$ is not long enough, $\Lambda_t$ has certain distance with ideal random diagonal unitary and $X_H$ cannot fully describe the action of $\mathcal{N}$. We need to recalculate the integral of $\Phi_2^{\mathrm{D}}(X_2)$ as 
\begin{equation}
\Phi_2^{\mathrm{\Delta t}}(X_2)=\frac{1}{t_{\mathrm{max}}-t_{\mathrm{min}}}\int_{t_{\mathrm{min}}}^{t_{\mathrm{max}}}\Lambda_t^{\otimes 2}X_2\overline{\Lambda}_t^{\otimes 2}dt,
\end{equation}
where $\Lambda_t=\mathrm{diag}(e^{-iE_1t},e^{-iE_1t},\cdots,e^{-iE_dt})$. Elements of $\Phi_2^{\Delta t}(X_2)_{ij,ij}$ and $\Phi_2^{\Delta t}(X_2)_{ij,ji}$ are the same as $\Phi_2^{\mathrm{D}}(X_2)_{ij,ij}$ and $\Phi_2^{\mathrm{D}}(X_2)_{ij,ji}$ as phases cancel out. While, other terms are not zero when $\Delta t=t_{\mathrm{max}}-t_{\mathrm{min}}$ is finite. This can be proved from the integral,
\begin{equation}\label{eq:element_integral}
\begin{aligned}
\Phi_2^{\Delta t}(X_2)_{ij,kl}=&\frac{1}{t_{\mathrm{max}}-t_{\mathrm{min}}}\int_{t_{\mathrm{min}}}^{t_{\mathrm{max}}}e^{-i(E_i+E_j-E_k-E_l)t}(X_2)_{ij,kl}dt\\
=&\frac{(X_2)_{ij,kl}}{i(t_{\mathrm{max}}-t_{\mathrm{min}})(E_k+E_l-E_i-E_j)}\left(e^{-i(E_i+E_j-E_k-E_l)t_{\mathrm{max}}}-e^{-i(E_i+E_j-E_k-E_l)t_{\mathrm{min}}}\right).
\end{aligned}
\end{equation}
Thus, except for cases where $E_{i}+E_j=E_k+E_l$, other matrix elements of $\Phi_2^{\Delta t}(X_2)$ all decrease with the time period $\Delta t$. While for finite $\Delta t$, these elements are normally not zero. 

The first information we get from Eq.~\eqref{eq:element_integral} is that the only requirement for eigenvalues of $H$ is that except for $(i,j)=(k,l)$ and $(i,j)=(l,k)$, $E_i+E_j\neq E_k+E_l$. 
In addition to this, other correlations of different eigenvalues do not affect the unbiasedness of Hamiltonian shadow when $\Delta t$ is sufficiently large. 
This is because when $E_i+E_j\neq E_k+E_l$, $[\Phi_2^{\Delta t}(X_2)]_{ij,kl}$ will decay to zero when $\Delta t\to\infty$. Such conclusion also meets our analysis in Appendix \ref{subsec:RDU}, stated that the degeneracy higher than order two does not affect the second-order integral.
Besides, when $\Delta t$ is finite, we can adjust the post-processing of Hamiltonian shadow protocol to remove the bias. 
In this case, the matrix $X_H$ cannot fully describe the map $\mathcal{N}$, as the Choi matrix of $\mathcal{N}$ is replaced from $\Phi_2^{\mathrm{D}}(X_2)$ to $\Phi_2^{\Delta t}(X_2)$. While, combining Eq.~\eqref{eq:element_integral} and information of $\Delta t$ and $H$, we can still numerically determine the map $\mathcal{N}$ and its inverse $\mathcal{N}^{-1}$. Then, following the same logic, Eq.~\eqref{eq:estimator_in_app} with a new definition of $\mathcal{N}^{-1}$ can give the unbiased estimator of $\rho$.

In the case of finite time scale, we can also derive an analytical expression of the frame potential, which can be used to show the difference between $\Lambda_t$ and the ideal random diagonal unitary. Assuming $z_j=e^{-iE_j(t_2-t_1)}$ for $1\le j\le d$, the $k$-th order frame potential is
\begin{equation}
\begin{aligned}
&F^{(k)}_{\Delta t}=\mathbb{E}_{t_1,t_2}\abs{\Tr(\overline{\Lambda}_{t_1}\Lambda_{t_2})}^{2k} \\
=&\mathbb{E}_{t_1,t_2}\left(\sum_{n_1,\cdots,n_d\in\mathbb{N}, n_1+\cdots+n_d=k}\binom{k}{n_1,\cdots,n_d}z_1^{n_1}\cdots z_d^{n_d}\right)\left(\sum_{n_1^\prime,\cdots,n_d^\prime\in\mathbb{N}, n_1^\prime+\cdots+n_d^\prime=k}\binom{k}{n_1^\prime,\cdots,n_d^\prime}(z_1^{n_1^\prime}\cdots z_d^{n_d^\prime})^*\right)\\
=&\frac{1}{(t_{\max}-t_{\min})^2}\sum_{n_1,\cdots,n_d\in\mathbb{N}, n_1+\cdots+n_d=k}\sum_{n_1^\prime,\cdots,n_d^\prime\in\mathbb{N}, n_1^\prime+\cdots+n_d^\prime=k}\binom{k}{n_1,\cdots,n_d}\binom{k}{n_1^\prime,\cdots,n_d^\prime}\int_{t_{\min}}^{t_{\max}}\int_{t_{\min}}^{t_{\max}}z_1^{n_1-n_1^\prime}\cdots z_d^{n_d-n_d^\prime}dt_1dt_2\\
=&\frac{1}{(t_{\max}-t_{\min})^2}\sum_{n_1,\cdots,n_d}\sum_{n_1^\prime,\cdots,n_d^\prime}\binom{k}{n_1,\cdots,n_d}\binom{k}{n_1^\prime,\cdots,n_d^\prime}\int_{t_{\min}}^{t_{\max}}\int_{t_{\min}}^{t_{\max}}e^{-i(n_1-n_1^\prime)E_1(t_2-t_1)}\cdots e^{-i(n_d-n_d^\prime)E_d(t_2-t_1)}dt_1dt_2\\
=&\sum_{n_1,\cdots,n_d}\sum_{n_1^\prime,\cdots,n_d^\prime}\binom{k}{n_1,\cdots,n_d}\binom{k}{n_1^\prime,\cdots,n_d^\prime}\abs{\frac{\left(e^{i\left(\sum_{j=1}^d(n_j^\prime-n_j)E_j\right)t_{\max}}-e^{i\left(\sum_{j=1}^d(n_j^\prime-n_j)E_j\right)t_{\min}}\right)}{i(t_{\max}-t_{\min})\sum_{j=1}^d(n_j^\prime-n_j)E_j}}^2,
\end{aligned}
\end{equation}
where we use $\sum_{n_1,\cdots,n_d}$ to simplify the notation of $\sum_{n_1,\cdots,n_d\in\mathbb{N}, n_1+\cdots+n_d=k}$. Lets consider a practical Hamiltonian $H=\sum_iH_i$, where each $H_i$ has small locality. Normally speaking, eigenvalues of the Hamiltonian scale polynomially while the number of eigenvalues scales exponentially with the qubit number. If we fix the time scale $\Delta t=t_{\max}-t_{\min}$ and $k$, the value of $\binom{k}{n_1,\cdots,n_d}\binom{k}{n_1^\prime,\cdots,n_d^\prime}\abs{\frac{\left(e^{i\left(\sum_{j=1}^d(n_j^\prime-n_j)E_j\right)t_{\max}}-e^{i\left(\sum_{j=1}^d(n_j^\prime-n_j)E_j\right)t_{\min}}\right)}{i(t_{\max}-t_{\min})\sum_{j=1}^d(n_j^\prime-n_j)E_j}}^2$ decays only polynomially with qubit number. However, the number of summation terms contributing to the $k$-th order frame potential scales asymptotically $2^{2kN}$ with qubit number. Compared with $2^{kN}$ terms left in the frame potential of random diagonal unitary, Eq.~\eqref{eq:frame_potential_RDU}, the frame potential of finite time will be exponentially larger than the ideal random diagonal unitary.

Taking the Rydberg atom array Hamiltonian as an example,
\begin{equation}
H=\frac{\Omega}{2}\sum_j\left(e^{i\phi}\ketbra{g_j}{r_j}+h.c.\right)-\Delta \sum_j \hat{n}_j+\sum_{j<k}V_{jk}\hat{n}_j\hat{n}_k,
\end{equation}
we follow the same parameter setting as the main context and fix the time scale $\Delta t=20\mu s$. The evolution unitary can be decomposed as $e^{-iHt}=V_H\Lambda_t V_H^\dagger$ and $\Lambda_t$ follows a finite-time random diagonal unitary distribution. We numerically calculate the scaling of frame potential of $\Lambda_t$ with the qubit number, shown in Fig.~\ref{fig:FP}. It is obvious that the finite-time frame potential is exponentially larger than the frame potential of ideal random diagonal unitary. Besides, the scaling slope becomes larger when $k$ increases, which matches our prediction. However, as shown in main context, the Hamiltonian shadow using the Rydberg atom array Hamiltonian with $\Delta t\le 20\mu s$ performs well in estimating observables. Thus, it is still an open problem to determine the relation between the frame potential and performance of Hamiltonian shadow.

\begin{figure}
    \centering
    \includegraphics[width=0.4\textwidth]{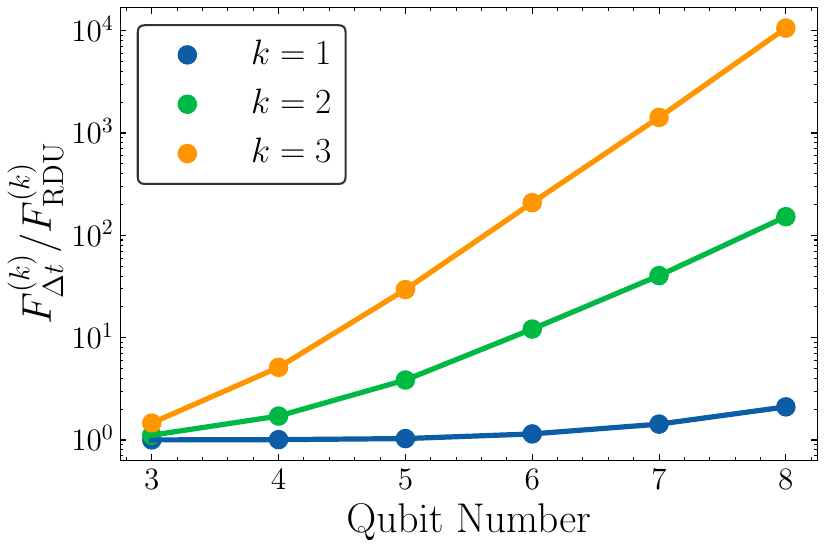}
    \caption{The scaling of finite-time frame potential with qubit number, with different $k$.}
    \label{fig:FP}
\end{figure}

\section{Case Studies}\label{sec:case_study}
\subsection{Single-Qubit Case}\label{sec:single_qubit}

In this section, we use the simplest example, the single-qubit case, to show how Hamiltonian shadow works and give the intuition of the factors influencing its performance. 
As $e^{-iHt}\rho e^{iHt}=V_H\Lambda_tV_H^\dagger\rho V_H\overline{\Lambda}_tV_H^\dagger$ and $V_H$ is a fixed unitary which is independent of $t$, we can define $\rho_H=V_H^\dagger\rho V_H$ and regard the whole process as $\rho_H$ under the evolution of $V_H\Lambda_t$. 
The complete learning of $\rho$ is equivalent with the complete learning of $\rho_H$. 

Considering a single-qubit quantum state $\rho_H = \begin{bmatrix}\rho_{00}&\rho_{01} \\ \rho_{10}&\rho_{11} \end{bmatrix}$, we evolve it using a random diagonal unitary $\Lambda_t=\mathrm{diag}(e^{i\theta_1},e^{i\theta_2})$ followed by a fixed unitary $V_H = \begin{bmatrix}\cos{\phi}&\sin{\phi} \\ -\sin{\phi}&\cos{\phi} \end{bmatrix}=\begin{bmatrix} c & s  \\ -s & c \end{bmatrix}$. After such evolution, the density matrix will become
\begin{equation}
\begin{aligned}
V_H\Lambda_t\rho_H\overline{\Lambda}_tV_H^\dagger&=\begin{bmatrix} c e^{i\theta_1} & s e^{i\theta_2} \\ -s e^{i\theta_1} & c e^{i\theta_2} \end{bmatrix} \rho \begin{bmatrix} c e^{-i\theta_1} & -s e^{-i\theta_1}  \\ s e^{-i\theta_2} & c e^{-i\theta_2} \end{bmatrix}\\
&=\begin{bmatrix} c^2\rho_{00}+s^2\rho_{11}+cse^{i(\theta_2-\theta_1)}\rho_{10}+cse^{i(\theta_1-\theta_2)}\rho_{01} & cs(\rho_{11}-\rho_{00})+c^2e^{i(\theta_1-\theta_2)}\rho_{01}-s^2e^{i(\theta_2-\theta_1)}\rho_{10}  \\ cs(\rho_{11}-\rho_{00})+c^2e^{i(\theta_2-\theta_1)}\rho_{10}-s^2e^{i(\theta_1-\theta_2)}\rho_{01} & s^2\rho_{00}+c^2\rho_{11}-cse^{i(\theta_2-\theta_1)}\rho_{10}-cse^{i(\theta_1-\theta_2)}\rho_{01} \end{bmatrix}.
\end{aligned}
\end{equation}
Assuming $\rho_{01}=a+ib$, diagonal terms of $V_H\Lambda_t\rho\overline{\Lambda}_tV_H^\dagger$ are $c^2\rho_{00}+s^2\rho_{11}+2cs\cos{(\theta_2-\theta_1)}a+2cs\sin{(\theta_2-\theta_1)}b$ and $c^2\rho_{00}+s^2\rho_{11}-2cs\cos{(\theta_2-\theta_1)}a-2cs\sin{(\theta_2-\theta_1)}b$, respectively. After measuring the evolved state in the Pauli-$Z$ basis for different values of $\theta_1$ and $\theta_2$, we get many equations to solve $\rho_{00}$, $\rho_{11}$, $a$, and $b$. If these equations are independent and complete, we can use them to learn $\rho$ completely.  

From this case study, it can be noticed that the Hamiltonian shadow does not work in some special cases, depending on the form of $V_H$ and $\Lambda_t$. When $c=s=\frac{1}{\sqrt{2}}$, diagonal terms become $\frac{1}{2}+a\cos{(\theta_2-\theta_1)}+b\sin{(\theta_2-\theta_1)}$ and $\frac{1}{2}-a\cos{(\theta_2-\theta_1)}-b\sin{(\theta_2-\theta_1)}$. In this scenario, the extraction of diagonal terms, $\rho_{00}$ and $\rho_{11}$, becomes infeasible. 
When $c=0$ or $s=0$, these two diagonal terms will be independent with $a$ and $b$, making the learning of off-diagonal terms infeasible. 
Supposing $\theta_1=E_1t$ and $\theta_2=E_2t$, when $E_1=E_2$, diagonal elements of evolved state become $c^2\rho_{00}+s^2\rho_{11}+2csa$ and $c^2\rho_{00}+s^2\rho_{11}-2csa$. These two terms are independent of time $t$ and contain three unknown parameters of $\rho$. Thus, measuring the evolved state in computational basis cannot help us to determine these unknown parameters. While, when $E_1$ and $E_2$ have correlation, like $E_1=-E_2$, the Hamiltonian shadow also works. This is because we can also adjust $\theta_1-\theta_2=(E_1-E_2)t$ to arbitrary value by adjusting the evolution time $t$ and get many independent equations.

These properties can also be derived from the theory constructed in Sec.~\ref{sec:shadow_map}. By definition, the matrix $X_2$ has the form of 
\begin{equation}
X_2 = (V^\dagger\ketbra{0}{0}V)^{\otimes 2}+(V^\dagger\ketbra{1}{1}V)^{\otimes 2} = \begin{bmatrix}
c^4 & c^3s & c^3s & c^2s^2\\
c^3s & c^2s^2 & c^2s^2 & cs^3\\
c^3s & c^2s^2 & c^2s^2 & cs^3\\
c^2s^2 & cs^3 & cs^3 & s^4
\end{bmatrix}
+
\begin{bmatrix}
s^4 & -s^3c & -s^3c & s^2c^2\\
-s^3c & s^2c^2 & s^2c^2 & -sc^3\\
-s^3c & s^2c^2 & s^2c^2 & -sc^3\\
s^2c^2 & -sc^3 & -sc^3 & c^4
\end{bmatrix}.
\end{equation}
Thus, the matrix $X_H$ is $X_H = \begin{bmatrix}
    c^4+s^4 & 2s^2c^2\\ 2s^2c^2 & c^4+s^4
\end{bmatrix}$. In most cases, $X_H$ is invertible and its off-diagonal terms of are nonzero. When $c=s=\frac{1}{\sqrt{2}}$, $X_H = \frac{1}{2}\begin{bmatrix}
    1 & 1\\ 1 & 1
\end{bmatrix}$ is not invertible, which means that the diagonal terms of $\rho_H$ cannot be derived from the Hamiltonian shadow protocol. When $c=0$ or $s=0$, $(X_H)_{i,j}=0$ for $i\neq j$, which means that the Hamiltonian shadow protocol fails to extract off-diagonal elements of $\rho_H$. 

In asymptotic scenario where $c$ or $s$ approaches zero, the diagonal elements of evolved state, $V\Lambda\rho_H\overline{\Lambda}V^\dagger$ are $c^2\rho_{00}+s^2\rho_{11}+2cs\cos{(\theta_2-\theta_1)}a-2cs\sin{(\theta_2-\theta_1)}b$, contains little information of $a$ and $b$. Therefore, although it is still possible to learn off-diagonal elements, the sample complexity will be much higher. Similarly, when $c$ and $s$ approach $\frac{1}{\sqrt{2}}$, it is hard for Hamiltonian shadow to learn diagonal terms of $\rho_H$. As a result, in addition to the detection feasibility, the form of $V_H$ also influences the performance of Hamiltonian shadow. 

We can use the single-qubit case to numerically observe the sample complexity performance of Hamiltonian shadow. Suppose the Hamiltonian we use in Hamiltonian shadow is $H(\theta)=\cos(\theta)Z+\sin(\theta)X$, where $Z$ and $X$ are single-qubit Pauli matrices. We initialize the state to be $\rho=\ketbra{\psi}{\psi}$ with $\ket{\psi}$ being a random pure state. When using this Hamiltonian to estimate the expectation value of $X+Y+Z$, the variance of Hamiltonian shadow scales as Fig.~\ref{fig:single_variance}. Three peaks in the diagram can be explained using the discussions in this section. When $\theta$ approaches $\pi/2$, $V_H$ approaches the Hadamard gate and the post-processing matrix $X_H$ approaches $\frac{1}{2}\begin{bmatrix}
    1 & 1\\ 1 & 1
\end{bmatrix}$, which is invertible and the Hamiltonian shadow fails to be tomography-complete. When $\theta$ approaches zero and $\pi$, $V_H$ and $X_H$ both approach the identity matrix, which makes it impossible to estimate off-diagonal terms of $\rho$. This results in two peaks at $\theta=0$ and $\pi$. When choosing an appropriate value of $\theta$, the variance of Hamiltonian shadow is close to the variance of the original shadow using random Clifford unitaries, which shows the potential of Hamiltonian shadow.

\begin{figure}[htbp]
    \centering
    \includegraphics[width=0.4\textwidth]{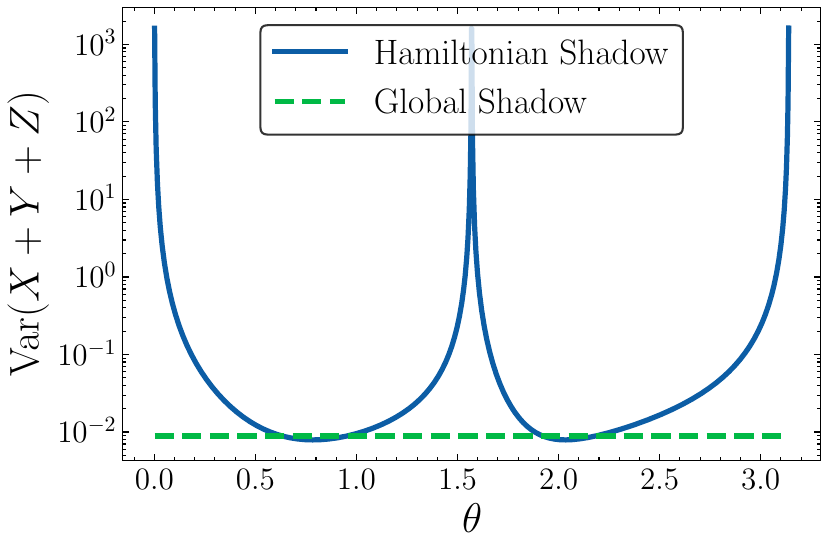}
    \caption{The variance performance of estimating $X+Y+Z$ using the Hamiltonian shadow with $H=\cos(\theta) Z+\sin(\theta)X$ and the original shadow. The measurement times is set to be $K=1000$ and the target state is a random pure state.}
    \label{fig:single_variance}
\end{figure}

\subsection{Multiqubit Hadamard Gates}\label{sec:hadamard}
In this section, we use another example to give the evidence that the performance of the Hamiltonian shadow is similar with the original shadow. The protocol is shown in Fig.~\ref{fig:Hadamard}, where the state $\rho$ evolves with a random diagonal unitary followed by a layer of Hadamard gates. Here $\rho$ can be regarded as the $\rho_H$ introduced in the previous section.

\begin{figure}[htbp]
    \centering
    \includegraphics[width=0.22\textwidth]{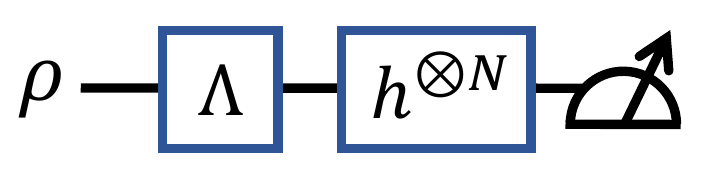}
    \caption{The Hamiltonian shadow with a random diagonal unitary followed by a layer of Hadamard gates.}
    \label{fig:Hadamard}
\end{figure}

Following the derivation in Sec.~\ref{sec:shadow_map}, the shadow map for this setting is $\mathcal{M}(\rho)=\mathcal{N}(\rho)$, as the evolving unitary is $V_H\Lambda$ instead of $V_H\Lambda V_H^\dagger$. Using the matrix form of Hadamard gate, $h=\frac{1}{\sqrt{2}}\begin{bmatrix}
    1 & 1 \\ 1 & -1
\end{bmatrix}$, every element of $X_H$ is
\begin{equation}
(X_H)_{i,j}=\sum_b|\bra{i}h^{\otimes N}\ket{b}|^2\times|\bra{j}h^{\otimes N}\ket{b}|^2=2^{-N}.
\end{equation}
Thus, the shadow map is
\begin{equation}
\mathcal{M}(\rho)=\sum_i\left(\sum_j(X_H)_{i,j}\rho_{j,j}\right)\ketbra{i}{i}+\sum_{i\neq j}X_{ij,ji}\rho_{j,i}\ketbra{j}{i}=\frac{\mathbb{I}}{2^N}+\frac{1}{2^N}\sum_{i\neq j}\rho_{j,i}\ketbra{j}{i}.
\end{equation}
This Hamiltonian shadow map losses information of diagonal terms of $\rho$ and is thus not invertible. This can also be derived from the fact that the post-processing matrix $X_H$ is not invertible. At the same time, all off-diagonal terms of $X_H$ are nonzero. Therefore, $2^N\overline{\Lambda}h^{\otimes N}\ketbra{b}{b}h^{\otimes N}\Lambda-\mathbb{I}$ gives the unbiased estimator of $\rho-\mathrm{diag}(\rho)$. This estimator can be used to estimate the expectation value of some observables with zero diagonal elements by $\hat{o}=2^n\Tr\left[\overline{\Lambda}h^{\otimes n}\ketbra{b}{b}h^{\otimes n}\Lambda O\right]$. The variance of $\hat{o}$ is 
\begin{equation}
\mathrm{Var}\left(\hat{o}\right)=2^{2N}\mathbb{E}_{\Lambda}\sum_b\Tr\left[\overline{\Lambda}h^{\otimes N}\ketbra{b}{b}h^{\otimes N}\Lambda \rho\right]\Tr\left[\overline{\Lambda}h^{\otimes N}\ketbra{b}{b}h^{\otimes N}\Lambda O\right]^2-\Tr(O\rho)^2
=2^{2N}\begin{tabular}{c}
    \includegraphics[scale=0.25]{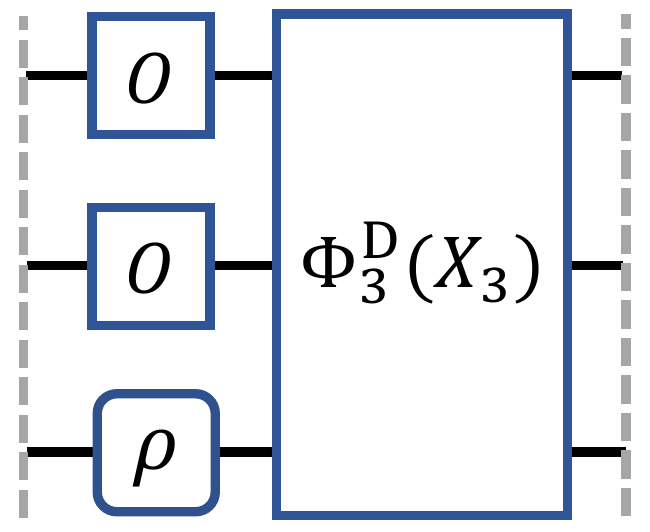}
\end{tabular}-\Tr(O\rho)^2,
\end{equation}
where $X_3 = \sum_b h^{\otimes 3N}\ketbra{bbb}{bbb}h^{\otimes 3N}$. 
To calculate the variance, we need to utilize the properties of $\Phi_3^{\mathrm{D}}(X_3)$. As shown in Fig.~\ref{fig:integral}, the third order integral has many terms. While, good thing is that, nonzero elements of all matrices are $2^{-2N}$. For example,
\begin{equation}
\left(\begin{tabular}{c}
    \includegraphics[scale=0.25]{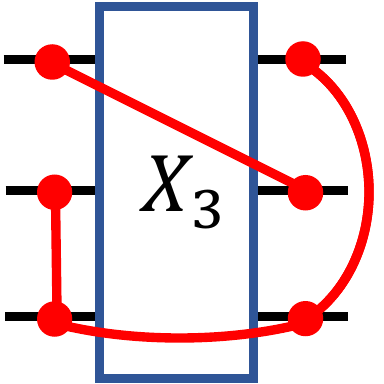}
\end{tabular}\right)_{ijj,jij}=\sum_b\bra{ijj}h^{\otimes 3N}\ket{bbb}\bra{bbb}h^{\otimes 3N}\ket{jij}=\sum_b|\bra{i}h^{\otimes N}\ket{b}|^2|\bra{j}h^{\otimes N}\ket{b}|^4=2^{-2N}.
\end{equation}
The proof for other terms are same. Therefore, with the indices contraction rule, the variance can be written as
\begin{equation}
\begin{aligned}
&\mathrm{Var}(\hat{o})=\sum_{i,j,k}\left(O_{i,i}O_{j,j}\rho_{k,k}+O_{i,i}O_{j,k}\rho_{k,j}+O_{i,k}O_{j,j}\rho_{k,i}+O_{i,j}O_{j,i}\rho_{k,k}+O_{i,j}O_{j,k}\rho_{k,i}+O_{i,k}O_{j,i}\rho_{k,j}\right)\\
-&\sum_{i,j}\left(O_{i,i}O_{i,i}\rho_{j,j}+O_{i,j}O_{j,j}\rho_{j,i}+O_{i,i}O_{j,i}\rho_{i,j}+O_{i,j}O_{i,i}\rho_{j,i}+O_{i,i}O_{j,j}\rho_{j,j}+O_{i,j}O_{j,i}\rho_{i,i}+O_{i,i}O_{i,j}\rho_{j,i}+O_{i,j}O_{j,i}\rho_{j,j}+O_{i,i}O_{j,j}\rho_{i,i}\right)\\
+&4\sum_iO_{i,i}O_{i,i}\rho_{i,i}-\Tr(O\rho)^2.
\end{aligned}
\end{equation}
As $O$ has no diagonal terms, $O_{i,i}=0$, only a few terms in the above equation survives
\begin{equation}
\begin{aligned}
\mathrm{Var}(\hat{o})=&\sum_{i,j,k}\left(O_{i,j}O_{j,i}\rho_{k,k}+O_{i,j}O_{j,k}\rho_{k,i}+O_{i,k}O_{j,i}\rho_{k,j}\right)-\sum_{i,j}\left(O_{i,j}O_{j,i}\rho_{j,j}+O_{i,j}O_{j,i}\rho_{i,i}\right)-\Tr(O\rho)^2\\
\le& \Tr(O^2)+2\Tr(O^2\rho)\le 3\Tr(O^2),
\end{aligned}
\end{equation}
where the first inequality holds  as $O_{j,i}O_{i,j}\rho_{i,i}=|O_{j,i}|^2\rho_{i,i}\ge 0$. Notice that $\Tr(O^2)+2\Tr(O^2\rho)$ is exactly the upper bound for original shadow when using global Clifford unitaries \cite{Huang2020predicting}. This result shows a strong evidence that the random diagonal shadow can have a similar performance with the original shadow protocol. Besides, the inability to estimate diagonal information of $\rho$ is not a big problem in practice, as we can directly perform computational basis measurements on $\rho$ to extract these information. The final thing we want to emphasis is that, similar with the original shadow protocol, the leading term of variance is also 
\begin{equation}
\sum_{i,j,k}(X_3)_{ijk,jik}O_{i,j}O_{j,i}\rho_{k,k}=\begin{tabular}{c}
    \includegraphics[scale=0.25]{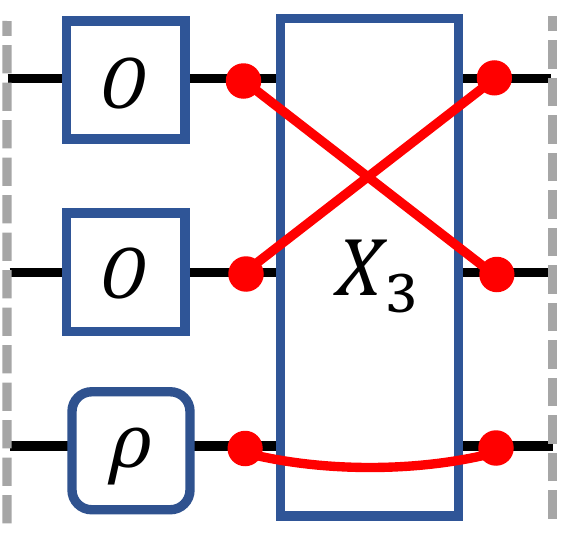}
\end{tabular},
\end{equation}
which also introduces the indices contraction between two observable matrices.

We also numerically show that the performance of Hadamard-based shadow is similar to the global shadow. In Fig.~\ref{fig:Hadamard_variance}, we adopt the original shadow with global Clifford gates and the Hadamard-based Hamiltonian shadow to estimate the Pauli observable $X^{\otimes N}$, which has no diagonal elements and satisfies the requirement of detecting with Hadamard shadow. It is shown that the variance scaling of these two protocols are highly similar.
\begin{figure}
    \centering
    \includegraphics[width=0.4\textwidth]{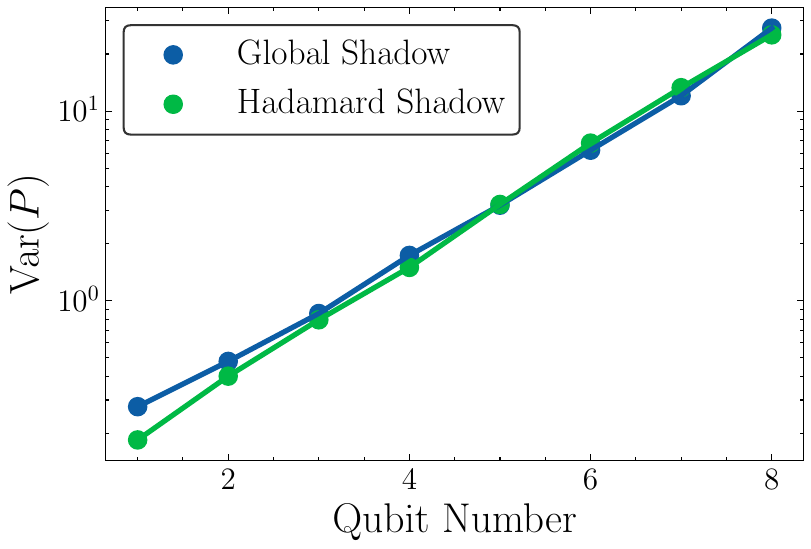}
    \caption{The comparison between the global shadow and Hadamard-based Hamiltonian shadow in estimating the expectation value of $P=X^{\otimes N}$. We set $K=10$ and the target state to be a random pure state. }
    \label{fig:Hadamard_variance}
\end{figure}

\section{Variance Analysis}\label{sec:variance}
In this section, we provide a systematic analysis of the variance of Hamiltonian shadow and use some approximation to simplify its expression. We also start from global version of Hamiltonian shadow, shown in Fig.~\ref{fig:Hamiltonian}(a), and generalize the result to local one, shown in Fig.~\ref{fig:Hamiltonian}(b).

Given the target observable $O$, the shadow estimator is
\begin{equation}\label{eq:estimator_o}
\hat{o}=\Tr\left[O V_H\mathcal{N}^{-1}(\overline{\Lambda}_tV_H^\dagger\ketbra{b}{b}V_H\Lambda_t)V_H^\dagger\right]=\Tr\left[(\mathcal{N}^{-1})^\dagger(V_H^\dagger O V_H)\overline{\Lambda}_tV_H^\dagger\ketbra{b}{b}V_H\Lambda_t\right].
\end{equation}
With the definition of $\mathcal{N}^{-1}$, we have
\begin{equation}
\Tr[A\mathcal{N}^{-1}(B)]=\sum_{i,j}A_{i,i}(X_H^{-1})_{i,j}B_{j,j}+\sum_{i\neq j}(X_H)_{i,j}^{-1}B_{i,j}A_{j,i}=\sum_{i,j}B_{i,i}(X_H^{-1})_{i,j}A_{j,j}+\sum_{i\neq j}(X_H)_{i,j}^{-1}A_{i,j}B_{j,i}=\Tr[\mathcal{N}^{-1}(A)B]
\end{equation}
for two square matrices $A$ and $B$, where the second equality holds as $X_H$ is a real symmetric matrix. This equation shows that $(\mathcal{N}^{-1})^\dagger=\mathcal{N}^{-1}$.
Combining this property and Eq.~\eqref{eq:estimator_o}, the variance can be written as
\begin{equation}\label{eq:H_variance}
\begin{aligned}
\mathrm{Var}(\hat{o})&=\mathbb{E}_{b,t}\hat{o}^2-\Tr(O\rho)^2\\
&=\mathbb{E}_t\sum_b\bra{b}e^{-iHt}\rho e^{iHt}\ket{b}\hat{o}^2-\Tr(O\rho)^2\\
&=\mathbb{E}_t\sum_b\Tr\left\{\left[\mathcal{N}^{-1}(V_H^\dagger O V_H)^{\otimes 2}\otimes V_H^\dagger\rho V_H\right]\left[\left(\overline{\Lambda}_tV_H^\dagger\right)^{\otimes 3}\ketbra{bbb}{bbb}\left(V_H\Lambda_t\right)^{\otimes 3}\right]\right\}-\Tr(O\rho)^2.
\end{aligned}
\end{equation}
Define $X_3=\sum_b\left(V_H^\dagger\ketbra{b}{b}V_H\right)^{\otimes 3}$, the leading term of variance can be graphically represented as 
\begin{equation}
\begin{aligned}
\mathbb{E}_{b,t}\hat{o}^2&=\mathbb{E}_t\left(\sum_b\begin{tabular}{c}
    \includegraphics[scale=0.25]{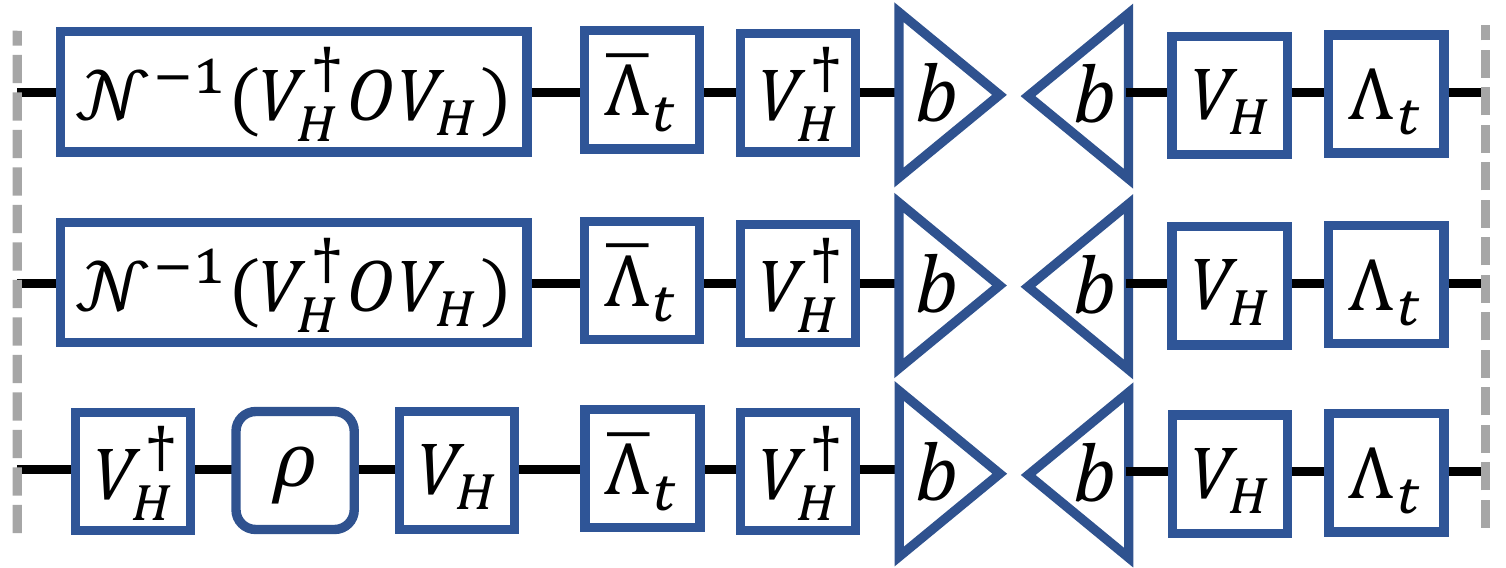}
\end{tabular}\right)
=\begin{tabular}{c}
    \includegraphics[scale=0.25]{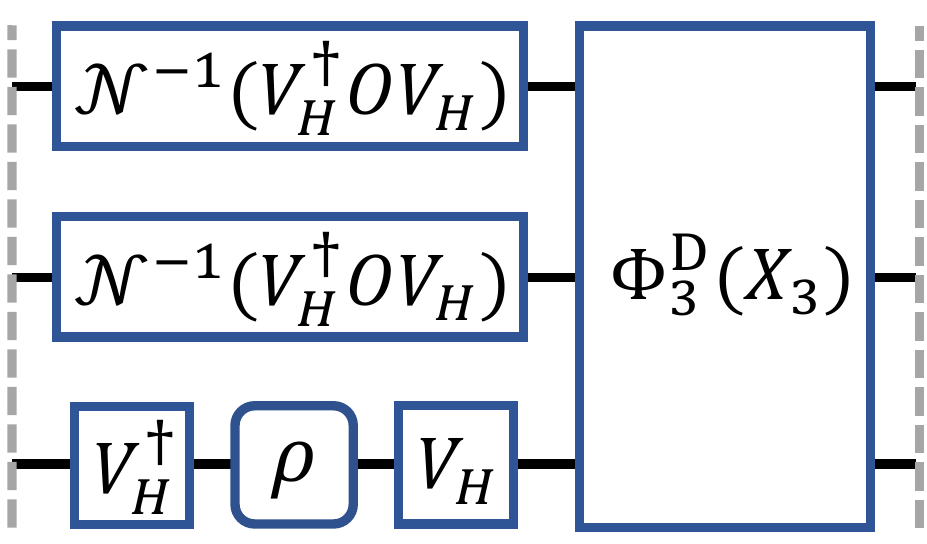}
\end{tabular}\\
&=\Tr\left\{\left[\mathcal{N}^{-1}(V_H^\dagger O V_H)^{\otimes 2}\otimes V_H^\dagger\rho V_H\right]\Phi_3^{\mathrm{D}}(X_3)\right\},
\end{aligned}
\end{equation}
where grey dashed lines represent the trace function. In addition to the observable and state, the variance of random diagonal shadow highly depends on the diagonal unitary $V_H$. As discussed in Appendix \ref{sec:single_qubit}, when $V_H$ approaches the identity matrix, it is easy for the Hamiltonian shadow to extract information of diagonal elements of $V_H^\dagger\rho V_H$ while hard for off-diagonal elements. And as shown in Appendix \ref{sec:hadamard}, when $V_H$ approaches the tensor product of Hadamard gates, it is hard to extract information of diagonal terms of $V_H^\dagger\rho V_H$. 

Similar with the original shadow protocol, the variance of the Hamiltonian shadow protocol does not necessitate that $\Lambda_t$ be a perfect random diagonal unitary; a diagonal unitary three-design is sufficient. This is because that in the derivation of the variance, we only require the use of $\Phi_3^{\mathrm{D}}(\cdot)$. Another consequence of the third-order integral is that, although the third-order degeneracy of Hamiltonian does not affect the unbiasedness of Hamiltonian shadow, it does affect the variance of it.

An important property of Hamiltonian shadow variance is that $\mathbb{E}_{b,t}\hat{o}^2=1$ when $O=\mathbb{I}$. This property can be easily derived from the fact that $\mathcal{N}$ and $\mathcal{N}^{-1}$ are both trace-preserving. We give another proof for this property. Given $O=\mathbb{I}$, the estimator can be written as
\begin{equation}
\begin{aligned}
\hat{o}=&\Tr\left[\mathcal{N}^{-1}(V_H^\dagger\mathbb{I}V_H)\overline{\Lambda}_tV_H^\dagger\ketbra{b}{b}V_H\Lambda_t\right]
=\bra{b}V_H\mathcal{N}^{-1}(\mathbb{I})V_H^\dagger\ket{b}\\
=&\sum_i\sum_j(X_H^{-1})_{ij}\bra{b}V_H\ketbra{i}{i}V_H^\dagger\ket{b}\\
=&\sum_i\sum_j\left[(V_H^{\mathrm{sq}})^{-1}((V_H^{\mathrm{sq}})^{T})^{-1}\right]_{i,j}(V_H^{\mathrm{sq}})_{b,i}\\
=&\sum_j(V_H^{\mathrm{sq}})^{-1}_{j,b},
\end{aligned}
\end{equation}
where we adopt the definition of $\mathcal{N}^{-1}$ and $X_H=(V_H^{\mathrm{sq}})^TV_H^{\mathrm{sq}}$. As $(V_H^{\mathrm{sq}})^{-1}V_H^{\mathrm{sq}}=\mathbb{I}$, we have 
\begin{equation}
\sum_j(V_H^{\mathrm{sq}})^{-1}_{j,b}=\sum_j(V_H^{\mathrm{sq}})^{-1}_{j,b}\sum_b(V_H^{\mathrm{sq}})_{b,i}=\sum_j\left(\sum_b(V_H^{\mathrm{sq}})^{-1}_{j,b}(V_H^{\mathrm{sq}})_{b,i}\right)=\sum_j(\delta_{i,j})=1,
\end{equation}
where we use the property of unitary matrix that $\sum_b(V_H^{\mathrm{sq}})_{b,i}=1$. From above derivations, we know that, when setting $O=\mathbb{I}$, the estimator and $\mathbb{E}_{t,b}\hat{o}^2$ are always $1$, no matter of the measurement result $b$. Thus, the corresponding variance is zero. This is an important property for local version of Hamiltonian shadow. Suppose we use the evolution of $\bigotimes_{p=1}^Ne^{-iH_pt}$ to perform Hamiltonian shadow on $\rho$, shown in Fig.~\ref{fig:Hamiltonian}, the variance of estimating $O=\bigotimes_{p=1}^NO_p$ can be easily generalized from the result of global version Hamiltonian shadow, 
\begin{equation}
\mathbb{E}_{b,t}\hat{o}^2=\prod_{p=1}^N\mathbb{E}_{t,b_p}\hat{o}_p^2.
\end{equation}
Using the property of $\mathbb{E}_{b,t}\hat{o}^2=1$ for $O=\mathbb{I}$, the variance of estimating some local observable does not depend on the total qubit number of $\rho$, but only the locality of $O$.

\subsection{Hamiltonian Shadow Norm}
To benefit our description of the sample complexity of Hamiltonian shadow, we define the Hamiltonian shadow norm of an observable as
\begin{equation}
\begin{aligned}
\norm{O}_{\mathrm{HShadow}}=\max_\sigma\left(\mathbb{E}_t\sum_b\bra{b}e^{-iHt}\sigma e^{iHt}\ket{b}\hat{o}^2\right)^{1/2}=&\max_\sigma\left(\mathbb{E}_t\sum_b\bra{b}e^{-iHt}\sigma e^{iHt}\ket{b}\bra{b}V_H\Lambda_t\mathcal{N}^{-1}(V_H^\dagger O V_H)\overline{\Lambda}_tV_H^\dagger\ket{b}^2\right)^{1/2}\\
=&\max_\sigma\left(\begin{tabular}{c}
    \includegraphics[scale=0.25]{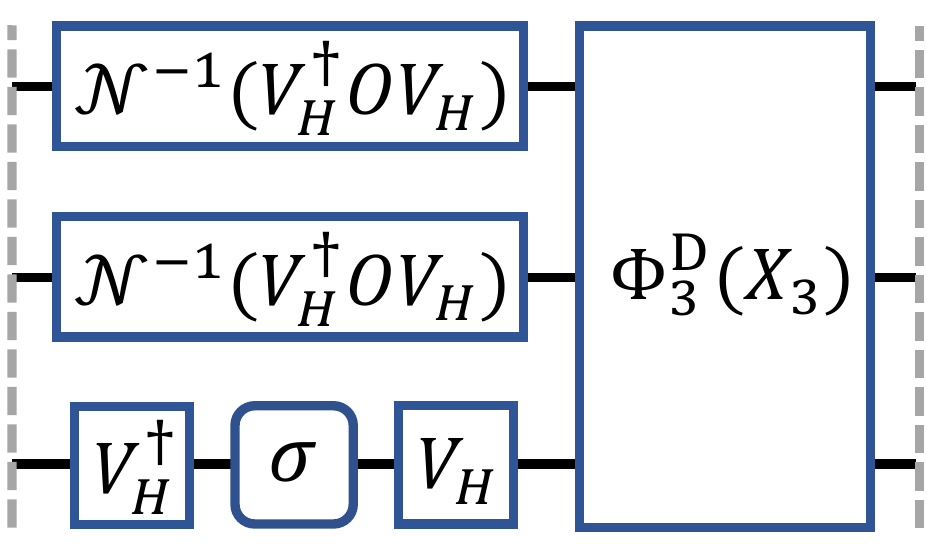}
\end{tabular}\right)^{1/2}.
\end{aligned}
\end{equation}
It is easy to prove that the Hamiltonian shadow norm is non-negative and homogeneous, $\norm{0}_{\mathrm{HShadow}}=0$. We need to prove the triangle inequality of Hamiltonian shadow norm. Denoting $x_b=\bra{b}e^{-iHt}\sigma e^{iHt}\ket{b}^{1/2}\bra{b}V_H\Lambda_t\mathcal{N}^{-1}(V_H^\dagger O_1 V_H)\overline{\Lambda}_tV_H^\dagger\ket{b}$, and $y_b=\bra{b}e^{-iHt}\sigma e^{iHt}\ket{b}^{1/2}\bra{b}V_H\Lambda_t\mathcal{N}^{-1}(V_H^\dagger O_2 V_H)\overline{\Lambda}_tV_H^\dagger\ket{b}$, the Hamiltonian shadow norm of $O_1+O_2$ is thus
\begin{equation}
\begin{aligned}
\norm{O_1+O_2}_{\mathrm{HShadow}}=\max_{\sigma}\left(\mathbb{E}_t\sum_b(x_b+y_b)^2\right)^{1/2}.
\end{aligned}
\end{equation}
According to the triangle inequality of coordinate norm, we have
\begin{equation}
\left(\sum_b(x_b+y_b)^2\right)^{1/2}=\abs{\vec{x}+\vec{y}}\le\abs{\vec{x}}+\abs{\vec{y}}=\left(\sum_bx_b^2\right)^{1/2}+\left(\sum_by_b^2\right)^{1/2}.
\end{equation}
According to the property of root function, one can easily prove that $\left(\mathbb{E}_t\sum_b(x_b+y_b)^2\right)^{1/2}\le\left(\mathbb{E}_t\sum_bx_b^2\right)^{1/2}+\left(\mathbb{E}_t\sum_by_b^2\right)^{1/2}$. Therefore, 
\begin{equation}
\begin{aligned}
\norm{O_1+O_2}_{\mathrm{HShadow}}\le&\max_{\sigma}\left[\left(\mathbb{E}_t\sum_bx_b^2\right)^{1/2}+\left(\mathbb{E}_t\sum_by_b^2\right)^{1/2}\right]\\
\le&\max_{\sigma}\left(\mathbb{E}_t\sum_bx_b^2\right)^{1/2}+\max_{\sigma}\left(\mathbb{E}_t\sum_bx_b^2\right)^{1/2}\\
=&\norm{O_1}_{\mathrm{HShadow}}+\norm{O_2}_{\mathrm{HShadow}},
\end{aligned}
\end{equation}
which concludes the triangle inequality. 

According to Eq.~\eqref{eq:H_variance}, when estimating $\tr(O\rho)$ with Hamiltonian shadow, the variance is upper bounded by $\norm{O}_{\mathrm{HShadow}}^2$.
Therefore, adopting the median-of-mean data processing method \cite{Huang2020predicting}, we can bound the sample complexity $K$ of estimating $M$ observables using the Hamiltonian shadow norm    
\begin{equation}
K=\mathcal{O}\left(\mathrm{log}(M)\frac{\max_i\norm{O_i}_{\mathrm{HShadow}}^2}{\epsilon^2}\right),
\end{equation}
where $\epsilon$ is the additive error.

\subsection{Variance Approximation}

It is hard to derive a simple expression for the variance of Hamiltonian shadow. While, we can use some reasonable approximations to give a function that can largely reflect the scaling of variance. According to derivations of original shadow protocol and the Hadamard-based diagonal shadow, we can reasonably assume that the leading term of Hamiltonian shadow variance is
\begin{equation}
\begin{tabular}{c}
    \includegraphics[scale=0.25]{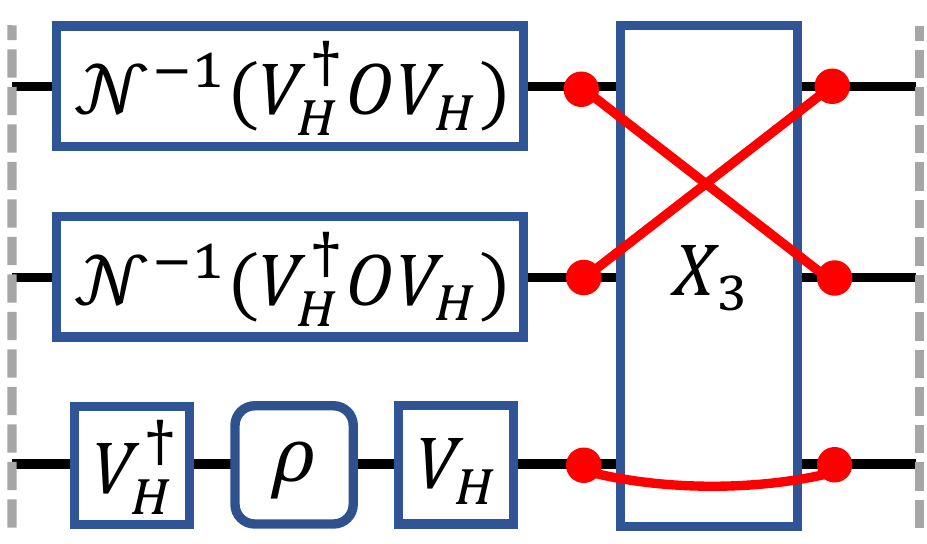}
\end{tabular}
=\Tr\left\{\left(\mathcal{N}^{-1}\otimes\mathcal{N}^{-1}\otimes\mathcal{I}\right)\left[\sum_{i,j,k}(X_3)_{ijk,jik}\ketbra{ijk}{jik}\right]\left[(V_H^\dagger OV_H)^{\otimes 2}\otimes V_H^\dagger \rho V_H\right]\right\},
\end{equation}
where $\mathcal{I}$ represents the identity map. This term is the leading term because it introduces the indices contraction between two observables $O$ while traces over the density matrix $\rho$. Similar terms in variances of the original shadow protocol and the Hadamard diagonal shadow protocol give the term of $\Tr(O^2)$, which is the leading term of those two variances. Now, we will try to simplify this term.

We first divide this term into two parts
\begin{equation}\label{eq:leading_terms}
\begin{aligned}
&
\Tr\left\{\left(\mathcal{N}^{-1}\otimes\mathcal{N}^{-1}\otimes\mathcal{I}\right)\left[\sum_{i,j,k}(X_3)_{ijk,jik}\ketbra{ijk}{jik}\right]\left[(V_H^\dagger OV_H)^{\otimes 2}\otimes V_H^\dagger \rho V_H\right]\right\}\\
=&\sum_{i\neq j}\sum_k(X_3)_{ijk,jik}\Tr\left[\mathcal{N}^{-1}\left(\ketbra{i}{j}\right)V_H^\dagger OV_H\right]\Tr\left[\mathcal{N}^{-1}\left(\ketbra{j}{i}\right)V_H^\dagger OV_H\right]\bra{k}V_H^\dagger \rho V_H\ket{k}\\
+&\sum_{i}\sum_k(X_3)_{iik,iik}\Tr\left[\mathcal{N}^{-1}(\ketbra{i}{i})V_H^\dagger OV_H\right]^2\bra{k}V_H^\dagger \rho V_H\ket{k}.
\end{aligned}
\end{equation}
The first part contains exponentially more terms and is generally exponentially larger than the second part. Thus, we only focus on the first part and expand it. By definition, the element of $X_3$ is
\begin{equation}
(X_3)_{ijk,jik}=\sum_b\bra{ijk}(V_H^\dagger)^{\otimes 3}\ketbra{bbb}{bbb}V_H^{\otimes 3}\ket{jik}=\sum_b|(V_H)_{b,i}|^2|(V_H)_{b,j}|^2|(V_H)_{b,k}|^2.
\end{equation}
Substituting it into the Eq.~\eqref{eq:leading_terms}, we find the first part to be
\begin{equation}
\begin{aligned}
&\sum_{i\neq j}\sum_k\sum_b|(V_H)_{b,i}|^2|(V_H)_{b,j}|^2|(V_H)_{b,k}|^2 (X_H)_{i,j}^{-2}\bra{j}V_H^\dagger OV_H\ket{i}\bra{i}V_H^\dagger OV_H\ket{j}\bra{k}V_H^\dagger\rho V_H\ket{k}\\
=&\sum_{i\neq j}\sum_b|(V_H)_{b,i}|^2|(V_H)_{b,j}|^2 (X_H)_{i,j}^{-2}\bra{j}V_H^\dagger OV_H\ket{i}\bra{i}V_H^\dagger OV_H\ket{j}\sum_k|(V_H)_{b,k}|^2\bra{k}V_H^\dagger\rho V_H\ket{k}.
\end{aligned}
\end{equation}
Here, $\bra{k}V_H^\dagger\rho V_H\ket{k}$ is the diagonal term of $V_H^\dagger\rho V_H$ and $\rho$ is the target quantum state which normally has no relationship with $V_H$. Thus, it is reasonable to assume $\bra{k}V_H^\dagger\rho V_H\ket{k}\approx\frac{1}{d}$ for all $k$. With this approximation, we have $\sum_k|(V_H)_{b,k}|^2\bra{k}V_H^\dagger\rho V_H\ket{k}=\frac{1}{d}$. Then, the first part becomes
\begin{equation}\label{eq:variance_leading}
\begin{aligned}
&\sum_{i\neq j}\sum_b|(V_H)_{b,i}|^2|(V_H)_{b,j}|^2 (X_H)_{i,j}^{-2}\bra{j}V_H^\dagger OV_H\ket{i}\bra{i}V_H^\dagger OV_H\ket{j}\sum_k|(V_H)_{b,k}|^2\bra{k}V_H^\dagger\rho V_H\ket{k}\\
\approx&\frac{1}{d}\sum_{i\neq j}\sum_b|(V_H)_{b,i}|^2|(V_H)_{b,j}|^2 (X_H)_{i,j}^{-2}\bra{j}V_H^\dagger OV_H\ket{i}\bra{i}V_H^\dagger OV_H\ket{j}\\
=&\frac{1}{d}\sum_{i\neq j}(X_H)_{i,j} (X_H)_{i,j}^{-2}\bra{j}V_H^\dagger OV_H\ket{i}\bra{i}V_H^\dagger OV_H\ket{j}\\
=&\frac{1}{d}\sum_{i\neq j}(X_H)_{i,j}^{-1}|\bra{i}V_H^\dagger OV_H\ket{j}|^2,
\end{aligned}
\end{equation}
which is the function of $f(O,V_H)$ we introduced in the main context.

\subsection{Nonlinear Observables}

It is straightforward to use the Hamiltonian shadow to estimate nonlinear quantities, like $\Tr(O\rho^{\otimes 2})$. After $K$ times of experiments, we get a total of $K$ unbiased estimators of $\rho$, labeled as $\{\hat{\rho}_i\}_{i=1}^K$. Using these estimators, we can construct the unbiased estimator of $\Tr(O\rho^{\otimes 2})$ as
\begin{equation}
\hat{o}_2=\frac{1}{K(K-1)}\sum_{i\neq j}\Tr\left[O\left(\hat{\rho}_i\otimes\hat{\rho}_j\right)\right].
\end{equation}
The variance is
\begin{equation}
\mathrm{Var}(\hat{o}_2)=\frac{1}{K(K-1)}\mathrm{Var}\left(\Tr\left[O\left(\hat{\rho}\otimes\hat{\rho}^\prime\right)\right]\right)=\frac{1}{K(K-1)}\left\{\mathbb{E}\Tr\left[O\left(\hat{\rho}\otimes\hat{\rho}^\prime\right)\right]^2-\Tr(O\rho^{\otimes 2})^2\right\},
\end{equation}
where $\hat{\rho}$ and $\hat{\rho}^\prime$ represent two independent estimators constructed using Hamiltonian shadow.

When $O=S$, $\tr(O\rho^{\otimes 2})$ gives the value of purity, $\tr(\rho^2)$, which is important for our numerical demonstration in main context. So, we start from $O=S$ and derive the approximate variance of it, which can be easily extended to general nonlinear observables. Substituting the form of Hamiltonian shadow estimator, we have
\begin{equation}
\mathbb{E}\Tr\left[S\left(\hat{\rho}\otimes\hat{\rho}^\prime\right)\right]^2=\mathbb{E}_{t,t^\prime}\sum_{b,b^\prime}\bra{b}V_H\Lambda_t V_H^\dagger\rho V_H\overline{\Lambda}_tV_H^\dagger\ket{b}\bra{b^\prime}V_H\Lambda_{t^\prime} V_H^\dagger\rho V_H\overline{\Lambda}_{t^\prime}V_H^\dagger\ket{b^\prime}\Tr\left[S\left(\hat{\rho}\otimes\hat{\rho}^\prime\right)\right]^2,
\end{equation}
which can be calculated in a graphical way
\begin{equation}
\begin{aligned}
\mathbb{E}\Tr\left[S\left(\hat{\rho}\otimes\hat{\rho}^\prime\right)\right]^2=&\mathbb{E}_{t,t^\prime}\left[\sum_{b,b^\prime}\begin{tabular}{c}
    \includegraphics[scale=0.25]{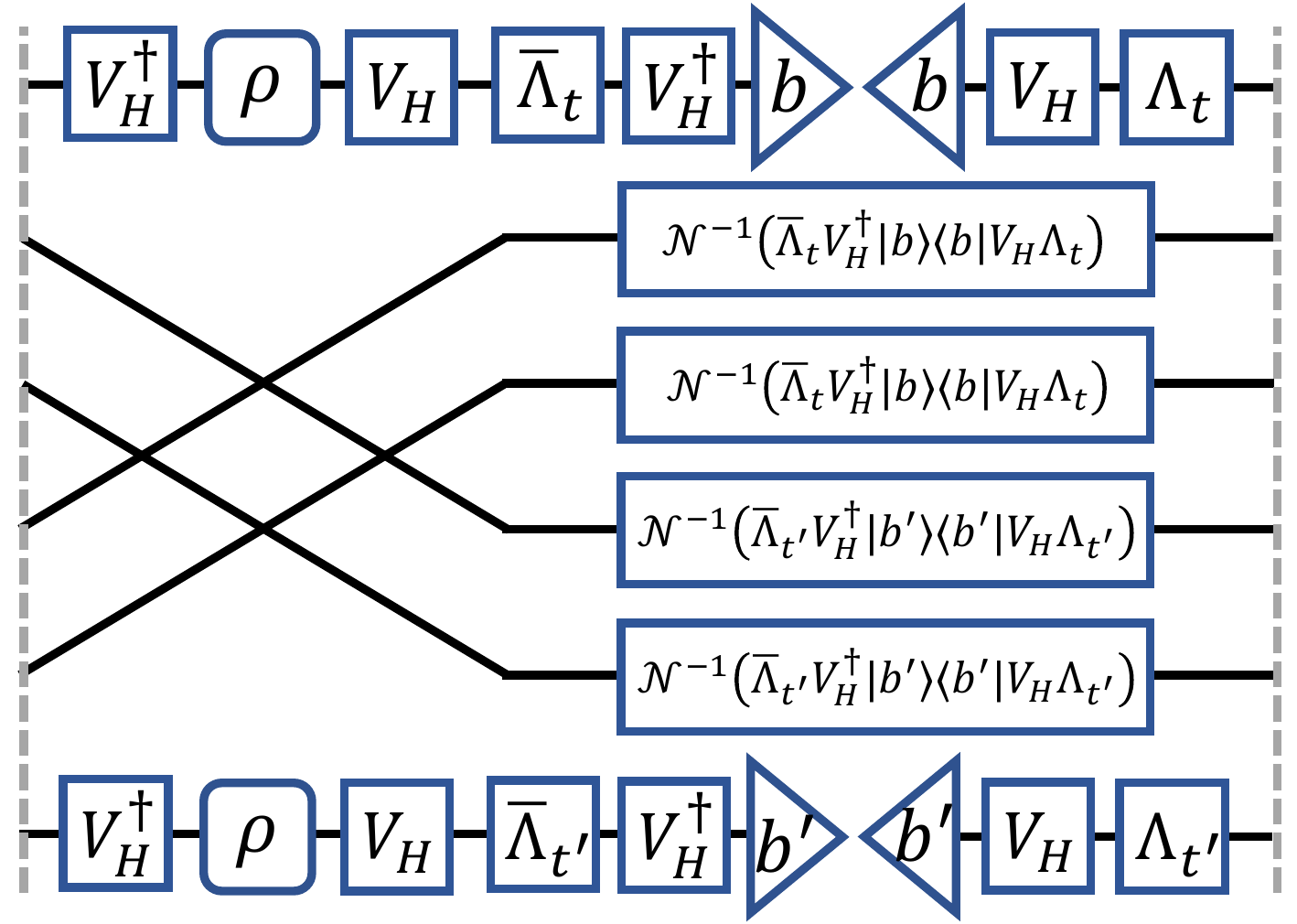}
\end{tabular}\right]
=\begin{tabular}{c}
    \includegraphics[scale=0.25]{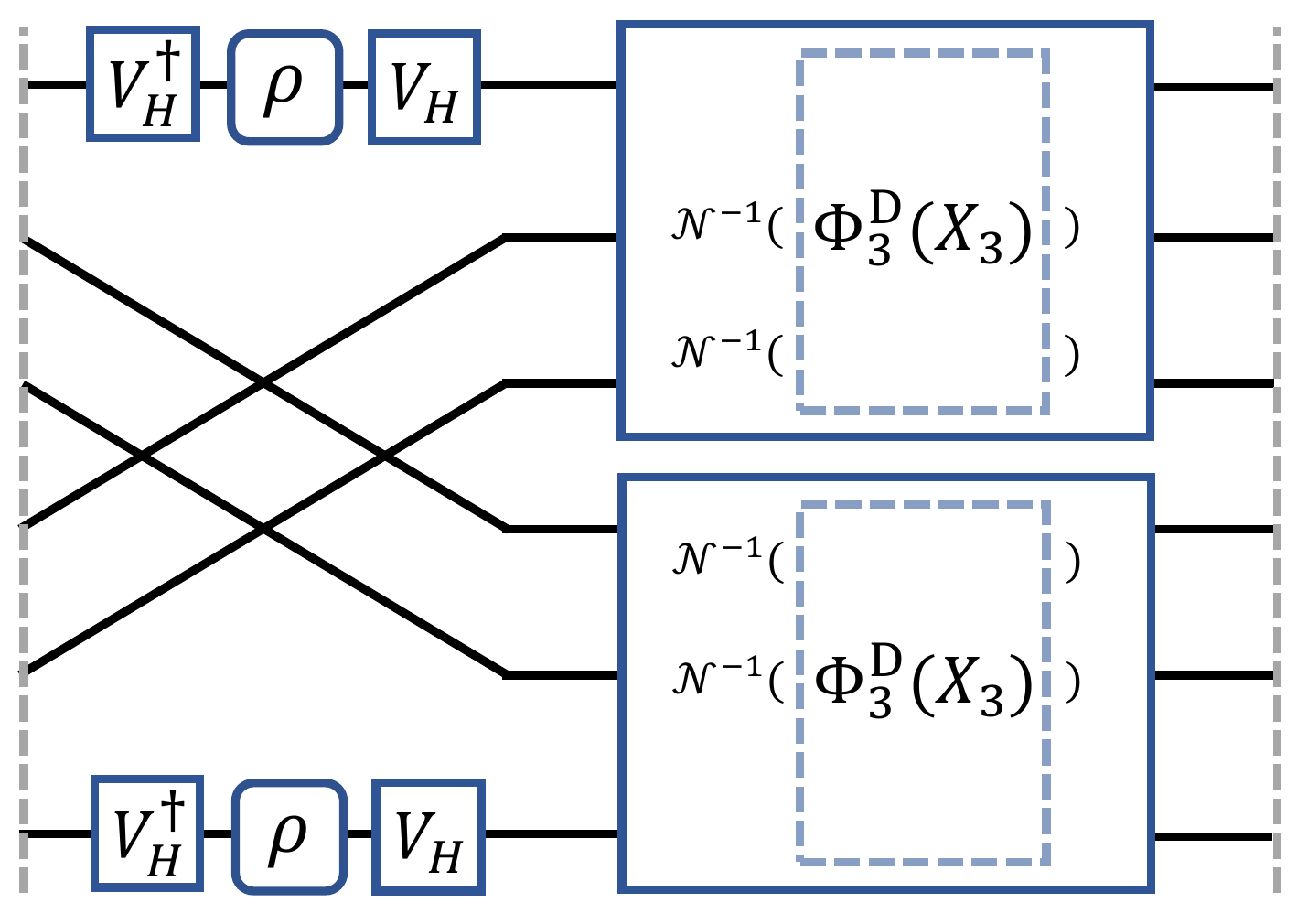}
\end{tabular}\\
=&\Tr\left\{\Tr_3\left\{\mathcal{N}^{-1}\otimes\mathcal{N}^{-1}\otimes\mathcal{I}\left[\Phi_3^{\mathrm{D}}(X_3)\right] \left(\mathbb{I}^{\otimes 2}\otimes V_H^\dagger\rho V_H\right)\right\}^2\right\}.
\end{aligned}
\end{equation}
Adopting the same approximation for linear observables, the matrix $\Tr_3\left\{\mathcal{N}^{-1}\otimes\mathcal{N}^{-1}\otimes\mathcal{I}\left[\Phi_3^{\mathrm{D}}(X_3)\right] \left(\mathbb{I}^{\otimes 2}\otimes V_H^\dagger\rho V_H\right)\right\}$ can be approximated by $\frac{1}{d}\sum_{i\neq j}\frac{1}{(X_H)_{i,j}}\ketbra{ij}{ji}$. Substituting this approximation, we have
\begin{equation}
\begin{aligned}
\mathbb{E}\Tr\left[S\left(\hat{\rho}\otimes\hat{\rho}^\prime\right)\right]^2\sim&
\begin{tabular}{c}
    \includegraphics[scale=0.25]{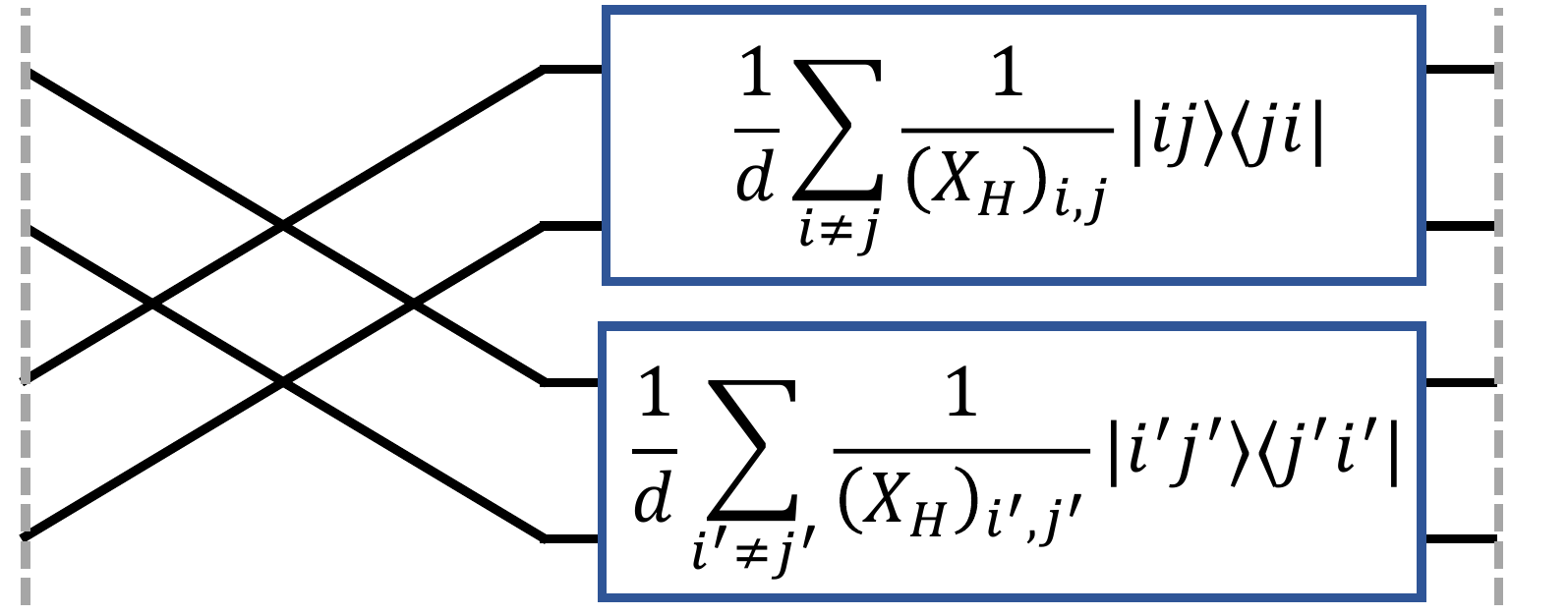}
\end{tabular}\\
=&\frac{1}{d^2}\sum_{i\neq j}\sum_{i^\prime\neq j^\prime}\frac{1}{(X_H)_{i,j}(X_H)_{i^\prime,j^\prime}}\Tr(S\ketbra{ii^\prime}{jj^\prime})\Tr(S\ketbra{jj^\prime}{ii^\prime})\\
=&\frac{1}{d^2}\sum_{i\neq j}\frac{1}{(X_H)_{i,j}^2}.
\end{aligned}
\end{equation}
For a general nonlinear observable $O$, the exact and approximate variances can be derived in a similar way,
\begin{equation}
\begin{aligned}
\mathbb{E}\Tr\left[O\left(\hat{\rho}\otimes\hat{\rho}^\prime\right)\right]^2=
&\Tr\left\{V_H^{\dagger\otimes 4}\left(S_{23}O^{\otimes 2}S_{23}\right)V_H^{\otimes 4}\Tr_3\left\{\mathcal{N}^{-1}\otimes\mathcal{N}^{-1}\otimes\mathcal{I}\left[\Phi_3^{\mathrm{D}}(X_3)\right] \left(\mathbb{I}^{\otimes 2}\otimes V_H^\dagger\rho V_H\right)\right\}^{\otimes 2}\right\}\\
\sim&\frac{1}{d^2}\sum_{i\neq j}\sum_{i^\prime\neq j^\prime}\frac{1}{(X_H)_{i,j}(X_H)_{i^\prime,j^\prime}}\abs{\bra{jj^\prime}(V_H^{\dagger})^{\otimes 2}OV_H^{\otimes 2}\ket{ii^\prime}}^2,
\end{aligned}
\end{equation}
where $S_{23}$ is the swap operator acting on the second party of the first $O$ and the first party of the second $O$.

We also use the numerical experiment to show the accuracy of our approximation, as shown in Fig.~\ref{fig:purity_theta}. Besides, it is also shown that the performance of Hamiltonian shadow can approach the global shadow with a proper Hamiltonian, even in predicting nonlinear observables.
\begin{figure}
\centering
\includegraphics[width=0.4\textwidth]{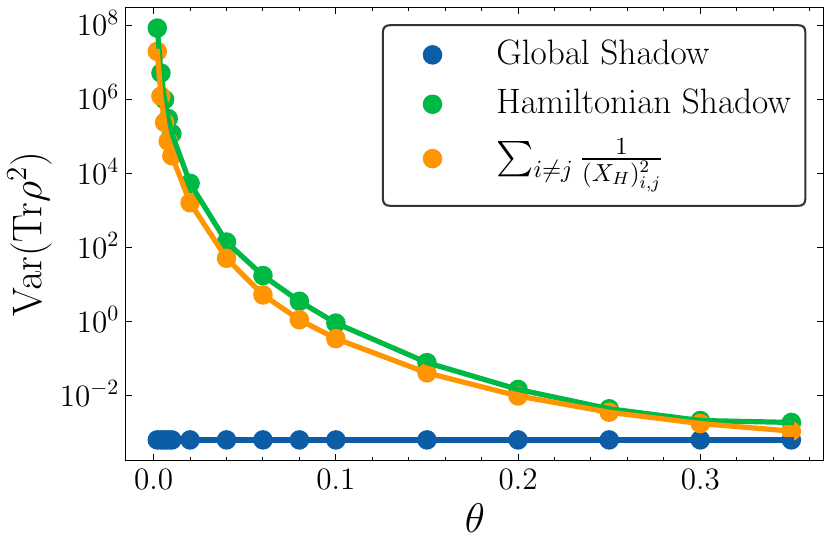}
\caption{The scaling of variance for estimating state purity. Similar with the setting in main context, we set $V_H=e^{iP\theta}$, where $P$ is a random Hermitian matrix and $\theta$ quantifies the difference between $V_H$ and a diagonal unitary. The value of every point is chosen by taking the median of ten independently sampled $P$. The target state is set to be a four-qubit GHZ state and the experiment times is set to be $K=1000$. }
\label{fig:purity_theta}
\end{figure}

\section{Detection Capability}\label{sec:capability}
After systematically demonstrating the effectiveness of Hamiltonian shadow protocol and its performances, we now discuss several scenarios where the Hamiltonian shadow protocol does not work. In Appendix \ref{sec:shadow_map}, we have shown how to determine Choi matrices of Hamiltonian shadow map in different scenarios, including those involving limited evolution time and existence of degeneracy. In principle, we can directly employ these Choi matrices and mathematically determine if they correspond to invertible maps. In this section, we analyze the reversibility of Hamiltonian shadow map from the perspective of state learning instead of linear algebra. These analysis not only help us to choose appropriate Hamiltonians for Hamiltonian shadow, but also deepen our understanding for Hamiltonian shadow protocol. We discuss four main factors that affect the reversibility of Hamiltonian shadow map, the eigenvalues and eigenstates of Hamiltonian $H$, the incomplete measurement, and the target state.

\subsection{Eigenvalues}\label{subsec:eigenvalues}
Degeneracy is the first property of real-world Hamiltonian that would affect the reversibility of Hamiltonian shadow map, which makes several eigenvalues of $e^{-iHt}$ correlated. 
We have introduced the integral rule of random diagonal unitary with degeneracy by Eq.~\eqref{eq:degeneracy} and Eq.~\eqref{eq:second_degeneracy}. 
In these scenarios, even with infinite time scale $\Delta t$, the shadow map $\mathcal{M}_H$, or equivalently, $\mathcal{N}$, will deviate from the case of ideal random diagonal unitaries, as $\Phi_2^{\Delta t}(X_2)\neq\Phi_2^{\mathrm{D}}(X_2)$ for $\Delta t\to\infty$. 
This not only changes the Choi matrix of Hamiltonian shadow map, but may even make it irreversible. 
Based on the logic of state learning, we can formally prove that:

\begin{proposition}\label{prop:degeneracy}
If $H$ has at least two same eigenvalues, the Hamiltonian shadow protocol is not tomography-complete.
\end{proposition}
\begin{proof}
Without loss of generality, we assume $e^{-iHt}=V_H\Lambda_tV_H^\dagger$ with two same eigenvalues $\Lambda_{a_1,a_1}=\Lambda_{a_2,a_2}$. As $V_H$ is a fixed unitary, completely learning $\rho$ is equivalent with completely learning $\rho_H=V_H^\dagger\rho V_H$. As the estimating probability $\bra{b}e^{-iHt}\rho e^{iHt}\ket{b}=\bra{b}V_H\Lambda_t\rho_H\overline{\Lambda}_tV_H^\dagger\ket{b}$, we need to decide whether measurements in basis of $\{\overline{\Lambda}_tV_H^\dagger\ketbra{b}{b}V_H\Lambda_t\}_{t,b}$ can extract full information of a state. Assuming $\Lambda_t=\mathrm{diag}(e^{-iE_1t},\cdots,e^{-iE_dt})$, the matrix form of $\Lambda_t\rho_H\overline{\Lambda}_t$ is 
\begin{equation}
\begin{aligned}
\Lambda_t\rho_H\overline{\Lambda}_t=\sum_{j,k}(\rho_{H})_{j,k}e^{-i(E_j-E_k)t}\ketbra{j}{k}.
\end{aligned}
\end{equation}
With the condition of $\Lambda_{a_1,a_1}=\Lambda_{a_2,a_2}$, in addition to the diagonal terms, $(\Lambda_t\rho_H\overline{\Lambda}_t)_{a_1,a_2}$ and $(\Lambda_t\rho_H\overline{\Lambda}_t)_{a_2,a_1}$ are also independent with the evolution time $t$. As $V_H$ is a fixed unitary, the probability $\bra{b}V_H\Lambda_t\rho_H\overline{\Lambda}_tV_H^\dagger\ket{b}$ is a linear function of all elements of $\rho_H$. While, coefficients of diagonal terms of $\rho_H$, $(\rho_H)_{a_1,a_2}$, and $(\rho_H)_{a_2,a_1}$ are independent with $t$. Thus, varying evolution time will not provide more independent equations to determine these terms. Using $d$ independent equations to estimate a total of $d+2$ unknown parameters is prohibited, so one cannot learn $\rho$ completely and the Hamiltonian shadow protocol is not tomography-complete.

\end{proof}

Notice that, even with degeneracy, other elements of $\Lambda_t\rho_H\overline{\Lambda}_t$ that depends on $t$ can still be estimated. This is because one could measure in many different evolution time $t$ and acquire many independent equations to determine these elements. Thus, although the Hamiltonian shadow protocol is not tomography-complete in this case, we can still estimate certain observables with special forms. Besides, if more than one degenerate Hamiltonians with different $V_H$ are accessible, it is possible to recover the tomography-completeness of the Hamiltonian shadow.

We can also use the proof to understand why a Hamiltonian without any degeneracy can be used to completely learn a state. In this scenario, all off-diagonal terms of $\Lambda_t\rho_H\overline{\Lambda}_t$ depend on evolution time $t$ with different coefficients. Subsequently, the diagonal elements of $V_H\Lambda_t\rho_H\overline{\Lambda}_tV_H^\dagger$ linearly depend on these off-diagonal terms with different time-dependent coefficients. Thus, one could estimate all these off-diagonal terms by measuring $V_H\Lambda_t\rho_H\overline{\Lambda}_tV_H^\dagger$ in computational basis for different evolution time, as these measurements provide sufficiently many independent equations to determine these terms. At the same time, the number of time-independent diagonal terms of $\Lambda_t\rho_H\overline{\Lambda}_t$ is no more than the number of measurement results. As a result, all elements of $\rho_H$ can be estimated and the Hamiltonian shadow protocol is tomography-complete.

As introduced previously, the specific form of Hamiltonian shadow map is based on the second-order integral of random diagonal unitaries. As shown in Eq.~\eqref{eq:second_degeneracy}, the second-order degeneracy can also affect the second-order integral, and therefore change the Choi matrix of Hamiltonian shadow map and make the $X_H$ matrix cannot fully describe $\mathcal{N}$. Thus, a natural question is, does the second-order degeneracy affect the reversibility of the shadow map? We prove that this is not the case:

\begin{proposition}\label{prop:resonance}
Given a non-degenerate Hamiltonian with second order degeneracy, the corresponding Hamiltonian shadow protocol is tomography-complete. Here by second-order degeneracy, we mean that some eigen-energies of $H$ satisfy $E_{a_1}+E_{a_2}=E_{b_1}+E_{b_2}$ with $a_1\neq b_1$, $a_1\neq b_2$, $a_2\neq b_1$, $a_2\neq b_2$, and $a_1\neq a_2$. 
\end{proposition}
\begin{proof}
With second-order degeneracy, the evolved state in this situation is 
\begin{equation}
\begin{aligned}
\Lambda_t\rho_H\overline{\Lambda}_t=&\sum_{j,k}(\rho_H)_{j,k}e^{-i(E_j-E_k)t}\ketbra{j}{k}\\
=&\cdots + \left[(\rho_H)_{a_1,b_1}e^{-i(E_{a_1}-E_{b_1})t}\ketbra{a_1}{b_1}+(\rho_H)_{b_2,a_2}e^{-i(E_{b_2}-E_{a_2})t}\ketbra{b_2}{a_2} + h.c\right]\\
=&\cdots + \left\{e^{-i(E_{a_1}-E_{b_1})t}\left[(\rho_H)_{a_1,b_1}\ketbra{a_1}{b_1}+(\rho_H)_{b_2,a_2}\ketbra{b_2}{a_2}\right] + h.c\right\},
\end{aligned}
\end{equation}
where we use the condition that $E_{a_1}-E_{b_1}=E_{b_2}-E_{a_2}$. Thus, coefficients of $(\rho_H)_{a_1,b_1}$ and $(\rho_H)_{b_2,a_2}$ are always same at any evolution time $t$. This also holds true for coefficients of $(\rho_H)_{a_1,b_2}$ and $(\rho_H)_{b_1,a_2}$. Similar to the case of non-degenerate Hamiltonian, in the current setting, other off-diagonal elements except for $(\rho_H)_{a_1,b_1}$, $(\rho_H)_{b2,a2}$, $(\rho_H)_{a_1,b_2}$, and $(\rho_H)_{b_1,a_2}$ and all diagonal elements of $\Lambda_t\rho_H\overline{\Lambda}_t$ can still be estimated. We will prove that these four elements can also be estimated.

We first consider the estimation of $(\rho_H)_{a_1,b_1}$ and $(\rho_H)_{b_2,a_2}$. Assuming $(\rho_H)_{a_1,b_1}=r_1+is_1$ and $(\rho_H)_{a_2,b_2}=r_2+is_2$, the diagonal element of $V_H\Lambda_t\rho_H\overline{\Lambda}_t V_H^\dagger$ can be written as
\begin{equation}
\bra{j}V_H\Lambda_t\rho_H\overline{\Lambda}_tV_H^\dagger\ket{j}=\cdots+\cos\left[(E_{a_1}-E_{b_1})t\right]\left(\alpha^j_{1}r_1+\alpha^j_2r_2+\alpha^j_3s_1+\alpha^j_4s_2\right)+\sin\left[(E_{a_1}-E_{b_1})t\right]\left(\beta^j_1r_1+\beta^j_2r_2+\beta^j_3r_3+\beta^j_4r_4\right),
\end{equation}
where $\alpha$ and $\beta$ are coefficients determined solely by $V_H$. So, by collecting measurement results  $\bra{j}V_H\Lambda_t\rho_H\overline{\Lambda}_tV_H^\dagger\ket{j}$ with different time $t$ and $\ket{j}$ can help us to estimate values of $\left(\alpha^j_1r_1+\alpha^j_2r_2+\alpha^j_3s_1+\alpha^j_4s_2\right)$ and $\left(\beta^j_1r_1+\beta^j_2r_2+\beta^j_3r_3+\beta^j_4r_4\right)$ for many different $\alpha$ and $\beta$. Solving these equations helps to get values of $r_1$, $r_2$, $s_1$, and $s_2$. Following the same logic, $(\rho_H)_{a_1,b_2}$, and $(\rho_H)_{b_1,a_2}$ can also be estimated. Thus, the Hamiltonian shadow is still tomography-complete.
\end{proof}

Denote $R(\mathcal{C})$ to be a matrix constructed by permuting indices of Choi matrix, $[R(\mathcal{C})]_{ij,kl}=\mathcal{C}_{lj,ik}$. According to the concatenating rule of Choi matrices, if $R(\mathcal{C})$ is an invertible matrix, the corresponding map is invertible, and vise versa. We numerically substantiated Proposition \ref{prop:degeneracy} and Proposition \ref{prop:resonance} by constructing Choi matrices of Hamiltonian shadow maps using Eq.~\eqref{eq:degeneracy} and Eq.~\eqref{eq:second_degeneracy} and the matrix $R(\mathcal{C})$. We verified that the matrix $R(\mathcal{C})$ is irreversible for first-order degeneracy, while invertible for second-order case.

% In fact, one can prove that Proposition \ref{prop:resonance} summarizes all situations caused by eigenvalues of $H$ that would make the Hamiltonian shadow not tomography-complete. This is because that in addition to these situations, off-diagonal terms of $\Lambda_t\rho_H\overline{\Lambda}_t$ all have independent time-dependent coefficients. This can also be derived from discussions in Appendix \ref{subsec:limit_time}. It shows that without these conditions, $\Phi_2^{\Delta t}(X_2)$ will converge to $\Phi_2^{\mathrm{D}}(X_2)$ when $\Delta t\to\infty$.

Our analysis is based on the most fundamental perspective, the state learning task, which can also be used to understand why single-patch Hamiltonians used in local Hamiltonian shadow protocol, Fig.~\ref{fig:Hamiltonian}(b), need to be independent. If the evolution unitary is $U=\bigotimes_{p=1}^Ne^{-iH_pt}$, the corresponding Hamiltonian is $H=\sum_{p=1}^NH_p\otimes\mathbb{I}_{[N]-p}$, where $H_p$ only acts on patch $p$ and $\mathbb{I}_{[N]-p}$ is the identity operator acting on other patches. If two Hamiltonians are same, $H_p=H_{p^\prime}$, $H$ will be a degenerate Hamiltonian. According to Proposition \ref{prop:degeneracy}, such Hamiltonian cannot be used to completely learn a state. If one can set different evolution time for different patches, $U=\bigotimes_{p=1}^Ne^{-iH_pt_p}$, this is equivalent with changing the Hamiltonian to $H=\sum_{p=1}^N\frac{t_p}{t}H_p\otimes\mathbb{I}_{[N]-p}$. By randomly setting $t_p$, the Hamiltonian can lose its degeneracy.

\subsection{Incomplete Measurement}\label{subsec:limit_measurement}
For some practical analog quantum systems, addressing single particle may not be a easy task. It might be more feasible for them to measure some global properties like the total $z$-direction spin and the parity. So, it is important to decide whether the Hamiltonian shadow protocol is tomography-complete with incomplete measurement. While, a direct corollary from the analysis in the previous section is that an incomplete measurement will result in a tomography-incomplete Hamiltonian shadow map.

\begin{proposition}
In the final stage of the Hamiltonian shadow protocol, if the $N$-qubit measurement cannot give $2^N-1$ independent measurement results, the Hamiltonian shadow protocol is not tomography-complete.
\end{proposition}
\begin{proof}
The proof is similar with the proof of Proposition \ref{prop:degeneracy}. There is a total of $2^N$ time-independent elements in $\Lambda_t\rho_H \Lambda_t^\dagger$. Considering the normalization condition, there exists a total of $2^N-1$ independent time-independent parameters. To estimate these elements by applying measurements on $V_H\Lambda_t\rho_H \Lambda_t^\dagger V_H^\dagger$, one needs $2^N-1$ independent measurement results.
\end{proof}

It is worth mentioning that, although we cannot estimate diagonal elements of $\rho_H$ in the case of incomplete measurement, it is still feasible to estimate all off-diagonal elements of $\rho_H$, as they are all dependent on the evolution time. Besides, if we adopt the original shadow protocol instead of the Hamiltonian shadow, the complete estimation of target state becomes possible. As in this case, all the elements of $U\rho U^\dagger$ are dependent on the unitary $U$.

\subsection{Eigenstates}\label{subsec:eigenstates}

In addition to eigenvalues of $H$, the eigenstates of $H$, or equivalently $V_H$, can also affect the reversibility of the Hamiltonian shadow map. As shown in Appendix \ref{sec:shadow_map}, without any degeneracy, the reversibility of Hamiltonian shadow map is determined by $X_H$. Specifically, the Hamiltonian shadow map is invertible if and only if $X_H$ is an invertible matrix and $(X_H)_{i,j}\neq 0$ for all $i\neq j$. These conditions appear to have limited physical implications. We show a specific scenario, where the Hamiltonian shadow map is not invertible caused by its eigenstates.

\begin{proposition}\label{theorem:computational_eigenstates}
If the Hamiltonian $H$ has some eigenstates that aligns with measurement basis, the Hamiltonian shadow protocol is not tomography-complete.
\end{proposition}
\begin{proof}
This theorem can be easily proved using the $X_H$ matrix. In this case, the Hamiltonian has a block-diagonal form in computational basis, $H = \begin{bmatrix} H^\prime & 0 \\ 0 & H^{\mathrm{D}} \end{bmatrix}$, where $H^{\mathrm{D}}$ is a diagonal matrix constructed by those computational basis eigenstates and their eigenvalues. At the same time, the unitary $V_H$ and matrix $X_H$ all have the form $\begin{bmatrix} V^\prime & 0 \\ 0 & V^{\mathrm{D}} \end{bmatrix}$ and $\begin{bmatrix} X^\prime & 0 \\ 0 & X^{\mathrm{D}}  \end{bmatrix}$, which does not satisfy the condition of $(X_H)_{i,j}\neq 0$ for $i\neq j$. In this case, some off-diagonal terms of $\rho$ cannot be estimated and the Hamiltonian shadow protocol is not tomography-complete.

We can also prove this proposition from the viewpoint of state learning. The evolution unitary of this Hamiltonian has the form of $e^{-iHt}=\begin{bmatrix} V_{H}^\prime\Lambda_1 V_{H}^{\prime\dagger} & 0 \\ 0 & \Lambda_2 \end{bmatrix}$, where $\Lambda_1$ and $\Lambda_2$ are independent random diagonal unitaries. The evolved state is 
\begin{equation}
e^{-iHt}\begin{bmatrix}
    \rho_{00} & \rho_{01} \\ \rho_{10} & \rho_{11}
\end{bmatrix}
e^{iHt}=\begin{bmatrix}
V_{H}^\prime\Lambda_1 V_{H}^{\prime\dagger}\rho_{00}V_{H}^\prime\overline{\Lambda}_1 V_{H}^{\prime\dagger} & V_{H}^\prime\Lambda_1 V_{H}^{\prime\dagger}\rho_{01}\overline{\Lambda}_2 \\ \Lambda_2\rho_{10}V_{H}^\prime\overline{\Lambda}_1 V_{H}^{\prime\dagger} & \Lambda_2\rho_{11}\overline{\Lambda}_2
\end{bmatrix}.
\end{equation}
where ${\rho_{ij}}_{i,j}$ denote blocks constituting $\rho$. Following the same logic of Hamiltonian shadow, measuring the evolved state in computational basis can only estimate all elements of $\rho_{00}$ and diagonal elements of $\rho_{11}$. This fact would strongly limit the detection capability of Hamiltonian shadow, and only observables of the form $O=\begin{bmatrix}
O^\prime & 0 \\ 0 & O^{\mathrm{D}}
\end{bmatrix}$ can be estimated.
\end{proof}

However, an intriguing observation is that, the Hamiltonian itself has this form. 
This means that, when computational basis eigenstates exist, the expectation value of $\Tr(H\rho)$ can still be estimated using the Hamiltonian shadow constructed using $e^{-iHt}$. 
The protocol needs to be slightly modified as following. 
One first needs to treat the whole Hilbert system as the direct sum of two systems $\mathcal{H}=\mathcal{H}^\prime\oplus\mathcal{H}_{\mathrm{D}}$, where $\mathcal{H}_{\mathrm{D}}$ is the Hilbert space corresponding to those computational basis eigenstates. Then, if the measurement outcome $\ket{b}$ is in $\mathcal{H}_{\mathrm{D}}$, the estimator of $\Tr(H\rho)$ is $\bra{b}H\ket{b}$. If $\ket{b}$ is in $\mathcal{H}^\prime$, we can first rewrite it as a lower-dimensional vector and use $H^\prime$ as the Hamiltonian to construct the Hamiltonian shadow estimator $\hat{\rho}^\prime$, which is the unbiased estimator of $\rho_{00}$. Then, we use $\Tr(H^\prime\hat{\rho}^\prime)$ to estimate the energy $\Tr(H\rho)$.

\begin{corollary}
When $H$ has some eigenstates that align with computational basis, the Hamiltonian shadow protocol can still estimate the expectation value of $\Tr(H\rho)$.
\end{corollary}

In fact, the Hamiltonian shadow protocol will approach the modified protocol we introduced above for estimating $\tr(H\rho)$ when eigenstates of $H$ approach computational basis. We can numerically show this by Fig.~\ref{fig:Hamiltonian_theta}. The numerical setting is similar with Fig.~3 in the main context, where we choose $V_H=e^{iP\theta}$ with $P$ being a random Hermitian matrix. When estimating observables like Pauli matrix and purity, variances increase when the value of $\theta$ decreases. This is because that $V_H$ approaches the identity matrix and it is hard to estimate off-diagonal terms of $\rho$. However, when we set the observable to the Hamiltonian $H$ itself, the variance of Hamiltonian shadow keeps a constant when $\theta$ approaches zero. This provides an evidence for our statement.

\begin{figure}
\centering
\includegraphics[width=0.4\textwidth]{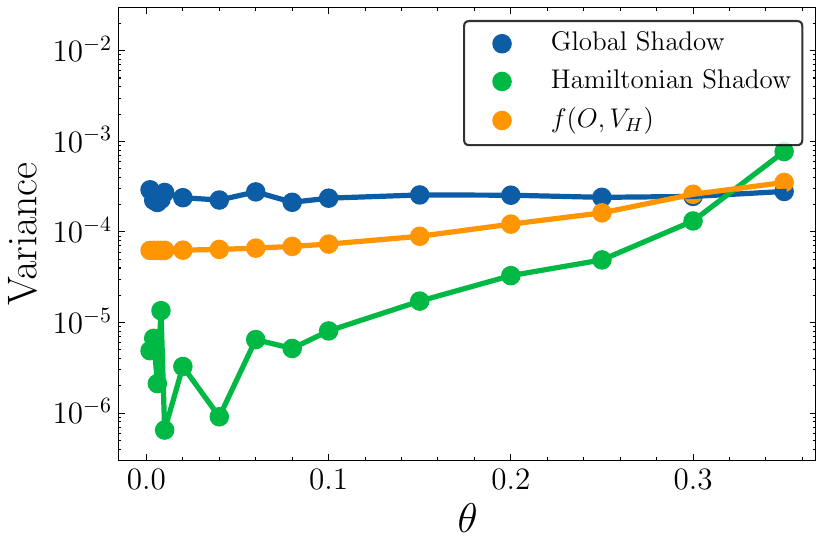}
\caption{The scaling of variance of estimating $H=V_H L V_H^\dagger$ when using $V_H=e^{iP\theta}$ to perform Hamiltonian shadow estimation. $L$ is a non-degenerate real diagonal matrix to represent eigenvalues of the Hamiltonian. The target state is a four-qubit GHZ state and the experiment times is set to be $K=1000$. The value of every point is chosen by taking the median of ten independently sampled $P$.}
\label{fig:Hamiltonian_theta}
\end{figure}

Following the same logic, the condition of computational eigenstates can be covered by a more general condition:
\begin{corollary}\label{coro:block_diagonal}
When the Hamiltonian $H$ is block-diagonal in computational basis, $H=\begin{bmatrix}
H_1 & 0 \\ 0 & H_2
\end{bmatrix}$, the Hamiltonian shadow protocol is not tomography-complete. While if the observable $O$ or target quantum state $\rho$ has a same block-diagonal structure, the value of $\Tr(O\rho)$ can still be estimated using Hamiltonian shadow.
\end{corollary}

Corollary \ref{coro:block_diagonal} is just one situation where the Hamiltonian shadow map becomes invertible caused by eigenstates, there are many other scenarios in which $X_H$ is not invertible or has zero off-diagonal elements. An example is shown in Sec.~\ref{sec:hadamard}, where $V_H=h^{\otimes N}$ and $(X_H)_{i,j}=2^{-N}$ for all $i$ and $j$. Note that, we can develop a simpler criterion to decide whether $X_H$ is invertible or not. As $X_H=(V_H^{\mathrm{sq}})^TV_H^{\mathrm{sq}}$, deciding whether $X_H$ is invertible or not is equivalent to deciding whether $V_H^{\mathrm{sq}}$ is invertible or not.

\subsection{Target State}

Until now, all discussions are about the evolution unitary $e^{-iHt}$ and measurement. It seems that the reversibility of Hamiltonian shadow map is not related with the target state $\rho$. This is counter-intuitive as if the state is a thermal state $\rho=\frac{e^{-\beta H}}{\Tr(e^{-\beta H})}$, it does not change under Hamiltonian evolution, $e^{-iHt}\rho e^{iHt}=\rho$. Therefore, the final computational basis measurements can only estimate the diagonal terms of $\rho$, and the Hamiltonian shadow seems to be not tomography-complete. While we surprisingly find that this is not the case!

The reason is simple while interesting. Since the Hamiltonian shadow protocol allows us to have the complete information of the Hamiltonian, the unitary $V_H$ is also known. If the state commutes with the Hamiltonian, it has the form of $\rho=V_HL V_H^\dagger$, where $L$ is a positive diagonal matrix. Thus, when $V_H$ is known, there are only $2^N$ unknown parameters to determine $\rho$. As there are also $2^N$ diagonal elements of $\rho$, it is possible that one can reconstruct the whole density matrix $\rho$ with only its diagonal elements. Here we will show that the Hamiltonian shadow can do this task directly.

\begin{proposition}
If the target state commutes with the Hamiltonian, $[H,\rho]=0$, the Hamiltonian shadow protocol also gives the unbiased estimator of $\rho$,
\begin{equation}
\mathbb{E}_{t,b}\left[V_H\mathcal{N}^{-1}\left(\overline{\Lambda}_tV_H^\dagger\ketbra{b}{b}V_H\Lambda_t\right)V_H^\dagger\right]=\rho
\end{equation}
\end{proposition}
% \red{hzh: $V_H \mathcal{N}^{-1}()V_H^\dag$?}

\begin{proof}
\comments{
We first show how diagonal terms of $L$ are related with diagonal terms of $\rho$. By definition, we have
\begin{equation}
\rho_{i,i}=(V_H L V_H^\dagger)_{i,i}=\sum_j(V_H)_{i,j}L_{j,j}(V_H^\dagger)_{j,i}=\sum_j\abs{(V_H)_{i,j}}^2L_{j,j}=\sum_{j}(V_H^{\mathrm{sq}})_{i,j}L_{j,j}.
\end{equation}
Therefore, $L$ can be derived using diagonal terms of $\rho$ via the relation of $L_{i,i}=\sum_j[(V_H^{\mathrm{sq}})^{-1}]_{i,j}\rho_{j,j}$.
}
To prove the Hamiltonian shadow also works for the case of $[H,\rho]=0$, we just need to prove 
\begin{equation}
\mathbb{E}_{t,b}\left[\mathcal{N}^{-1}\left(\overline{\Lambda}_tV_H^\dagger\ketbra{b}{b}V_H\Lambda_t\right)\right]=L.
\end{equation}
Instituting Born's rule and the condition of $\rho=V_HLV_H^\dagger$, we have
\begin{equation}
\begin{aligned}
\mathbb{E}_{t,b}\left[\mathcal{N}^{-1}\left(\overline{\Lambda}_tV_H^\dagger\ketbra{b}{b}V_H\Lambda_t\right)\right]=&\mathbb{E}_t\left[\sum_b\bra{b}e^{-iHt}\rho e^{iHt}\ket{b}\mathcal{N}^{-1}\left(\overline{\Lambda}_tV_H^\dagger\ketbra{b}{b}V_H\Lambda_t\right)\right]\\
=&\sum_b\bra{b}V_HLV_H^\dagger\ket{b}\mathcal{N}^{-1}\left(\mathbb{E}_t\overline{\Lambda}_tV_H^\dagger\ketbra{b}{b}V_H\Lambda_t\right)\\
=&\sum_b\bra{b}V_HLV_H^\dagger\ket{b}\mathcal{N}^{-1}\left[\sum_i\left(V_H^\dagger\ketbra{b}{b}V_H\right)_{i,i}\ketbra{i}{i}\right]\\
=&\sum_{b}\sum_k(V_H)_{b,k}L_{k,k}(V_H^\dagger)_{k,b}\sum_{i,j}(X_H^{-1})_{i,j}\left(V_H^\dagger\ketbra{b}{b}V_H\right)_{j,j}\ketbra{i}{i}\\
=&\sum_{i,j,k}\sum_b\abs{(V_H)_{b,k}}^2\abs{(V_H)_{b,j}}^2(X_H^{-1})_{i,j}L_{k,k}\ketbra{i}{i}\\
=&\sum_{i,j,k}\left[(V_H^{\mathrm{sq}})^TV_H^{\mathrm{sq}}\right]_{k,j}(X_H^{-1})_{i,j}L_{k,k}\ketbra{i}{i}\\
=&\sum_{i,j,k}(X_H)_{k,j}(X_H^{-1})_{j,i}L_{k,k}\ketbra{i}{i}\\
=&\sum_{i,k}\delta_{i,k}L_{k,k}\ketbra{i}{i}\\
=&\sum_iL_{i,i}\ketbra{i}{i}=L,
\end{aligned}
\end{equation}
where the third equal sign is due to the property of first-order integral of diagonal random unitary, where all the off-diagonal terms are twirled out; the fourth to the last equal sign is derived using the fact that $X_H$ and $X_H^{-1}$ are both real symmetric matrices; the third to the last equal sign is derived using a simple fact that $X_HX_H^{-1}=\mathbb{I}$.
\end{proof}


\begin{thebibliography}{57}%
	\makeatletter
	\providecommand \@ifxundefined [1]{%
		\@ifx{#1\undefined}
	}%
	\providecommand \@ifnum [1]{%
		\ifnum #1\expandafter \@firstoftwo
		\else \expandafter \@secondoftwo
		\fi
	}%
	\providecommand \@ifx [1]{%
		\ifx #1\expandafter \@firstoftwo
		\else \expandafter \@secondoftwo
		\fi
	}%
	\providecommand \natexlab [1]{#1}%
	\providecommand \enquote  [1]{``#1''}%
	\providecommand \bibnamefont  [1]{#1}%
	\providecommand \bibfnamefont [1]{#1}%
	\providecommand \citenamefont [1]{#1}%
	\providecommand \href@noop [0]{\@secondoftwo}%
	\providecommand \href [0]{\begingroup \@sanitize@url \@href}%
	\providecommand \@href[1]{\@@startlink{#1}\@@href}%
	\providecommand \@@href[1]{\endgroup#1\@@endlink}%
	\providecommand \@sanitize@url [0]{\catcode `\\12\catcode `\$12\catcode
		`\&12\catcode `\#12\catcode `\^12\catcode `\_12\catcode `\%12\relax}%
	\providecommand \@@startlink[1]{}%
	\providecommand \@@endlink[0]{}%
	\providecommand \url  [0]{\begingroup\@sanitize@url \@url }%
	\providecommand \@url [1]{\endgroup\@href {#1}{\urlprefix }}%
	\providecommand \urlprefix  [0]{URL }%
	\providecommand \Eprint [0]{\href }%
	\providecommand \doibase [0]{https://doi.org/}%
	\providecommand \selectlanguage [0]{\@gobble}%
	\providecommand \bibinfo  [0]{\@secondoftwo}%
	\providecommand \bibfield  [0]{\@secondoftwo}%
	\providecommand \translation [1]{[#1]}%
	\providecommand \BibitemOpen [0]{}%
	\providecommand \bibitemStop [0]{}%
	\providecommand \bibitemNoStop [0]{.\EOS\space}%
	\providecommand \EOS [0]{\spacefactor3000\relax}%
	\providecommand \BibitemShut  [1]{\csname bibitem#1\endcsname}%
	\let\auto@bib@innerbib\@empty
	%</preamble>
	\bibitem [{\citenamefont {Bernien}\ \emph {et~al.}(2017)\citenamefont
		{Bernien}, \citenamefont {Schwartz}, \citenamefont {Keesling}, \citenamefont
		{Levine}, \citenamefont {Omran}, \citenamefont {Pichler}, \citenamefont
		{Choi}, \citenamefont {Zibrov}, \citenamefont {Endres}, \citenamefont
		{Greiner}, \citenamefont {Vuleti{\'c}},\ and\ \citenamefont
		{Lukin}}]{bernien2017dynamics}%
	\BibitemOpen
	\bibfield  {author} {\bibinfo {author} {\bibfnamefont {H.}~\bibnamefont
			{Bernien}}, \bibinfo {author} {\bibfnamefont {S.}~\bibnamefont {Schwartz}},
		\bibinfo {author} {\bibfnamefont {A.}~\bibnamefont {Keesling}}, \bibinfo
		{author} {\bibfnamefont {H.}~\bibnamefont {Levine}}, \bibinfo {author}
		{\bibfnamefont {A.}~\bibnamefont {Omran}}, \bibinfo {author} {\bibfnamefont
			{H.}~\bibnamefont {Pichler}}, \bibinfo {author} {\bibfnamefont
			{S.}~\bibnamefont {Choi}}, \bibinfo {author} {\bibfnamefont {A.~S.}\
			\bibnamefont {Zibrov}}, \bibinfo {author} {\bibfnamefont {M.}~\bibnamefont
			{Endres}}, \bibinfo {author} {\bibfnamefont {M.}~\bibnamefont {Greiner}},
		\bibinfo {author} {\bibfnamefont {V.}~\bibnamefont {Vuleti{\'c}}},\ and\
		\bibinfo {author} {\bibfnamefont {M.~D.}\ \bibnamefont {Lukin}},\ }\bibfield
	{title} {\bibinfo {title} {Probing many-body dynamics on a 51-atom quantum
			simulator},\ }\href {https://doi.org/10.1038/nature24622} {\bibfield
		{journal} {\bibinfo  {journal} {Nature}\ }\textbf {\bibinfo {volume} {551}},\
		\bibinfo {pages} {579} (\bibinfo {year} {2017})}\BibitemShut {NoStop}%
	\bibitem [{\citenamefont {Ebadi}\ \emph {et~al.}(2021)\citenamefont {Ebadi},
		\citenamefont {Wang}, \citenamefont {Levine}, \citenamefont {Keesling},
		\citenamefont {Semeghini}, \citenamefont {Omran}, \citenamefont {Bluvstein},
		\citenamefont {Samajdar}, \citenamefont {Pichler}, \citenamefont {Ho},
		\citenamefont {Choi}, \citenamefont {Sachdev}, \citenamefont {Greiner},
		\citenamefont {Vuleti{\'c}},\ and\ \citenamefont
		{Lukin}}]{ebadi2021quantumphase}%
	\BibitemOpen
	\bibfield  {author} {\bibinfo {author} {\bibfnamefont {S.}~\bibnamefont
			{Ebadi}}, \bibinfo {author} {\bibfnamefont {T.~T.}\ \bibnamefont {Wang}},
		\bibinfo {author} {\bibfnamefont {H.}~\bibnamefont {Levine}}, \bibinfo
		{author} {\bibfnamefont {A.}~\bibnamefont {Keesling}}, \bibinfo {author}
		{\bibfnamefont {G.}~\bibnamefont {Semeghini}}, \bibinfo {author}
		{\bibfnamefont {A.}~\bibnamefont {Omran}}, \bibinfo {author} {\bibfnamefont
			{D.}~\bibnamefont {Bluvstein}}, \bibinfo {author} {\bibfnamefont
			{R.}~\bibnamefont {Samajdar}}, \bibinfo {author} {\bibfnamefont
			{H.}~\bibnamefont {Pichler}}, \bibinfo {author} {\bibfnamefont {W.~W.}\
			\bibnamefont {Ho}}, \bibinfo {author} {\bibfnamefont {S.}~\bibnamefont
			{Choi}}, \bibinfo {author} {\bibfnamefont {S.}~\bibnamefont {Sachdev}},
		\bibinfo {author} {\bibfnamefont {M.}~\bibnamefont {Greiner}}, \bibinfo
		{author} {\bibfnamefont {V.}~\bibnamefont {Vuleti{\'c}}},\ and\ \bibinfo
		{author} {\bibfnamefont {M.~D.}\ \bibnamefont {Lukin}},\ }\bibfield  {title}
	{\bibinfo {title} {Quantum phases of matter on a 256-atom programmable
			quantum simulator},\ }\href {https://doi.org/10.1038/s41586-021-03582-4}
	{\bibfield  {journal} {\bibinfo  {journal} {Nature}\ }\textbf {\bibinfo
			{volume} {595}},\ \bibinfo {pages} {227} (\bibinfo {year}
		{2021})}\BibitemShut {NoStop}%
	\bibitem [{\citenamefont {Daley}\ \emph {et~al.}(2022)\citenamefont {Daley},
		\citenamefont {Bloch}, \citenamefont {Kokail}, \citenamefont {Flannigan},
		\citenamefont {Pearson}, \citenamefont {Troyer},\ and\ \citenamefont
		{Zoller}}]{daley2022practicaladvantage}%
	\BibitemOpen
	\bibfield  {author} {\bibinfo {author} {\bibfnamefont {A.~J.}\ \bibnamefont
			{Daley}}, \bibinfo {author} {\bibfnamefont {I.}~\bibnamefont {Bloch}},
		\bibinfo {author} {\bibfnamefont {C.}~\bibnamefont {Kokail}}, \bibinfo
		{author} {\bibfnamefont {S.}~\bibnamefont {Flannigan}}, \bibinfo {author}
		{\bibfnamefont {N.}~\bibnamefont {Pearson}}, \bibinfo {author} {\bibfnamefont
			{M.}~\bibnamefont {Troyer}},\ and\ \bibinfo {author} {\bibfnamefont
			{P.}~\bibnamefont {Zoller}},\ }\bibfield  {title} {\bibinfo {title}
		{Practical quantum advantage in quantum simulation},\ }\href
	{https://doi.org/10.1038/s41586-022-04940-6} {\bibfield  {journal} {\bibinfo
			{journal} {Nature}\ }\textbf {\bibinfo {volume} {607}},\ \bibinfo {pages}
		{667} (\bibinfo {year} {2022})}\BibitemShut {NoStop}%
	\bibitem [{\citenamefont {Evered}\ \emph {et~al.}(2023)\citenamefont {Evered},
		\citenamefont {Bluvstein}, \citenamefont {Kalinowski}, \citenamefont {Ebadi},
		\citenamefont {Manovitz}, \citenamefont {Zhou}, \citenamefont {Li},
		\citenamefont {Geim}, \citenamefont {Wang}, \citenamefont {Maskara},
		\citenamefont {Levine}, \citenamefont {Semeghini}, \citenamefont {Greiner},
		\citenamefont {Vuleti{\'c}},\ and\ \citenamefont
		{Lukin}}]{evered2023highfidelity}%
	\BibitemOpen
	\bibfield  {author} {\bibinfo {author} {\bibfnamefont {S.~J.}\ \bibnamefont
			{Evered}}, \bibinfo {author} {\bibfnamefont {D.}~\bibnamefont {Bluvstein}},
		\bibinfo {author} {\bibfnamefont {M.}~\bibnamefont {Kalinowski}}, \bibinfo
		{author} {\bibfnamefont {S.}~\bibnamefont {Ebadi}}, \bibinfo {author}
		{\bibfnamefont {T.}~\bibnamefont {Manovitz}}, \bibinfo {author}
		{\bibfnamefont {H.}~\bibnamefont {Zhou}}, \bibinfo {author} {\bibfnamefont
			{S.~H.}\ \bibnamefont {Li}}, \bibinfo {author} {\bibfnamefont {A.~A.}\
			\bibnamefont {Geim}}, \bibinfo {author} {\bibfnamefont {T.~T.}\ \bibnamefont
			{Wang}}, \bibinfo {author} {\bibfnamefont {N.}~\bibnamefont {Maskara}},
		\bibinfo {author} {\bibfnamefont {H.}~\bibnamefont {Levine}}, \bibinfo
		{author} {\bibfnamefont {G.}~\bibnamefont {Semeghini}}, \bibinfo {author}
		{\bibfnamefont {M.}~\bibnamefont {Greiner}}, \bibinfo {author} {\bibfnamefont
			{V.}~\bibnamefont {Vuleti{\'c}}},\ and\ \bibinfo {author} {\bibfnamefont
			{M.~D.}\ \bibnamefont {Lukin}},\ }\bibfield  {title} {\bibinfo {title}
		{High-fidelity parallel entangling gates on a neutral-atom quantum
			computer},\ }\href {https://doi.org/10.1038/s41586-023-06481-y} {\bibfield
		{journal} {\bibinfo  {journal} {Nature}\ }\textbf {\bibinfo {volume} {622}},\
		\bibinfo {pages} {268} (\bibinfo {year} {2023})}\BibitemShut {NoStop}%
	\bibitem [{\citenamefont {Hu}\ \emph {et~al.}(2019)\citenamefont {Hu},
		\citenamefont {Liu}, \citenamefont {Grimes}, \citenamefont {Lin},
		\citenamefont {Gheorghe}, \citenamefont {Vexiau}, \citenamefont
		{Bouloufa-Maafa}, \citenamefont {Dulieu}, \citenamefont {Rosenband},\ and\
		\citenamefont {Ni}}]{hu2019molecule}%
	\BibitemOpen
	\bibfield  {author} {\bibinfo {author} {\bibfnamefont {M.-G.}\ \bibnamefont
			{Hu}}, \bibinfo {author} {\bibfnamefont {Y.}~\bibnamefont {Liu}}, \bibinfo
		{author} {\bibfnamefont {D.~D.}\ \bibnamefont {Grimes}}, \bibinfo {author}
		{\bibfnamefont {Y.-W.}\ \bibnamefont {Lin}}, \bibinfo {author} {\bibfnamefont
			{A.~H.}\ \bibnamefont {Gheorghe}}, \bibinfo {author} {\bibfnamefont
			{R.}~\bibnamefont {Vexiau}}, \bibinfo {author} {\bibfnamefont
			{N.}~\bibnamefont {Bouloufa-Maafa}}, \bibinfo {author} {\bibfnamefont
			{O.}~\bibnamefont {Dulieu}}, \bibinfo {author} {\bibfnamefont
			{T.}~\bibnamefont {Rosenband}},\ and\ \bibinfo {author} {\bibfnamefont
			{K.-K.}\ \bibnamefont {Ni}},\ }\bibfield  {title} {\bibinfo {title} {Direct
			observation of bimolecular reactions of ultracold krb molecules},\ }\href
	{https://doi.org/10.1126/science.aay9531} {\bibfield  {journal} {\bibinfo
			{journal} {Science}\ }\textbf {\bibinfo {volume} {366}},\ \bibinfo {pages}
		{1111} (\bibinfo {year} {2019})}\BibitemShut {NoStop}%
	\bibitem [{\citenamefont {Burchesky}\ \emph {et~al.}(2021)\citenamefont
		{Burchesky}, \citenamefont {Anderegg}, \citenamefont {Bao}, \citenamefont
		{Yu}, \citenamefont {Chae}, \citenamefont {Ketterle}, \citenamefont {Ni},\
		and\ \citenamefont {Doyle}}]{burchesky2021rotation}%
	\BibitemOpen
	\bibfield  {author} {\bibinfo {author} {\bibfnamefont {S.}~\bibnamefont
			{Burchesky}}, \bibinfo {author} {\bibfnamefont {L.}~\bibnamefont {Anderegg}},
		\bibinfo {author} {\bibfnamefont {Y.}~\bibnamefont {Bao}}, \bibinfo {author}
		{\bibfnamefont {S.~S.}\ \bibnamefont {Yu}}, \bibinfo {author} {\bibfnamefont
			{E.}~\bibnamefont {Chae}}, \bibinfo {author} {\bibfnamefont {W.}~\bibnamefont
			{Ketterle}}, \bibinfo {author} {\bibfnamefont {K.-K.}\ \bibnamefont {Ni}},\
		and\ \bibinfo {author} {\bibfnamefont {J.~M.}\ \bibnamefont {Doyle}},\
	}\bibfield  {title} {\bibinfo {title} {Rotational coherence times of polar
			molecules in optical tweezers},\ }\href
	{https://doi.org/10.1103/PhysRevLett.127.123202} {\bibfield  {journal}
		{\bibinfo  {journal} {Phys. Rev. Lett.}\ }\textbf {\bibinfo {volume} {127}},\
		\bibinfo {pages} {123202} (\bibinfo {year} {2021})}\BibitemShut {NoStop}%
	\bibitem [{\citenamefont {Bloch}\ \emph {et~al.}(2012)\citenamefont {Bloch},
		\citenamefont {Dalibard},\ and\ \citenamefont {Nascimb{\`e}ne}}]{Bloch}%
	\BibitemOpen
	\bibfield  {author} {\bibinfo {author} {\bibfnamefont {I.}~\bibnamefont
			{Bloch}}, \bibinfo {author} {\bibfnamefont {J.}~\bibnamefont {Dalibard}},\
		and\ \bibinfo {author} {\bibfnamefont {S.}~\bibnamefont {Nascimb{\`e}ne}},\
	}\bibfield  {title} {\bibinfo {title} {Quantum simulations with ultracold
			quantum gases},\ }\href {https://doi.org/10.1038/nphys2259} {\bibfield
		{journal} {\bibinfo  {journal} {Nature Physics}\ }\textbf {\bibinfo {volume}
			{8}},\ \bibinfo {pages} {267} (\bibinfo {year} {2012})}\BibitemShut {NoStop}%
	\bibitem [{\citenamefont {Hartke}\ \emph {et~al.}(2023)\citenamefont {Hartke},
		\citenamefont {Oreg}, \citenamefont {Turnbaugh}, \citenamefont {Jia},\ and\
		\citenamefont {Zwierlein}}]{thomas2023fermionpairing}%
	\BibitemOpen
	\bibfield  {author} {\bibinfo {author} {\bibfnamefont {T.}~\bibnamefont
			{Hartke}}, \bibinfo {author} {\bibfnamefont {B.}~\bibnamefont {Oreg}},
		\bibinfo {author} {\bibfnamefont {C.}~\bibnamefont {Turnbaugh}}, \bibinfo
		{author} {\bibfnamefont {N.}~\bibnamefont {Jia}},\ and\ \bibinfo {author}
		{\bibfnamefont {M.}~\bibnamefont {Zwierlein}},\ }\bibfield  {title} {\bibinfo
		{title} {Direct observation of nonlocal fermion pairing in an attractive
			fermi-hubbard gas},\ }\href {https://doi.org/10.1126/science.ade4245}
	{\bibfield  {journal} {\bibinfo  {journal} {Science}\ }\textbf {\bibinfo
			{volume} {381}},\ \bibinfo {pages} {82} (\bibinfo {year} {2023})}\BibitemShut
	{NoStop}%
	\bibitem [{\citenamefont {L{\'e}onard}\ \emph
		{et~al.}(2023{\natexlab{a}})\citenamefont {L{\'e}onard}, \citenamefont {Kim},
		\citenamefont {Kwan}, \citenamefont {Segura}, \citenamefont {Grusdt},
		\citenamefont {Repellin}, \citenamefont {Goldman},\ and\ \citenamefont
		{Greiner}}]{lenard2023FQH}%
	\BibitemOpen
	\bibfield  {author} {\bibinfo {author} {\bibfnamefont {J.}~\bibnamefont
			{L{\'e}onard}}, \bibinfo {author} {\bibfnamefont {S.}~\bibnamefont {Kim}},
		\bibinfo {author} {\bibfnamefont {J.}~\bibnamefont {Kwan}}, \bibinfo {author}
		{\bibfnamefont {P.}~\bibnamefont {Segura}}, \bibinfo {author} {\bibfnamefont
			{F.}~\bibnamefont {Grusdt}}, \bibinfo {author} {\bibfnamefont
			{C.}~\bibnamefont {Repellin}}, \bibinfo {author} {\bibfnamefont
			{N.}~\bibnamefont {Goldman}},\ and\ \bibinfo {author} {\bibfnamefont
			{M.}~\bibnamefont {Greiner}},\ }\bibfield  {title} {\bibinfo {title}
		{Realization of a fractional quantum hall state with ultracold atoms},\
	}\href {https://doi.org/10.1038/s41586-023-06122-4} {\bibfield  {journal}
		{\bibinfo  {journal} {Nature}\ }\textbf {\bibinfo {volume} {619}},\ \bibinfo
		{pages} {495} (\bibinfo {year} {2023}{\natexlab{a}})}\BibitemShut {NoStop}%
	\bibitem [{\citenamefont {Kiczynski}\ \emph {et~al.}(2022)\citenamefont
		{Kiczynski}, \citenamefont {Gorman}, \citenamefont {Geng}, \citenamefont
		{Donnelly}, \citenamefont {Chung}, \citenamefont {He}, \citenamefont
		{Keizer},\ and\ \citenamefont {Simmons}}]{Kiczynski2022quantumdot}%
	\BibitemOpen
	\bibfield  {author} {\bibinfo {author} {\bibfnamefont {M.}~\bibnamefont
			{Kiczynski}}, \bibinfo {author} {\bibfnamefont {S.~K.}\ \bibnamefont
			{Gorman}}, \bibinfo {author} {\bibfnamefont {H.}~\bibnamefont {Geng}},
		\bibinfo {author} {\bibfnamefont {M.~B.}\ \bibnamefont {Donnelly}}, \bibinfo
		{author} {\bibfnamefont {Y.}~\bibnamefont {Chung}}, \bibinfo {author}
		{\bibfnamefont {Y.}~\bibnamefont {He}}, \bibinfo {author} {\bibfnamefont
			{J.~G.}\ \bibnamefont {Keizer}},\ and\ \bibinfo {author} {\bibfnamefont
			{M.~Y.}\ \bibnamefont {Simmons}},\ }\bibfield  {title} {\bibinfo {title}
		{Engineering topological states in atom-based semiconductor quantum dots},\
	}\href {https://doi.org/10.1038/s41586-022-04706-0} {\bibfield  {journal}
		{\bibinfo  {journal} {Nature}\ }\textbf {\bibinfo {volume} {606}},\ \bibinfo
		{pages} {694} (\bibinfo {year} {2022})}\BibitemShut {NoStop}%
	\bibitem [{\citenamefont {Wang}\ \emph {et~al.}(2022)\citenamefont {Wang},
		\citenamefont {Khatami}, \citenamefont {Fei}, \citenamefont {Wyrick},
		\citenamefont {Namboodiri}, \citenamefont {Kashid}, \citenamefont {Rigosi},
		\citenamefont {Bryant},\ and\ \citenamefont {Silver}}]{wang2022quantumdot2}%
	\BibitemOpen
	\bibfield  {author} {\bibinfo {author} {\bibfnamefont {X.}~\bibnamefont
			{Wang}}, \bibinfo {author} {\bibfnamefont {E.}~\bibnamefont {Khatami}},
		\bibinfo {author} {\bibfnamefont {F.}~\bibnamefont {Fei}}, \bibinfo {author}
		{\bibfnamefont {J.}~\bibnamefont {Wyrick}}, \bibinfo {author} {\bibfnamefont
			{P.}~\bibnamefont {Namboodiri}}, \bibinfo {author} {\bibfnamefont
			{R.}~\bibnamefont {Kashid}}, \bibinfo {author} {\bibfnamefont {A.~F.}\
			\bibnamefont {Rigosi}}, \bibinfo {author} {\bibfnamefont {G.}~\bibnamefont
			{Bryant}},\ and\ \bibinfo {author} {\bibfnamefont {R.}~\bibnamefont
			{Silver}},\ }\bibfield  {title} {\bibinfo {title} {Experimental realization
			of an extended fermi-hubbard model using a 2d lattice of dopant-based quantum
			dots},\ }\href {https://doi.org/10.1038/s41467-022-34220-w} {\bibfield
		{journal} {\bibinfo  {journal} {Nature Communications}\ }\textbf {\bibinfo
			{volume} {13}},\ \bibinfo {pages} {6824} (\bibinfo {year}
		{2022})}\BibitemShut {NoStop}%
	\bibitem [{\citenamefont {Aidelsburger}\ \emph {et~al.}(2015)\citenamefont
		{Aidelsburger}, \citenamefont {Lohse}, \citenamefont {Schweizer},
		\citenamefont {Atala}, \citenamefont {Barreiro}, \citenamefont
		{Nascimb{\`e}ne}, \citenamefont {Cooper}, \citenamefont {Bloch},\ and\
		\citenamefont {Goldman}}]{aidelsburger2015chernnumber}%
	\BibitemOpen
	\bibfield  {author} {\bibinfo {author} {\bibfnamefont {M.}~\bibnamefont
			{Aidelsburger}}, \bibinfo {author} {\bibfnamefont {M.}~\bibnamefont {Lohse}},
		\bibinfo {author} {\bibfnamefont {C.}~\bibnamefont {Schweizer}}, \bibinfo
		{author} {\bibfnamefont {M.}~\bibnamefont {Atala}}, \bibinfo {author}
		{\bibfnamefont {J.~T.}\ \bibnamefont {Barreiro}}, \bibinfo {author}
		{\bibfnamefont {S.}~\bibnamefont {Nascimb{\`e}ne}}, \bibinfo {author}
		{\bibfnamefont {N.~R.}\ \bibnamefont {Cooper}}, \bibinfo {author}
		{\bibfnamefont {I.}~\bibnamefont {Bloch}},\ and\ \bibinfo {author}
		{\bibfnamefont {N.}~\bibnamefont {Goldman}},\ }\bibfield  {title} {\bibinfo
		{title} {Measuring the chern number of hofstadter bands with ultracold
			bosonic atoms},\ }\href {https://doi.org/10.1038/nphys3171} {\bibfield
		{journal} {\bibinfo  {journal} {Nature Physics}\ }\textbf {\bibinfo {volume}
			{11}},\ \bibinfo {pages} {162} (\bibinfo {year} {2015})}\BibitemShut
	{NoStop}%
	\bibitem [{\citenamefont {Smits}\ \emph {et~al.}(2018)\citenamefont {Smits},
		\citenamefont {Liao}, \citenamefont {Stoof},\ and\ \citenamefont {van~der
			Straten}}]{smits2018timecrystal}%
	\BibitemOpen
	\bibfield  {author} {\bibinfo {author} {\bibfnamefont {J.}~\bibnamefont
			{Smits}}, \bibinfo {author} {\bibfnamefont {L.}~\bibnamefont {Liao}},
		\bibinfo {author} {\bibfnamefont {H.~T.~C.}\ \bibnamefont {Stoof}},\ and\
		\bibinfo {author} {\bibfnamefont {P.}~\bibnamefont {van~der Straten}},\
	}\bibfield  {title} {\bibinfo {title} {Observation of a space-time crystal in
			a superfluid quantum gas},\ }\href
	{https://doi.org/10.1103/PhysRevLett.121.185301} {\bibfield  {journal}
		{\bibinfo  {journal} {Phys. Rev. Lett.}\ }\textbf {\bibinfo {volume} {121}},\
		\bibinfo {pages} {185301} (\bibinfo {year} {2018})}\BibitemShut {NoStop}%
	\bibitem [{\citenamefont {Semeghini}\ \emph {et~al.}(2021)\citenamefont
		{Semeghini}, \citenamefont {Levine}, \citenamefont {Keesling}, \citenamefont
		{Ebadi}, \citenamefont {Wang}, \citenamefont {Bluvstein}, \citenamefont
		{Verresen}, \citenamefont {Pichler}, \citenamefont {Kalinowski},
		\citenamefont {Samajdar}, \citenamefont {Omran}, \citenamefont {Sachdev},
		\citenamefont {Vishwanath}, \citenamefont {Greiner}, \citenamefont
		{Vuletić},\ and\ \citenamefont {Lukin}}]{semeghini2021spinliquid}%
	\BibitemOpen
	\bibfield  {author} {\bibinfo {author} {\bibfnamefont {G.}~\bibnamefont
			{Semeghini}}, \bibinfo {author} {\bibfnamefont {H.}~\bibnamefont {Levine}},
		\bibinfo {author} {\bibfnamefont {A.}~\bibnamefont {Keesling}}, \bibinfo
		{author} {\bibfnamefont {S.}~\bibnamefont {Ebadi}}, \bibinfo {author}
		{\bibfnamefont {T.~T.}\ \bibnamefont {Wang}}, \bibinfo {author}
		{\bibfnamefont {D.}~\bibnamefont {Bluvstein}}, \bibinfo {author}
		{\bibfnamefont {R.}~\bibnamefont {Verresen}}, \bibinfo {author}
		{\bibfnamefont {H.}~\bibnamefont {Pichler}}, \bibinfo {author} {\bibfnamefont
			{M.}~\bibnamefont {Kalinowski}}, \bibinfo {author} {\bibfnamefont
			{R.}~\bibnamefont {Samajdar}}, \bibinfo {author} {\bibfnamefont
			{A.}~\bibnamefont {Omran}}, \bibinfo {author} {\bibfnamefont
			{S.}~\bibnamefont {Sachdev}}, \bibinfo {author} {\bibfnamefont
			{A.}~\bibnamefont {Vishwanath}}, \bibinfo {author} {\bibfnamefont
			{M.}~\bibnamefont {Greiner}}, \bibinfo {author} {\bibfnamefont
			{V.}~\bibnamefont {Vuletić}},\ and\ \bibinfo {author} {\bibfnamefont
			{M.~D.}\ \bibnamefont {Lukin}},\ }\bibfield  {title} {\bibinfo {title}
		{Probing topological spin liquids on a programmable quantum simulator},\
	}\href {https://doi.org/10.1126/science.abi8794} {\bibfield  {journal}
		{\bibinfo  {journal} {Science}\ }\textbf {\bibinfo {volume} {374}},\ \bibinfo
		{pages} {1242} (\bibinfo {year} {2021})}\BibitemShut {NoStop}%
	\bibitem [{\citenamefont {Randall}\ \emph {et~al.}(2021)\citenamefont
		{Randall}, \citenamefont {Bradley}, \citenamefont {van~der Gronden},
		\citenamefont {Galicia}, \citenamefont {Abobeih}, \citenamefont {Markham},
		\citenamefont {Twitchen}, \citenamefont {Machado}, \citenamefont {Yao},\ and\
		\citenamefont {Taminiau}}]{randall2021MBL}%
	\BibitemOpen
	\bibfield  {author} {\bibinfo {author} {\bibfnamefont {J.}~\bibnamefont
			{Randall}}, \bibinfo {author} {\bibfnamefont {C.~E.}\ \bibnamefont
			{Bradley}}, \bibinfo {author} {\bibfnamefont {F.~V.}\ \bibnamefont {van~der
				Gronden}}, \bibinfo {author} {\bibfnamefont {A.}~\bibnamefont {Galicia}},
		\bibinfo {author} {\bibfnamefont {M.~H.}\ \bibnamefont {Abobeih}}, \bibinfo
		{author} {\bibfnamefont {M.}~\bibnamefont {Markham}}, \bibinfo {author}
		{\bibfnamefont {D.~J.}\ \bibnamefont {Twitchen}}, \bibinfo {author}
		{\bibfnamefont {F.}~\bibnamefont {Machado}}, \bibinfo {author} {\bibfnamefont
			{N.~Y.}\ \bibnamefont {Yao}},\ and\ \bibinfo {author} {\bibfnamefont {T.~H.}\
			\bibnamefont {Taminiau}},\ }\bibfield  {title} {\bibinfo {title}
		{Many-body–localized discrete time crystal with a programmable spin-based
			quantum simulator},\ }\href {https://doi.org/10.1126/science.abk0603}
	{\bibfield  {journal} {\bibinfo  {journal} {Science}\ }\textbf {\bibinfo
			{volume} {374}},\ \bibinfo {pages} {1474} (\bibinfo {year}
		{2021})}\BibitemShut {NoStop}%
	\bibitem [{\citenamefont {L{\'e}onard}\ \emph
		{et~al.}(2023{\natexlab{b}})\citenamefont {L{\'e}onard}, \citenamefont {Kim},
		\citenamefont {Kwan}, \citenamefont {Segura}, \citenamefont {Grusdt},
		\citenamefont {Repellin}, \citenamefont {Goldman},\ and\ \citenamefont
		{Greiner}}]{leonard2023FQH}%
	\BibitemOpen
	\bibfield  {author} {\bibinfo {author} {\bibfnamefont {J.}~\bibnamefont
			{L{\'e}onard}}, \bibinfo {author} {\bibfnamefont {S.}~\bibnamefont {Kim}},
		\bibinfo {author} {\bibfnamefont {J.}~\bibnamefont {Kwan}}, \bibinfo {author}
		{\bibfnamefont {P.}~\bibnamefont {Segura}}, \bibinfo {author} {\bibfnamefont
			{F.}~\bibnamefont {Grusdt}}, \bibinfo {author} {\bibfnamefont
			{C.}~\bibnamefont {Repellin}}, \bibinfo {author} {\bibfnamefont
			{N.}~\bibnamefont {Goldman}},\ and\ \bibinfo {author} {\bibfnamefont
			{M.}~\bibnamefont {Greiner}},\ }\bibfield  {title} {\bibinfo {title}
		{Realization of a fractional quantum hall state with ultracold atoms},\
	}\href {https://doi.org/10.1038/s41586-023-06122-4} {\bibfield  {journal}
		{\bibinfo  {journal} {Nature}\ }\textbf {\bibinfo {volume} {619}},\ \bibinfo
		{pages} {495} (\bibinfo {year} {2023}{\natexlab{b}})}\BibitemShut {NoStop}%
	\bibitem [{\citenamefont {Hart}\ \emph {et~al.}(2015)\citenamefont {Hart},
		\citenamefont {Duarte}, \citenamefont {Yang}, \citenamefont {Liu},
		\citenamefont {Paiva}, \citenamefont {Khatami}, \citenamefont {Scalettar},
		\citenamefont {Trivedi}, \citenamefont {Huse},\ and\ \citenamefont
		{Hulet}}]{hart2015AFM}%
	\BibitemOpen
	\bibfield  {author} {\bibinfo {author} {\bibfnamefont {R.~A.}\ \bibnamefont
			{Hart}}, \bibinfo {author} {\bibfnamefont {P.~M.}\ \bibnamefont {Duarte}},
		\bibinfo {author} {\bibfnamefont {T.-L.}\ \bibnamefont {Yang}}, \bibinfo
		{author} {\bibfnamefont {X.}~\bibnamefont {Liu}}, \bibinfo {author}
		{\bibfnamefont {T.}~\bibnamefont {Paiva}}, \bibinfo {author} {\bibfnamefont
			{E.}~\bibnamefont {Khatami}}, \bibinfo {author} {\bibfnamefont {R.~T.}\
			\bibnamefont {Scalettar}}, \bibinfo {author} {\bibfnamefont {N.}~\bibnamefont
			{Trivedi}}, \bibinfo {author} {\bibfnamefont {D.~A.}\ \bibnamefont {Huse}},\
		and\ \bibinfo {author} {\bibfnamefont {R.~G.}\ \bibnamefont {Hulet}},\
	}\bibfield  {title} {\bibinfo {title} {Observation of antiferromagnetic
			correlations in the hubbard model with ultracold atoms},\ }\href
	{https://doi.org/10.1038/nature14223} {\bibfield  {journal} {\bibinfo
			{journal} {Nature}\ }\textbf {\bibinfo {volume} {519}},\ \bibinfo {pages}
		{211} (\bibinfo {year} {2015})}\BibitemShut {NoStop}%
	\bibitem [{\citenamefont {Chiu}\ \emph
		{et~al.}(2019{\natexlab{a}})\citenamefont {Chiu}, \citenamefont {Ji},
		\citenamefont {Bohrdt}, \citenamefont {Xu}, \citenamefont {Knap},
		\citenamefont {Demler}, \citenamefont {Grusdt}, \citenamefont {Greiner},\
		and\ \citenamefont {Greif}}]{chiu2019strings}%
	\BibitemOpen
	\bibfield  {author} {\bibinfo {author} {\bibfnamefont {C.~S.}\ \bibnamefont
			{Chiu}}, \bibinfo {author} {\bibfnamefont {G.}~\bibnamefont {Ji}}, \bibinfo
		{author} {\bibfnamefont {A.}~\bibnamefont {Bohrdt}}, \bibinfo {author}
		{\bibfnamefont {M.}~\bibnamefont {Xu}}, \bibinfo {author} {\bibfnamefont
			{M.}~\bibnamefont {Knap}}, \bibinfo {author} {\bibfnamefont {E.}~\bibnamefont
			{Demler}}, \bibinfo {author} {\bibfnamefont {F.}~\bibnamefont {Grusdt}},
		\bibinfo {author} {\bibfnamefont {M.}~\bibnamefont {Greiner}},\ and\ \bibinfo
		{author} {\bibfnamefont {D.}~\bibnamefont {Greif}},\ }\bibfield  {title}
	{\bibinfo {title} {String patterns in the doped hubbard model},\ }\href
	{https://doi.org/10.1126/science.aav3587} {\bibfield  {journal} {\bibinfo
			{journal} {Science}\ }\textbf {\bibinfo {volume} {365}},\ \bibinfo {pages}
		{251} (\bibinfo {year} {2019}{\natexlab{a}})}\BibitemShut {NoStop}%
	\bibitem [{\citenamefont {Hartke}\ \emph {et~al.}(2020)\citenamefont {Hartke},
		\citenamefont {Oreg}, \citenamefont {Jia},\ and\ \citenamefont
		{Zwierlein}}]{hartke2020hubbard}%
	\BibitemOpen
	\bibfield  {author} {\bibinfo {author} {\bibfnamefont {T.}~\bibnamefont
			{Hartke}}, \bibinfo {author} {\bibfnamefont {B.}~\bibnamefont {Oreg}},
		\bibinfo {author} {\bibfnamefont {N.}~\bibnamefont {Jia}},\ and\ \bibinfo
		{author} {\bibfnamefont {M.}~\bibnamefont {Zwierlein}},\ }\bibfield  {title}
	{\bibinfo {title} {Doublon-hole correlations and fluctuation thermometry in a
			fermi-hubbard gas},\ }\href {https://doi.org/10.1103/PhysRevLett.125.113601}
	{\bibfield  {journal} {\bibinfo  {journal} {Phys. Rev. Lett.}\ }\textbf
		{\bibinfo {volume} {125}},\ \bibinfo {pages} {113601} (\bibinfo {year}
		{2020})}\BibitemShut {NoStop}%
	\bibitem [{\citenamefont {Ji}\ \emph {et~al.}(2021)\citenamefont {Ji},
		\citenamefont {Xu}, \citenamefont {Kendrick}, \citenamefont {Chiu},
		\citenamefont {Br\"uggenj\"urgen}, \citenamefont {Greif}, \citenamefont
		{Bohrdt}, \citenamefont {Grusdt}, \citenamefont {Demler}, \citenamefont
		{Lebrat},\ and\ \citenamefont {Greiner}}]{ji2021mobilehole}%
	\BibitemOpen
	\bibfield  {author} {\bibinfo {author} {\bibfnamefont {G.}~\bibnamefont
			{Ji}}, \bibinfo {author} {\bibfnamefont {M.}~\bibnamefont {Xu}}, \bibinfo
		{author} {\bibfnamefont {L.~H.}\ \bibnamefont {Kendrick}}, \bibinfo {author}
		{\bibfnamefont {C.~S.}\ \bibnamefont {Chiu}}, \bibinfo {author}
		{\bibfnamefont {J.~C.}\ \bibnamefont {Br\"uggenj\"urgen}}, \bibinfo {author}
		{\bibfnamefont {D.}~\bibnamefont {Greif}}, \bibinfo {author} {\bibfnamefont
			{A.}~\bibnamefont {Bohrdt}}, \bibinfo {author} {\bibfnamefont
			{F.}~\bibnamefont {Grusdt}}, \bibinfo {author} {\bibfnamefont
			{E.}~\bibnamefont {Demler}}, \bibinfo {author} {\bibfnamefont
			{M.}~\bibnamefont {Lebrat}},\ and\ \bibinfo {author} {\bibfnamefont
			{M.}~\bibnamefont {Greiner}},\ }\bibfield  {title} {\bibinfo {title}
		{Coupling a mobile hole to an antiferromagnetic spin background: Transient
			dynamics of a magnetic polaron},\ }\href
	{https://doi.org/10.1103/PhysRevX.11.021022} {\bibfield  {journal} {\bibinfo
			{journal} {Phys. Rev. X}\ }\textbf {\bibinfo {volume} {11}},\ \bibinfo
		{pages} {021022} (\bibinfo {year} {2021})}\BibitemShut {NoStop}%
	\bibitem [{\citenamefont {Martinez}\ \emph {et~al.}(2016)\citenamefont
		{Martinez}, \citenamefont {Muschik}, \citenamefont {Schindler}, \citenamefont
		{Nigg}, \citenamefont {Erhard}, \citenamefont {Heyl}, \citenamefont {Hauke},
		\citenamefont {Dalmonte}, \citenamefont {Monz}, \citenamefont {Zoller},\ and\
		\citenamefont {Blatt}}]{martinez2016gaugefield}%
	\BibitemOpen
	\bibfield  {author} {\bibinfo {author} {\bibfnamefont {E.~A.}\ \bibnamefont
			{Martinez}}, \bibinfo {author} {\bibfnamefont {C.~A.}\ \bibnamefont
			{Muschik}}, \bibinfo {author} {\bibfnamefont {P.}~\bibnamefont {Schindler}},
		\bibinfo {author} {\bibfnamefont {D.}~\bibnamefont {Nigg}}, \bibinfo {author}
		{\bibfnamefont {A.}~\bibnamefont {Erhard}}, \bibinfo {author} {\bibfnamefont
			{M.}~\bibnamefont {Heyl}}, \bibinfo {author} {\bibfnamefont {P.}~\bibnamefont
			{Hauke}}, \bibinfo {author} {\bibfnamefont {M.}~\bibnamefont {Dalmonte}},
		\bibinfo {author} {\bibfnamefont {T.}~\bibnamefont {Monz}}, \bibinfo {author}
		{\bibfnamefont {P.}~\bibnamefont {Zoller}},\ and\ \bibinfo {author}
		{\bibfnamefont {R.}~\bibnamefont {Blatt}},\ }\bibfield  {title} {\bibinfo
		{title} {Real-time dynamics of lattice gauge theories with a few-qubit
			quantum computer},\ }\href {https://doi.org/10.1038/nature18318} {\bibfield
		{journal} {\bibinfo  {journal} {Nature}\ }\textbf {\bibinfo {volume} {534}},\
		\bibinfo {pages} {516} (\bibinfo {year} {2016})}\BibitemShut {NoStop}%
	\bibitem [{\citenamefont {Yang}\ \emph {et~al.}(2020)\citenamefont {Yang},
		\citenamefont {Sun}, \citenamefont {Ott}, \citenamefont {Wang}, \citenamefont
		{Zache}, \citenamefont {Halimeh}, \citenamefont {Yuan}, \citenamefont
		{Hauke},\ and\ \citenamefont {Pan}}]{yang2020gauge}%
	\BibitemOpen
	\bibfield  {author} {\bibinfo {author} {\bibfnamefont {B.}~\bibnamefont
			{Yang}}, \bibinfo {author} {\bibfnamefont {H.}~\bibnamefont {Sun}}, \bibinfo
		{author} {\bibfnamefont {R.}~\bibnamefont {Ott}}, \bibinfo {author}
		{\bibfnamefont {H.-Y.}\ \bibnamefont {Wang}}, \bibinfo {author}
		{\bibfnamefont {T.~V.}\ \bibnamefont {Zache}}, \bibinfo {author}
		{\bibfnamefont {J.~C.}\ \bibnamefont {Halimeh}}, \bibinfo {author}
		{\bibfnamefont {Z.-S.}\ \bibnamefont {Yuan}}, \bibinfo {author}
		{\bibfnamefont {P.}~\bibnamefont {Hauke}},\ and\ \bibinfo {author}
		{\bibfnamefont {J.-W.}\ \bibnamefont {Pan}},\ }\bibfield  {title} {\bibinfo
		{title} {Observation of gauge invariance in a 71-site bose--hubbard quantum
			simulator},\ }\href {https://doi.org/10.1038/s41586-020-2910-8} {\bibfield
		{journal} {\bibinfo  {journal} {Nature}\ }\textbf {\bibinfo {volume} {587}},\
		\bibinfo {pages} {392} (\bibinfo {year} {2020})}\BibitemShut {NoStop}%
	\bibitem [{\citenamefont {Paulson}\ \emph {et~al.}(2021)\citenamefont
		{Paulson}, \citenamefont {Dellantonio}, \citenamefont {Haase}, \citenamefont
		{Celi}, \citenamefont {Kan}, \citenamefont {Jena}, \citenamefont {Kokail},
		\citenamefont {van Bijnen}, \citenamefont {Jansen}, \citenamefont {Zoller},\
		and\ \citenamefont {Muschik}}]{paulson20212Dgauge}%
	\BibitemOpen
	\bibfield  {author} {\bibinfo {author} {\bibfnamefont {D.}~\bibnamefont
			{Paulson}}, \bibinfo {author} {\bibfnamefont {L.}~\bibnamefont
			{Dellantonio}}, \bibinfo {author} {\bibfnamefont {J.~F.}\ \bibnamefont
			{Haase}}, \bibinfo {author} {\bibfnamefont {A.}~\bibnamefont {Celi}},
		\bibinfo {author} {\bibfnamefont {A.}~\bibnamefont {Kan}}, \bibinfo {author}
		{\bibfnamefont {A.}~\bibnamefont {Jena}}, \bibinfo {author} {\bibfnamefont
			{C.}~\bibnamefont {Kokail}}, \bibinfo {author} {\bibfnamefont
			{R.}~\bibnamefont {van Bijnen}}, \bibinfo {author} {\bibfnamefont
			{K.}~\bibnamefont {Jansen}}, \bibinfo {author} {\bibfnamefont
			{P.}~\bibnamefont {Zoller}},\ and\ \bibinfo {author} {\bibfnamefont {C.~A.}\
			\bibnamefont {Muschik}},\ }\bibfield  {title} {\bibinfo {title} {Simulating
			2d effects in lattice gauge theories on a quantum computer},\ }\href
	{https://doi.org/10.1103/PRXQuantum.2.030334} {\bibfield  {journal} {\bibinfo
			{journal} {PRX Quantum}\ }\textbf {\bibinfo {volume} {2}},\ \bibinfo {pages}
		{030334} (\bibinfo {year} {2021})}\BibitemShut {NoStop}%
	\bibitem [{\citenamefont {Meth}\ \emph {et~al.}(2023)\citenamefont {Meth},
		\citenamefont {Haase}, \citenamefont {Zhang}, \citenamefont {Edmunds},
		\citenamefont {Postler}, \citenamefont {Steiner}, \citenamefont {Jena},
		\citenamefont {Dellantonio}, \citenamefont {Blatt}, \citenamefont {Zoller}
		\emph {et~al.}}]{meth2023simulating}%
	\BibitemOpen
	\bibfield  {author} {\bibinfo {author} {\bibfnamefont {M.}~\bibnamefont
			{Meth}}, \bibinfo {author} {\bibfnamefont {J.~F.}\ \bibnamefont {Haase}},
		\bibinfo {author} {\bibfnamefont {J.}~\bibnamefont {Zhang}}, \bibinfo
		{author} {\bibfnamefont {C.}~\bibnamefont {Edmunds}}, \bibinfo {author}
		{\bibfnamefont {L.}~\bibnamefont {Postler}}, \bibinfo {author} {\bibfnamefont
			{A.}~\bibnamefont {Steiner}}, \bibinfo {author} {\bibfnamefont {A.~J.}\
			\bibnamefont {Jena}}, \bibinfo {author} {\bibfnamefont {L.}~\bibnamefont
			{Dellantonio}}, \bibinfo {author} {\bibfnamefont {R.}~\bibnamefont {Blatt}},
		\bibinfo {author} {\bibfnamefont {P.}~\bibnamefont {Zoller}}, \emph
		{et~al.},\ }\bibfield  {title} {\bibinfo {title} {Simulating 2d lattice gauge
			theories on a qudit quantum computer},\ }\href
	{https://arxiv.org/abs/2310.12110} {\bibfield  {journal} {\bibinfo  {journal}
			{arXiv preprint arXiv:2310.12110}\ } (\bibinfo {year} {2023})}\BibitemShut
	{NoStop}%
	\bibitem [{\citenamefont {Huang}\ \emph {et~al.}(2020)\citenamefont {Huang},
		\citenamefont {Kueng},\ and\ \citenamefont {Preskill}}]{Huang2020predicting}%
	\BibitemOpen
	\bibfield  {author} {\bibinfo {author} {\bibfnamefont {H.-Y.}\ \bibnamefont
			{Huang}}, \bibinfo {author} {\bibfnamefont {R.}~\bibnamefont {Kueng}},\ and\
		\bibinfo {author} {\bibfnamefont {J.}~\bibnamefont {Preskill}},\ }\bibfield
	{title} {\bibinfo {title} {Predicting many properties of a quantum system
			from very few measurements},\ }\href
	{https://doi.org/10.1038/s41567-020-0932-7} {\bibfield  {journal} {\bibinfo
			{journal} {Nature Physics}\ }\textbf {\bibinfo {volume} {16}},\ \bibinfo
		{pages} {1050} (\bibinfo {year} {2020})}\BibitemShut {NoStop}%
	\bibitem [{\citenamefont {Elben}\ \emph {et~al.}(2023)\citenamefont {Elben},
		\citenamefont {Flammia}, \citenamefont {Huang}, \citenamefont {Kueng},
		\citenamefont {Preskill}, \citenamefont {Vermersch},\ and\ \citenamefont
		{Zoller}}]{Elben2023toolbox}%
	\BibitemOpen
	\bibfield  {author} {\bibinfo {author} {\bibfnamefont {A.}~\bibnamefont
			{Elben}}, \bibinfo {author} {\bibfnamefont {S.~T.}\ \bibnamefont {Flammia}},
		\bibinfo {author} {\bibfnamefont {H.-Y.}\ \bibnamefont {Huang}}, \bibinfo
		{author} {\bibfnamefont {R.}~\bibnamefont {Kueng}}, \bibinfo {author}
		{\bibfnamefont {J.}~\bibnamefont {Preskill}}, \bibinfo {author}
		{\bibfnamefont {B.}~\bibnamefont {Vermersch}},\ and\ \bibinfo {author}
		{\bibfnamefont {P.}~\bibnamefont {Zoller}},\ }\bibfield  {title} {\bibinfo
		{title} {The randomized measurement toolbox},\ }\href
	{https://doi.org/10.1038/s42254-022-00535-2} {\bibfield  {journal} {\bibinfo
			{journal} {Nature Reviews Physics}\ }\textbf {\bibinfo {volume} {5}},\
		\bibinfo {pages} {9} (\bibinfo {year} {2023})}\BibitemShut {NoStop}%
	\bibitem [{\citenamefont {Elben}\ \emph {et~al.}(2019)\citenamefont {Elben},
		\citenamefont {Vermersch}, \citenamefont {Roos},\ and\ \citenamefont
		{Zoller}}]{Elben2019toolbox}%
	\BibitemOpen
	\bibfield  {author} {\bibinfo {author} {\bibfnamefont {A.}~\bibnamefont
			{Elben}}, \bibinfo {author} {\bibfnamefont {B.}~\bibnamefont {Vermersch}},
		\bibinfo {author} {\bibfnamefont {C.~F.}\ \bibnamefont {Roos}},\ and\
		\bibinfo {author} {\bibfnamefont {P.}~\bibnamefont {Zoller}},\ }\bibfield
	{title} {\bibinfo {title} {Statistical correlations between locally
			randomized measurements: A toolbox for probing entanglement in many-body
			quantum states},\ }\href {https://doi.org/10.1103/PhysRevA.99.052323}
	{\bibfield  {journal} {\bibinfo  {journal} {Phys. Rev. A}\ }\textbf {\bibinfo
			{volume} {99}},\ \bibinfo {pages} {052323} (\bibinfo {year}
		{2019})}\BibitemShut {NoStop}%
	\bibitem [{\citenamefont {Elben}\ \emph {et~al.}(2018)\citenamefont {Elben},
		\citenamefont {Vermersch}, \citenamefont {Dalmonte}, \citenamefont {Cirac},\
		and\ \citenamefont {Zoller}}]{elben2018quench}%
	\BibitemOpen
	\bibfield  {author} {\bibinfo {author} {\bibfnamefont {A.}~\bibnamefont
			{Elben}}, \bibinfo {author} {\bibfnamefont {B.}~\bibnamefont {Vermersch}},
		\bibinfo {author} {\bibfnamefont {M.}~\bibnamefont {Dalmonte}}, \bibinfo
		{author} {\bibfnamefont {J.~I.}\ \bibnamefont {Cirac}},\ and\ \bibinfo
		{author} {\bibfnamefont {P.}~\bibnamefont {Zoller}},\ }\bibfield  {title}
	{\bibinfo {title} {R\'enyi entropies from random quenches in atomic hubbard
			and spin models},\ }\href {https://doi.org/10.1103/PhysRevLett.120.050406}
	{\bibfield  {journal} {\bibinfo  {journal} {Phys. Rev. Lett.}\ }\textbf
		{\bibinfo {volume} {120}},\ \bibinfo {pages} {050406} (\bibinfo {year}
		{2018})}\BibitemShut {NoStop}%
	\bibitem [{\citenamefont {Brydges}\ \emph {et~al.}(2019)\citenamefont
		{Brydges}, \citenamefont {Elben}, \citenamefont {Jurcevic}, \citenamefont
		{Vermersch}, \citenamefont {Maier}, \citenamefont {Lanyon}, \citenamefont
		{Zoller}, \citenamefont {Blatt},\ and\ \citenamefont
		{Roos}}]{brydges2019entropy}%
	\BibitemOpen
	\bibfield  {author} {\bibinfo {author} {\bibfnamefont {T.}~\bibnamefont
			{Brydges}}, \bibinfo {author} {\bibfnamefont {A.}~\bibnamefont {Elben}},
		\bibinfo {author} {\bibfnamefont {P.}~\bibnamefont {Jurcevic}}, \bibinfo
		{author} {\bibfnamefont {B.}~\bibnamefont {Vermersch}}, \bibinfo {author}
		{\bibfnamefont {C.}~\bibnamefont {Maier}}, \bibinfo {author} {\bibfnamefont
			{B.~P.}\ \bibnamefont {Lanyon}}, \bibinfo {author} {\bibfnamefont
			{P.}~\bibnamefont {Zoller}}, \bibinfo {author} {\bibfnamefont
			{R.}~\bibnamefont {Blatt}},\ and\ \bibinfo {author} {\bibfnamefont {C.~F.}\
			\bibnamefont {Roos}},\ }\bibfield  {title} {\bibinfo {title} {Probing rényi
			entanglement entropy via randomized measurements},\ }\href
	{https://doi.org/10.1126/science.aau4963} {\bibfield  {journal} {\bibinfo
			{journal} {Science}\ }\textbf {\bibinfo {volume} {364}},\ \bibinfo {pages}
		{260} (\bibinfo {year} {2019})}\BibitemShut {NoStop}%
	\bibitem [{\citenamefont {Chen}\ \emph {et~al.}(2021)\citenamefont {Chen},
		\citenamefont {Yu}, \citenamefont {Zeng},\ and\ \citenamefont
		{Flammia}}]{chen2021robust}%
	\BibitemOpen
	\bibfield  {author} {\bibinfo {author} {\bibfnamefont {S.}~\bibnamefont
			{Chen}}, \bibinfo {author} {\bibfnamefont {W.}~\bibnamefont {Yu}}, \bibinfo
		{author} {\bibfnamefont {P.}~\bibnamefont {Zeng}},\ and\ \bibinfo {author}
		{\bibfnamefont {S.~T.}\ \bibnamefont {Flammia}},\ }\bibfield  {title}
	{\bibinfo {title} {Robust shadow estimation},\ }\href
	{https://doi.org/10.1103/PRXQuantum.2.030348} {\bibfield  {journal} {\bibinfo
			{journal} {PRX Quantum}\ }\textbf {\bibinfo {volume} {2}},\ \bibinfo {pages}
		{030348} (\bibinfo {year} {2021})}\BibitemShut {NoStop}%
	\bibitem [{\citenamefont {Denzler}\ \emph {et~al.}(2023)\citenamefont
		{Denzler}, \citenamefont {Mele}, \citenamefont {Derbyshire}, \citenamefont
		{Guaita},\ and\ \citenamefont {Eisert}}]{denzler2023learning}%
	\BibitemOpen
	\bibfield  {author} {\bibinfo {author} {\bibfnamefont {J.}~\bibnamefont
			{Denzler}}, \bibinfo {author} {\bibfnamefont {A.~A.}\ \bibnamefont {Mele}},
		\bibinfo {author} {\bibfnamefont {E.}~\bibnamefont {Derbyshire}}, \bibinfo
		{author} {\bibfnamefont {T.}~\bibnamefont {Guaita}},\ and\ \bibinfo {author}
		{\bibfnamefont {J.}~\bibnamefont {Eisert}},\ }\bibfield  {title} {\bibinfo
		{title} {Learning fermionic correlations by evolving with random
			translationally invariant hamiltonians},\ }\href
	{https://arxiv.org/abs/2309.12933} {\bibfield  {journal} {\bibinfo  {journal}
			{arXiv preprint arXiv:2309.12933}\ } (\bibinfo {year} {2023})}\BibitemShut
	{NoStop}%
	\bibitem [{\citenamefont {Van~Kirk}\ \emph {et~al.}(2022)\citenamefont
		{Van~Kirk}, \citenamefont {Cotler}, \citenamefont {Huang},\ and\
		\citenamefont {Lukin}}]{van2022hardware}%
	\BibitemOpen
	\bibfield  {author} {\bibinfo {author} {\bibfnamefont {K.}~\bibnamefont
			{Van~Kirk}}, \bibinfo {author} {\bibfnamefont {J.}~\bibnamefont {Cotler}},
		\bibinfo {author} {\bibfnamefont {H.-Y.}\ \bibnamefont {Huang}},\ and\
		\bibinfo {author} {\bibfnamefont {M.~D.}\ \bibnamefont {Lukin}},\ }\bibfield
	{title} {\bibinfo {title} {Hardware-efficient learning of quantum many-body
			states},\ }\href {https://arxiv.org/abs/2212.06084} {\bibfield  {journal}
		{\bibinfo  {journal} {arXiv preprint arXiv:2212.06084}\ } (\bibinfo {year}
		{2022})}\BibitemShut {NoStop}%
	\bibitem [{\citenamefont {Hu}\ \emph {et~al.}(2023)\citenamefont {Hu},
		\citenamefont {Choi},\ and\ \citenamefont {You}}]{hu2023locallybias}%
	\BibitemOpen
	\bibfield  {author} {\bibinfo {author} {\bibfnamefont {H.-Y.}\ \bibnamefont
			{Hu}}, \bibinfo {author} {\bibfnamefont {S.}~\bibnamefont {Choi}},\ and\
		\bibinfo {author} {\bibfnamefont {Y.-Z.}\ \bibnamefont {You}},\ }\bibfield
	{title} {\bibinfo {title} {Classical shadow tomography with locally scrambled
			quantum dynamics},\ }\href {https://doi.org/10.1103/PhysRevResearch.5.023027}
	{\bibfield  {journal} {\bibinfo  {journal} {Phys. Rev. Res.}\ }\textbf
		{\bibinfo {volume} {5}},\ \bibinfo {pages} {023027} (\bibinfo {year}
		{2023})}\BibitemShut {NoStop}%
	\bibitem [{\citenamefont {Akhtar}\ \emph {et~al.}(2023)\citenamefont {Akhtar},
		\citenamefont {Hu},\ and\ \citenamefont {You}}]{Akhtar2023scalableflexible}%
	\BibitemOpen
	\bibfield  {author} {\bibinfo {author} {\bibfnamefont {A.~A.}\ \bibnamefont
			{Akhtar}}, \bibinfo {author} {\bibfnamefont {H.-Y.}\ \bibnamefont {Hu}},\
		and\ \bibinfo {author} {\bibfnamefont {Y.-Z.}\ \bibnamefont {You}},\
	}\bibfield  {title} {\bibinfo {title} {Scalable and {F}lexible {C}lassical
			{S}hadow {T}omography with {T}ensor {N}etworks},\ }\href
	{https://doi.org/10.22331/q-2023-06-01-1026} {\bibfield  {journal} {\bibinfo
			{journal} {{Quantum}}\ }\textbf {\bibinfo {volume} {7}},\ \bibinfo {pages}
		{1026} (\bibinfo {year} {2023})}\BibitemShut {NoStop}%
	\bibitem [{\citenamefont {Bertoni}\ \emph {et~al.}(2022)\citenamefont
		{Bertoni}, \citenamefont {Haferkamp}, \citenamefont {Hinsche}, \citenamefont
		{Ioannou}, \citenamefont {Eisert},\ and\ \citenamefont
		{Pashayan}}]{bertoni2022shallow}%
	\BibitemOpen
	\bibfield  {author} {\bibinfo {author} {\bibfnamefont {C.}~\bibnamefont
			{Bertoni}}, \bibinfo {author} {\bibfnamefont {J.}~\bibnamefont {Haferkamp}},
		\bibinfo {author} {\bibfnamefont {M.}~\bibnamefont {Hinsche}}, \bibinfo
		{author} {\bibfnamefont {M.}~\bibnamefont {Ioannou}}, \bibinfo {author}
		{\bibfnamefont {J.}~\bibnamefont {Eisert}},\ and\ \bibinfo {author}
		{\bibfnamefont {H.}~\bibnamefont {Pashayan}},\ }\bibfield  {title} {\bibinfo
		{title} {Shallow shadows: Expectation estimation using low-depth random
			clifford circuits},\ }\href {https://arxiv.org/abs/2209.12924} {\bibfield
		{journal} {\bibinfo  {journal} {arXiv preprint arXiv:2209.12924}\ } (\bibinfo
		{year} {2022})}\BibitemShut {NoStop}%
	\bibitem [{\citenamefont {Hu}\ and\ \citenamefont
		{You}(2022)}]{hu2022hamiltonian}%
	\BibitemOpen
	\bibfield  {author} {\bibinfo {author} {\bibfnamefont {H.-Y.}\ \bibnamefont
			{Hu}}\ and\ \bibinfo {author} {\bibfnamefont {Y.-Z.}\ \bibnamefont {You}},\
	}\bibfield  {title} {\bibinfo {title} {Hamiltonian-driven shadow tomography
			of quantum states},\ }\href
	{https://doi.org/10.1103/PhysRevResearch.4.013054} {\bibfield  {journal}
		{\bibinfo  {journal} {Phys. Rev. Res.}\ }\textbf {\bibinfo {volume} {4}},\
		\bibinfo {pages} {013054} (\bibinfo {year} {2022})}\BibitemShut {NoStop}%
	\bibitem [{\citenamefont {Tran}\ \emph {et~al.}(2023)\citenamefont {Tran},
		\citenamefont {Mark}, \citenamefont {Ho},\ and\ \citenamefont
		{Choi}}]{tran2023arbitrary}%
	\BibitemOpen
	\bibfield  {author} {\bibinfo {author} {\bibfnamefont {M.~C.}\ \bibnamefont
			{Tran}}, \bibinfo {author} {\bibfnamefont {D.~K.}\ \bibnamefont {Mark}},
		\bibinfo {author} {\bibfnamefont {W.~W.}\ \bibnamefont {Ho}},\ and\ \bibinfo
		{author} {\bibfnamefont {S.}~\bibnamefont {Choi}},\ }\bibfield  {title}
	{\bibinfo {title} {Measuring arbitrary physical properties in analog quantum
			simulation},\ }\href {https://doi.org/10.1103/PhysRevX.13.011049} {\bibfield
		{journal} {\bibinfo  {journal} {Phys. Rev. X}\ }\textbf {\bibinfo {volume}
			{13}},\ \bibinfo {pages} {011049} (\bibinfo {year} {2023})}\BibitemShut
	{NoStop}%
	\bibitem [{\citenamefont {McGinley}\ and\ \citenamefont
		{Fava}(2023)}]{mcginley2023shadow}%
	\BibitemOpen
	\bibfield  {author} {\bibinfo {author} {\bibfnamefont {M.}~\bibnamefont
			{McGinley}}\ and\ \bibinfo {author} {\bibfnamefont {M.}~\bibnamefont
			{Fava}},\ }\bibfield  {title} {\bibinfo {title} {Shadow tomography from
			emergent state designs in analog quantum simulators},\ }\href
	{https://doi.org/10.1103/PhysRevLett.131.160601} {\bibfield  {journal}
		{\bibinfo  {journal} {Phys. Rev. Lett.}\ }\textbf {\bibinfo {volume} {131}},\
		\bibinfo {pages} {160601} (\bibinfo {year} {2023})}\BibitemShut {NoStop}%
	\bibitem [{spe()}]{special_cases}%
	\BibitemOpen
	\href@noop {} {}\bibinfo {note} {When directions of Hamiltonian and
		measurement basis are parallel or orthorgonal, we cannot uniquely determine
		$\rho$ from the trajectories of $\langle Z(t)\rangle$ and the quench
		Hamiltonian.}\BibitemShut {Stop}%
	\bibitem [{RDU()}]{RDU}%
	\BibitemOpen
	\href@noop {} {}\bibinfo {note}
	{$\Lambda_t=\sum_{j}e^{i\theta_j}\ketbra{j}{j}$ is a random diagonal unitary
		if each $\theta_j$ is sampled independently and uniformly from
		$[0,2\pi)$.}\BibitemShut {Stop}%
	\bibitem [{\citenamefont {Nakata}\ and\ \citenamefont
		{Murao}(2013)}]{nakata2013diagonal}%
	\BibitemOpen
	\bibfield  {author} {\bibinfo {author} {\bibfnamefont {Y.}~\bibnamefont
			{Nakata}}\ and\ \bibinfo {author} {\bibfnamefont {M.}~\bibnamefont {Murao}},\
	}\bibfield  {title} {\bibinfo {title} {Diagonal-unitary 2-design and their
			implementations by quantum circuits},\ }\href
	{https://www.worldscientific.com/doi/abs/10.1142/S0219749913500627?casa_token=L4PcnI2rx0gAAAAA:MC_BxnGH_fI2t0rDx3WzBLPDJpCbLt8Z96I0zWYY1wlRWdF_l7QXQgs6dEVKxWisrGjn3d_cjeWG}
	{\bibfield  {journal} {\bibinfo  {journal} {International Journal of Quantum
				Information}\ }\textbf {\bibinfo {volume} {11}},\ \bibinfo {pages} {1350062}
		(\bibinfo {year} {2013})}\BibitemShut {NoStop}%
	\bibitem [{tom()}]{tomography_incomplete}%
	\BibitemOpen
	\href@noop {} {}\bibinfo {note} {If the Hamiltonian shadow is not
		tomography-complete, the post-processing matrix must satisfy
		$\mathrm{det}(X_H)=0$ or $(X_H)_{i,j}=0$ for $i\neq j$. These conditions,
		which represent equality constraints, are generally not met for a typical
		Hamiltonian.}\BibitemShut {Stop}%
	\bibitem [{\citenamefont {Reimann}(2008)}]{Reimann2008resonance}%
	\BibitemOpen
	\bibfield  {author} {\bibinfo {author} {\bibfnamefont {P.}~\bibnamefont
			{Reimann}},\ }\bibfield  {title} {\bibinfo {title} {Foundation of statistical
			mechanics under experimentally realistic conditions},\ }\href
	{https://doi.org/10.1103/PhysRevLett.101.190403} {\bibfield  {journal}
		{\bibinfo  {journal} {Phys. Rev. Lett.}\ }\textbf {\bibinfo {volume} {101}},\
		\bibinfo {pages} {190403} (\bibinfo {year} {2008})}\BibitemShut {NoStop}%
	\bibitem [{deg()}]{degeneracy}%
	\BibitemOpen
	\href@noop {} {}\bibinfo {note} {Specifically, suppose requirements for
		eigenstates are all satisfied, state learning with a single Hamiltonian is
		tomography-complete when Hamiltonian has no degeneracy. Moreover, if the
		Hamiltonian has no degeneracy while the non-resonance condition does not
		satisfy, the Hamiltonian shadow map cannot be described solely by
		$X_H$.}\BibitemShut {Stop}%
	\bibitem [{\citenamefont {Ueda}(2020)}]{Ueda2020ETH}%
	\BibitemOpen
	\bibfield  {author} {\bibinfo {author} {\bibfnamefont {M.}~\bibnamefont
			{Ueda}},\ }\bibfield  {title} {\bibinfo {title} {Quantum equilibration,
			thermalization and prethermalization in ultracold atoms},\ }\href
	{https://doi.org/10.1038/s42254-020-0237-x} {\bibfield  {journal} {\bibinfo
			{journal} {Nature Reviews Physics}\ }\textbf {\bibinfo {volume} {2}},\
		\bibinfo {pages} {669} (\bibinfo {year} {2020})}\BibitemShut {NoStop}%
	\bibitem [{med()}]{median}%
	\BibitemOpen
	\href@noop {} {}\bibinfo {note} {To alleviate the data fluctuation, the value
		of every point of Hamiltonian shadow and $f(O,V_H)$ is calculated by taking
		the median over ten random $P$.}\BibitemShut {Stop}%
	\bibitem [{\citenamefont {Wang}\ and\ \citenamefont {Cui}(2023)}]{wang2023MUB}%
	\BibitemOpen
	\bibfield  {author} {\bibinfo {author} {\bibfnamefont {Y.}~\bibnamefont
			{Wang}}\ and\ \bibinfo {author} {\bibfnamefont {W.}~\bibnamefont {Cui}},\
	}\bibfield  {title} {\bibinfo {title} {classical shadow tomography with
			mutually unbiased bases},\ }\href {https://arxiv.org/abs/2310.09644}
	{\bibfield  {journal} {\bibinfo  {journal} {arXiv preprint arXiv:2310.09644}\
		} (\bibinfo {year} {2023})}\BibitemShut {NoStop}%
	\bibitem [{\citenamefont {Qingyue~Zhang}(2023)}]{zhou2023MUB}%
	\BibitemOpen
	\bibfield  {author} {\bibinfo {author} {\bibfnamefont {Y.~Z.}\ \bibnamefont
			{Qingyue~Zhang}, \bibfnamefont {Qing~Liu}},\ }\bibfield  {title} {\bibinfo
		{title} {Minimal clifford shadow estimation by mutually unbiased bases},\
	}\href {https://arxiv.org/abs/2310.18749} {\bibfield  {journal} {\bibinfo
			{journal} {arXiv preprint arXiv:2310.18749}\ } (\bibinfo {year}
		{2023})}\BibitemShut {NoStop}%
	\bibitem [{Note()}]{Note}%
	\BibitemOpen
	\bibinfo {note} {All parameters of the Rydberg Hamiltonian are publically
		available from the website of \protect \href
		{https://queracomputing.github.io/Bloqade.jl/dev/hamiltonians/}{Bloqade}}\BibitemShut
	{NoStop}%
	\bibitem [{\citenamefont {Bluvstein}\ \emph {et~al.}(2022)\citenamefont
		{Bluvstein}, \citenamefont {Levine}, \citenamefont {Semeghini}, \citenamefont
		{Wang}, \citenamefont {Ebadi}, \citenamefont {Kalinowski}, \citenamefont
		{Keesling}, \citenamefont {Maskara}, \citenamefont {Pichler}, \citenamefont
		{Greiner}, \citenamefont {Vuleti{\'c}},\ and\ \citenamefont
		{Lukin}}]{Bluvstein2022moving}%
	\BibitemOpen
	\bibfield  {author} {\bibinfo {author} {\bibfnamefont {D.}~\bibnamefont
			{Bluvstein}}, \bibinfo {author} {\bibfnamefont {H.}~\bibnamefont {Levine}},
		\bibinfo {author} {\bibfnamefont {G.}~\bibnamefont {Semeghini}}, \bibinfo
		{author} {\bibfnamefont {T.~T.}\ \bibnamefont {Wang}}, \bibinfo {author}
		{\bibfnamefont {S.}~\bibnamefont {Ebadi}}, \bibinfo {author} {\bibfnamefont
			{M.}~\bibnamefont {Kalinowski}}, \bibinfo {author} {\bibfnamefont
			{A.}~\bibnamefont {Keesling}}, \bibinfo {author} {\bibfnamefont
			{N.}~\bibnamefont {Maskara}}, \bibinfo {author} {\bibfnamefont
			{H.}~\bibnamefont {Pichler}}, \bibinfo {author} {\bibfnamefont
			{M.}~\bibnamefont {Greiner}}, \bibinfo {author} {\bibfnamefont
			{V.}~\bibnamefont {Vuleti{\'c}}},\ and\ \bibinfo {author} {\bibfnamefont
			{M.~D.}\ \bibnamefont {Lukin}},\ }\bibfield  {title} {\bibinfo {title} {A
			quantum processor based on coherent transport of entangled atom arrays},\
	}\href {https://doi.org/10.1038/s41586-022-04592-6} {\bibfield  {journal}
		{\bibinfo  {journal} {Nature}\ }\textbf {\bibinfo {volume} {604}},\ \bibinfo
		{pages} {451} (\bibinfo {year} {2022})}\BibitemShut {NoStop}%
	\bibitem [{\citenamefont {Omran}\ \emph {et~al.}(2019)\citenamefont {Omran},
		\citenamefont {Levine}, \citenamefont {Keesling}, \citenamefont {Semeghini},
		\citenamefont {Wang}, \citenamefont {Ebadi}, \citenamefont {Bernien},
		\citenamefont {Zibrov}, \citenamefont {Pichler}, \citenamefont {Choi},
		\citenamefont {Cui}, \citenamefont {Rossignolo}, \citenamefont {Rembold},
		\citenamefont {Montangero}, \citenamefont {Calarco}, \citenamefont {Endres},
		\citenamefont {Greiner}, \citenamefont {Vuletić},\ and\ \citenamefont
		{Lukin}}]{omran2019catstate}%
	\BibitemOpen
	\bibfield  {author} {\bibinfo {author} {\bibfnamefont {A.}~\bibnamefont
			{Omran}}, \bibinfo {author} {\bibfnamefont {H.}~\bibnamefont {Levine}},
		\bibinfo {author} {\bibfnamefont {A.}~\bibnamefont {Keesling}}, \bibinfo
		{author} {\bibfnamefont {G.}~\bibnamefont {Semeghini}}, \bibinfo {author}
		{\bibfnamefont {T.~T.}\ \bibnamefont {Wang}}, \bibinfo {author}
		{\bibfnamefont {S.}~\bibnamefont {Ebadi}}, \bibinfo {author} {\bibfnamefont
			{H.}~\bibnamefont {Bernien}}, \bibinfo {author} {\bibfnamefont {A.~S.}\
			\bibnamefont {Zibrov}}, \bibinfo {author} {\bibfnamefont {H.}~\bibnamefont
			{Pichler}}, \bibinfo {author} {\bibfnamefont {S.}~\bibnamefont {Choi}},
		\bibinfo {author} {\bibfnamefont {J.}~\bibnamefont {Cui}}, \bibinfo {author}
		{\bibfnamefont {M.}~\bibnamefont {Rossignolo}}, \bibinfo {author}
		{\bibfnamefont {P.}~\bibnamefont {Rembold}}, \bibinfo {author} {\bibfnamefont
			{S.}~\bibnamefont {Montangero}}, \bibinfo {author} {\bibfnamefont
			{T.}~\bibnamefont {Calarco}}, \bibinfo {author} {\bibfnamefont
			{M.}~\bibnamefont {Endres}}, \bibinfo {author} {\bibfnamefont
			{M.}~\bibnamefont {Greiner}}, \bibinfo {author} {\bibfnamefont
			{V.}~\bibnamefont {Vuletić}},\ and\ \bibinfo {author} {\bibfnamefont
			{M.~D.}\ \bibnamefont {Lukin}},\ }\bibfield  {title} {\bibinfo {title}
		{Generation and manipulation of schrödinger cat states in rydberg atom
			arrays},\ }\href {https://doi.org/10.1126/science.aax9743} {\bibfield
		{journal} {\bibinfo  {journal} {Science}\ }\textbf {\bibinfo {volume}
			{365}},\ \bibinfo {pages} {570} (\bibinfo {year} {2019})}\BibitemShut
	{NoStop}%
	\bibitem [{\citenamefont {Verresen}\ \emph {et~al.}(2017)\citenamefont
		{Verresen}, \citenamefont {Moessner},\ and\ \citenamefont
		{Pollmann}}]{verresen2017spt}%
	\BibitemOpen
	\bibfield  {author} {\bibinfo {author} {\bibfnamefont {R.}~\bibnamefont
			{Verresen}}, \bibinfo {author} {\bibfnamefont {R.}~\bibnamefont {Moessner}},\
		and\ \bibinfo {author} {\bibfnamefont {F.}~\bibnamefont {Pollmann}},\
	}\bibfield  {title} {\bibinfo {title} {One-dimensional symmetry protected
			topological phases and their transitions},\ }\href
	{https://doi.org/10.1103/PhysRevB.96.165124} {\bibfield  {journal} {\bibinfo
			{journal} {Phys. Rev. B}\ }\textbf {\bibinfo {volume} {96}},\ \bibinfo
		{pages} {165124} (\bibinfo {year} {2017})}\BibitemShut {NoStop}%
	\bibitem [{\citenamefont {Mazurenko}\ \emph {et~al.}(2017)\citenamefont
		{Mazurenko}, \citenamefont {Chiu}, \citenamefont {Ji}, \citenamefont
		{Parsons}, \citenamefont {Kan{\'a}sz-Nagy}, \citenamefont {Schmidt},
		\citenamefont {Grusdt}, \citenamefont {Demler}, \citenamefont {Greif},\ and\
		\citenamefont {Greiner}}]{Anton2017FH}%
	\BibitemOpen
	\bibfield  {author} {\bibinfo {author} {\bibfnamefont {A.}~\bibnamefont
			{Mazurenko}}, \bibinfo {author} {\bibfnamefont {C.~S.}\ \bibnamefont {Chiu}},
		\bibinfo {author} {\bibfnamefont {G.}~\bibnamefont {Ji}}, \bibinfo {author}
		{\bibfnamefont {M.~F.}\ \bibnamefont {Parsons}}, \bibinfo {author}
		{\bibfnamefont {M.}~\bibnamefont {Kan{\'a}sz-Nagy}}, \bibinfo {author}
		{\bibfnamefont {R.}~\bibnamefont {Schmidt}}, \bibinfo {author} {\bibfnamefont
			{F.}~\bibnamefont {Grusdt}}, \bibinfo {author} {\bibfnamefont
			{E.}~\bibnamefont {Demler}}, \bibinfo {author} {\bibfnamefont
			{D.}~\bibnamefont {Greif}},\ and\ \bibinfo {author} {\bibfnamefont
			{M.}~\bibnamefont {Greiner}},\ }\bibfield  {title} {\bibinfo {title} {A
			cold-atom fermi--hubbard antiferromagnet},\ }\href
	{https://doi.org/10.1038/nature22362} {\bibfield  {journal} {\bibinfo
			{journal} {Nature}\ }\textbf {\bibinfo {volume} {545}},\ \bibinfo {pages}
		{462} (\bibinfo {year} {2017})}\BibitemShut {NoStop}%
	\bibitem [{\citenamefont {Chiu}\ \emph
		{et~al.}(2019{\natexlab{b}})\citenamefont {Chiu}, \citenamefont {Ji},
		\citenamefont {Bohrdt}, \citenamefont {Xu}, \citenamefont {Knap},
		\citenamefont {Demler}, \citenamefont {Grusdt}, \citenamefont {Greiner},\
		and\ \citenamefont {Greif}}]{christie2019string}%
	\BibitemOpen
	\bibfield  {author} {\bibinfo {author} {\bibfnamefont {C.~S.}\ \bibnamefont
			{Chiu}}, \bibinfo {author} {\bibfnamefont {G.}~\bibnamefont {Ji}}, \bibinfo
		{author} {\bibfnamefont {A.}~\bibnamefont {Bohrdt}}, \bibinfo {author}
		{\bibfnamefont {M.}~\bibnamefont {Xu}}, \bibinfo {author} {\bibfnamefont
			{M.}~\bibnamefont {Knap}}, \bibinfo {author} {\bibfnamefont {E.}~\bibnamefont
			{Demler}}, \bibinfo {author} {\bibfnamefont {F.}~\bibnamefont {Grusdt}},
		\bibinfo {author} {\bibfnamefont {M.}~\bibnamefont {Greiner}},\ and\ \bibinfo
		{author} {\bibfnamefont {D.}~\bibnamefont {Greif}},\ }\bibfield  {title}
	{\bibinfo {title} {String patterns in the doped hubbard model},\ }\href
	{https://doi.org/10.1126/science.aav3587} {\bibfield  {journal} {\bibinfo
			{journal} {Science}\ }\textbf {\bibinfo {volume} {365}},\ \bibinfo {pages}
		{251} (\bibinfo {year} {2019}{\natexlab{b}})}\BibitemShut {NoStop}%
	\bibitem [{\citenamefont {Gu}(2013)}]{gu2013moments}%
	\BibitemOpen
	\bibfield  {author} {\bibinfo {author} {\bibfnamefont {Y.}~\bibnamefont
			{Gu}},\ }\emph {\bibinfo {title} {Moments of random matrices and weingarten
			functions}},\ \href
	{https://qspace.library.queensu.ca/bitstream/handle/1974/8241/Gu_Yinzheng_201308_MSc.pdf?sequence=1}
	{Ph.D. thesis} (\bibinfo {year} {2013})\BibitemShut {NoStop}%
	\bibitem [{\citenamefont {Zhu}(2017)}]{zhu2017clifford}%
	\BibitemOpen
	\bibfield  {author} {\bibinfo {author} {\bibfnamefont {H.}~\bibnamefont
			{Zhu}},\ }\bibfield  {title} {\bibinfo {title} {Multiqubit clifford groups
			are unitary 3-designs},\ }\href {https://doi.org/10.1103/PhysRevA.96.062336}
	{\bibfield  {journal} {\bibinfo  {journal} {Phys. Rev. A}\ }\textbf {\bibinfo
			{volume} {96}},\ \bibinfo {pages} {062336} (\bibinfo {year}
		{2017})}\BibitemShut {NoStop}%
	\bibitem [{\citenamefont {Nechita}\ and\ \citenamefont
		{Singh}(2021)}]{nechita2021RDU}%
	\BibitemOpen
	\bibfield  {author} {\bibinfo {author} {\bibfnamefont {I.}~\bibnamefont
			{Nechita}}\ and\ \bibinfo {author} {\bibfnamefont {S.}~\bibnamefont
			{Singh}},\ }\bibfield  {title} {\bibinfo {title} {A graphical calculus for
			integration over random diagonal unitary matrices},\ }\href
	{https://doi.org/https://doi.org/10.1016/j.laa.2020.12.014} {\bibfield
		{journal} {\bibinfo  {journal} {Linear Algebra and its Applications}\
		}\textbf {\bibinfo {volume} {613}},\ \bibinfo {pages} {46} (\bibinfo {year}
		{2021})}\BibitemShut {NoStop}%
\end{thebibliography}
\end{document}